\documentclass[journal]{IEEEtran}
\usepackage{amsmath,amsfonts,amssymb,amsthm,mathtools,mathrsfs}
\usepackage{algorithm}
\usepackage{algpseudocode}
\usepackage{array}
\usepackage[caption=false,font=normalsize,labelfont=sf,textfont=sf]{subfig}
\usepackage{textcomp}
\usepackage{stfloats}
\usepackage{url}
\usepackage{multirow}
\usepackage{verbatim}
\usepackage{graphicx}
\usepackage{cite}
\usepackage{hyperref} 
\usepackage{geometry}
\usepackage{float}
\usepackage{enumitem}
\usepackage{tabularx, booktabs, ragged2e}
\newcolumntype{L}{>{\RaggedRight\arraybackslash}X}
\usepackage{placeins} 
\usepackage{makecell} 
\usepackage{etoolbox} 
\usepackage{chngcntr} 
\usepackage{pgfplots}
\usepgfplotslibrary{groupplots}
\pgfplotsset{compat=1.18}
\usetikzlibrary{calc}
\usepackage{tikz}
\usetikzlibrary{mindmap, shadows, shapes.geometric, positioning, calc}

\newcommand{\Dd}{\text{Dd}}
\newcommand{\Atom}{\text{Atom}}

\newcommand{\Cn}{\text{Cn}}
\newcommand{\enc}{\text{enc}}

\newcommand{\poly}{\operatorname{poly}}

\newcommand{\E}{\mathbb{E}}

\newcommand{\timeindex}{\mathsf{time}}

\newtheorem{axiom}{Axiom}[section]
\newtheorem{assumption}{Assumption}[section]
\newtheorem{theorem}{Theorem}[section]
\newtheorem{lemma}{Lemma}[section]
\newtheorem{proposition}{Proposition}[section]
\newtheorem{corollary}{Corollary}[section]
\newtheorem{definition}{Definition}[section]
\newtheorem{remark}{Remark}[section]

\numberwithin{equation}{section}

\makeatletter
\renewenvironment{proof}[1][\proofname]{\par
  \pushQED{\qed}%
  \normalfont \topsep6\p@\@plus6\p@\relax
  \trivlist
  \item[\hskip\labelsep
        \itshape
    #1\@addpunct{.}]\ignorespaces
}{%
  \popQED\endtrivlist\@endpefalse
}
\makeatother

\onecolumn
\begin{document}

\title{The Derivation Penalty in Premise-Erasure Caching: Capacity, Strong Converse, and Dispersion Dichotomy}

\author{Jianfeng~Xu$^{1}$%
\thanks{$^{1}$Koguan School of Law, China Institute for Smart Justice, School of Computer Science, Shanghai Jiao Tong University, Shanghai 200030, China. Email: xujf@sjtu.edu.cn}%
}

\maketitle

\begin{abstract}
We introduce an information-theoretic framework for caching in
derivation-based reasoning engines under independent premise
erasure.
Two decoder models are compared: a coded scheme using an arbitrary
bit-string cache with a general-purpose decoder, and a
derivation-constrained scheme where the cache consists of logical
facts and the decoder must produce a valid proof.
Four coding theorems are established.
The first proves that each derivation step carries a universal
per-step information content determined by the base size.
The second reveals an exponential capacity separation between
linear-chain and balanced-merge Datalog architectures at equal
depth.
The third identifies a critical access frequency separating the
regimes where caching and on-demand derivation are optimal.
The fourth determines the minimum derivation-constrained cache
under erasure, decomposing query information into reliable cache
and noisy channel capacity.
The central result is the derivation penalty: the ratio of the
derivation-constrained cache to the coded cache converges to the
reciprocal of the erasure rate, universally across query counts,
overlap structures, and reliability targets.
This penalty originates from a structural caching rigidity
theorem showing that only cache facts within the target query's
derivation DAG contribute to resilience, precluding
cross-coordinate error correction.
Beyond capacity, we prove a strong converse at the KL-divergence
rate with Bahadur--Rao prefactors, a dispersion dichotomy
(positive coded dispersion versus zero derivation-constrained
dispersion), and a complete eight-regime phase diagram.
The architecture-dependent depth-to-dependency mapping yields
exponentially sharper phase transitions for the merge
architecture.
All results transfer across synonymous representations.
\end{abstract}

\begin{IEEEkeywords}
Derivation depth,
premise erasure,
coded caching,
derivation penalty,
erasure channel,
strong converse,
channel dispersion,
MDS codes,
Datalog.
\end{IEEEkeywords}

\section{Introduction}
\label{sec:introduction}

\IEEEPARstart{D}{erivation-based} reasoning engines answer queries by constructing
proofs from stored premises.
Such engines arise in database query
evaluation~\cite{abiteboul1995foundations}, knowledge-base reasoning, automated
theorem proving, and retrieval-augmented architectures for large language
models~\cite{lewis2020retrieval,brown2020language}.
When the premise base is subject to stochastic loss---due to storage failures,
network partitions, or data expiration---a reliable cache of intermediate logical
facts can ensure continued derivability of target queries.
Two fundamental questions arise.
First, \emph{what is the minimum cache size for $\delta$-reliable query recovery
under i.i.d.\ premise erasure?}
Second, \emph{what penalty is incurred when the decoder must produce a valid
logical derivation rather than merely output the correct answer?}
These questions sit at the intersection of information theory, deductive reasoning,
and caching, yet no existing framework addresses them jointly.

Shannon's theory~\cite{shannon1948mathematical,cover2006elements} provides the
language for source coding and channel capacity, but does not model
\emph{inference cost}: the number of deductive steps needed to reach a conclusion
from stored premises.
Kolmogorov complexity~\cite{kolmogorov1965three,li2008introduction} quantifies
individual description length, and Bennett's logical
depth~\cite{bennett1988logical} connects it to computation time; however, neither
targets \emph{base-relative} query answering, where the available premise base is
the explicit parameter governing online cost.
In database systems, materialized view selection explicitly trades storage for
query latency~\cite{chaudhuri1995optimizing,gupta1997materialized,harinarayan1996implementing},
and caching policies implement this tradeoff
online~\cite{sleator1985amortized,belady1966study,megiddo2003arc}, but these
approaches lack information-theoretic optimality guarantees under noise.
Knowledge compilation~\cite{darwiche2002knowledge} pre-processes a knowledge base
into a tractable target language so that certain queries can be answered
efficiently, yet the compilation cost is not quantified against a noisy premise
channel.
Time--space tradeoffs are classical in complexity
theory~\cite{yao1985should,borodin2023time}, and the depth--width decomposition of
circuits~\cite{wegener1987complexity,vollmer1999introduction} provides structural
analogies, but neither framework parameterizes the tradeoff by a noisy premise
base.

On the channel-coding side, the $m$-ary erasure channel
$\mathrm{BEC}_m(\varepsilon)$ with capacity
$(1{-}\varepsilon)\log m$~\cite{cover2006elements} and MDS codes achieving the
sphere-packing exponent at all
rates~\cite[Theorem~5.8.3]{gallager1968information} are well understood.
The finite-blocklength theory of Polyanskiy, Poor, and
Verd\'{u}~\cite{polyanskiy2010channel} characterizes second-order coding rates via
channel dispersion.
The coded-caching paradigm of Maddah-Ali and
Niesen~\cite{maddah2014fundamental} demonstrates that coded multicast messages
yield multiplicative gains over uncoded placement in broadcast networks; however,
in that framework the receiver is an arbitrary algebraic decoder---not a proof
engine constrained to combine cached facts with surviving premises via valid
inference steps.
The strong converse and error
exponents~\cite{bahadur1960deviations,dembo2009large,gallager1968information},
moderate deviations~\cite{dembo2009large}, and exact asymptotics for binomial tails
provide the probabilistic toolkit, but connecting these tools to a deductive
decoder model requires new structural results.
Neuro-symbolic AI~\cite{garcez2023neurosymbolic} and large knowledge
graphs~\cite{bollacker2008freebase} further highlight the tension between
parametric storage and non-parametric retrieval or derivation, yet
information-theoretic bounds on this tension---particularly under premise
noise---have been lacking.

The existing literature does not provide a joint theory of (i)~base-relative
derivation cost, (ii)~information-theoretic cache limits under premise noise,
and (iii)~the precise price of requiring a \emph{valid derivation} rather than an
arbitrary decoder output. This paper addresses all three.

We formalize two decoder models: \emph{coded caching}
(Definition~\ref{def:coded-scheme}), where the cache is an arbitrary bit string
decoded by a general-purpose algorithm, and \emph{derivation-constrained caching}
(Definition~\ref{def:uncoded-scheme}), where the cache consists of logical facts
\(S\subseteq\Cn(B)\) and decoding must certify \(q\in\Cn(\tilde B\cup S)\).
Our key structural result is a caching rigidity theorem
(Theorem~\ref{thm:ancestral-relevance}): under faithful programs with unique
traces, only cached facts within the target query's derivation DAG can improve
resilience. This precludes cross-coordinate error correction and leads to the
universal \(1/\varepsilon\) derivation penalty.

The formal substrate fixes a logical language $\mathcal{L}$ instantiated as
$\mathrm{FO(LFP)}$~\cite{immerman1999descriptive}; on finite ordered structures,
$\mathrm{FO(LFP)}$ captures polynomial-time properties via the Immerman--Vardi
theorem~\cite{immerman1982relational,vardi1982complexity}, providing a
machine-independent semantic layer.
The two-domain information model $(S_O,S_C)$ is adopted from the objective
information viewpoint~\cite{xu2014objective,xu2024research,xu2025general,qiu2025research}.
The quantitative results are developed for two Datalog architectures: a
tuple-assembly chain $\Pi_k$ with linear dependency growth
$\kappa=k{+}d{-}1$ and a balanced-merge program $\Pi_k^{\parallel}$ with
exponential growth $\kappa=k\cdot 2^{d-1}$.

Figure~\ref{fig:roadmap} summarizes the logical architecture of the paper's
results, organized around four coding theorems (CT1--CT4) that form a logical
chain converging on the central $1/\varepsilon$ derivation penalty.

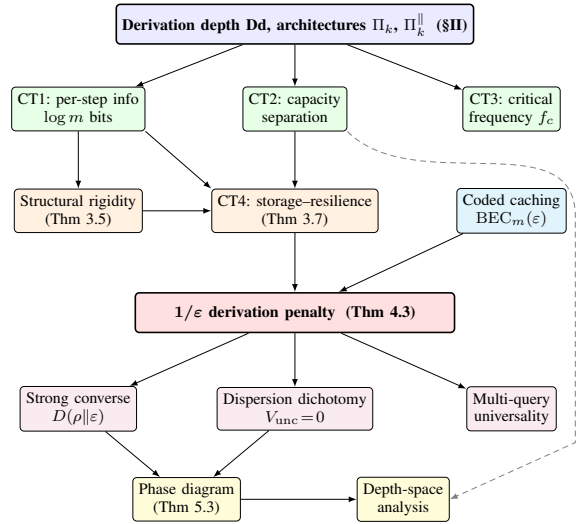
\begin{figure}[ht]
\centering
\scalebox{0.75}{%
\begin{tikzpicture}[>=latex, font=\footnotesize,
  nd/.style={draw, rounded corners=2pt, minimum height=6mm,
             align=center, inner sep=3pt},
  wd/.style={draw, rounded corners=2pt, minimum height=7mm,
             minimum width=5.6cm, align=center,
             font=\footnotesize\bfseries, thick, inner sep=3pt}]
\node[wd, fill=blue!8] (M) at (0,0)
  {Derivation depth $\Dd$, architectures
   $\Pi_k$, $\Pi_k^{\parallel}$\; (\S\ref{sec:model})};
\node[nd, fill=green!10] (C1) at (-3.8,-1.5)
  {CT1: per-step info\\$\log m$ bits};
\node[nd, fill=green!10] (C2) at (0,-1.5)
  {CT2: capacity\\separation};
\node[nd, fill=green!10] (C3) at (3.8,-1.5)
  {CT3: critical\\frequency $f_c$};
\node[nd, fill=orange!12] (R) at (-3.8,-3.3)
  {Structural rigidity\\(Thm~\ref{thm:ancestral-relevance})};
\node[nd, fill=orange!12] (C4) at (0,-3.3)
  {CT4: storage--resilience\\(Thm~\ref{thm:storage-resilience-capacity})};
\node[nd, fill=cyan!10] (CC) at (3.8,-3.3)
  {Coded caching\\$\mathrm{BEC}_m(\varepsilon)$};
\node[wd, fill=red!12] (P) at (0,-5.1)
  {$\boldsymbol{1/\varepsilon}$ \textbf{derivation penalty}\;
   (Thm~\ref{thm:source-channel-separation})};
\node[nd, fill=purple!8] (SC) at (-3.8,-6.8)
  {Strong converse\\$D(\rho\|\varepsilon)$};
\node[nd, fill=purple!8] (DI) at (0,-6.8)
  {Dispersion dichotomy\\$V_{\mathrm{unc}}\!=\!0$};
\node[nd, fill=purple!8] (MQ) at (3.8,-6.8)
  {Multi-query\\universality};
\node[nd, fill=yellow!18] (PH) at (-1.9,-8.4)
  {Phase diagram\\(Thm~\ref{thm:phase-diagram})};
\node[nd, fill=yellow!18] (DS) at (1.9,-8.4)
  {Depth-space\\analysis};
\draw[->] (M) -- (C1);
\draw[->] (M) -- (C2);
\draw[->] (M) -- (C3);
\draw[->] (C1.south) -- (R.north);
\draw[->] (C1.south east) -- (C4.north west);
\draw[->] (C2) -- (C4);
\draw[->] (R.east) -- (C4.west);
\draw[->] (C4) -- (P);
\draw[->] (CC) -- (P);
\draw[->] (P) -- (SC);
\draw[->] (P) -- (DI);
\draw[->] (P) -- (MQ);
\draw[->] (SC) -- (PH);
\draw[->] (DI) -- (PH);
\draw[->] (PH) -- (DS);
\draw[->, densely dashed, gray]
  (C2.south east) to[out=-45,in=90] ++(4.0,-1.5)
  -- ++(0,-4.0)
  -- (DS.east);
\end{tikzpicture}
}%
\caption{Logical dependency map.
Solid arrows denote primary dependencies; the dashed arrow indicates
CT2's architecture mapping $d\mapsto\kappa_{\mathcal A}(d)$ feeding
the depth-space re-parameterization of~\S\ref{sec:arch-depth}.
The $1/\varepsilon$ derivation penalty
(Theorem~\ref{thm:source-channel-separation}) arises from comparing
CT4's derivation-constrained converse---enabled by the structural
rigidity of Theorem~\ref{thm:ancestral-relevance}---with the coded
caching achievability over $\mathrm{BEC}_m(\varepsilon)$.
Representation invariance (Theorem~\ref{thm:CT-transfer}) ensures
all coding theorems transfer across synonymous representations.}
\label{fig:roadmap}
\end{figure}

The main contributions, corresponding to the nodes in
Figure~\ref{fig:roadmap}, are as follows.
CT1 (Theorem~\ref{thm:tight-datalog}) proves that each derivation step carries
exactly $\log m$ bits of conditional algorithmic information, establishing the
per-step information rate that denominates all subsequent results.
CT2 (Theorem~\ref{thm:exponential-capacity}) establishes an exponential capacity
separation between the chain and merge architectures at equal depth:
$C(\Pi_k,B,d)=(k{+}d{-}1)\log m$ versus
$C(\Pi_k^{\parallel},B,d)=k\cdot 2^{d-1}\log m$.
CT3 (Theorem~\ref{thm:freq_tradeoff}) identifies the critical access frequency
$f_c=\Theta(\rho_s\cdot\log(m{+}d))$, where $\rho_s$ is the normalized storage
cost, at which caching overtakes on-demand derivation.
CT4 (Theorem~\ref{thm:storage-resilience-capacity}) proves that under i.i.d.\
erasure at rate~$\varepsilon$ with resilience target~$\delta$, the minimum
derivation-constrained cache for a generic query with $\kappa$ distinct
dependencies is $\sigma^*_{\mathrm{unc}}=(\kappa{-}N^*)\log m+O(\kappa)$ with
$N^*\approx\delta/\varepsilon$, while the minimum coded cache is
$\sigma^*_{\mathrm{code}}=(1{+}o(1))\varepsilon\kappa\log m$
(Theorem~\ref{thm:source-channel-separation}).
The ratio $\sigma^*_{\mathrm{unc}}/\sigma^*_{\mathrm{code}}$ converges to
$1/\varepsilon$---the \emph{derivation penalty}---which is universal: it depends
only on~$\varepsilon$ and is independent of the number of queries, their overlap
structure, or~$\delta$ (Theorem~\ref{thm:multi-penalty-mq}).
Beyond first-order capacity, the success probability below the coded threshold
decays at the KL-divergence rate $D(\rho\|\varepsilon)$
(Theorem~\ref{thm:sc-exponent}) with Bahadur--Rao
prefactors~\cite{bahadur1960deviations} (Theorem~\ref{thm:bahadur-rao}), and the
two schemes exhibit a \emph{dispersion dichotomy}: coded second-order term
$\Theta(\sqrt\kappa\,\log m)$ matching the BEC
dispersion~\cite{polyanskiy2010channel} versus derivation-constrained
$O(\log m)$---zero effective dispersion
(Theorem~\ref{thm:dispersion-dichotomy}).
The complete phase diagram (Theorem~\ref{thm:phase-diagram}) comprises five coded
regimes and three derivation-constrained regimes, the latter exhibiting no phase
transition.
The architecture-dependent mapping $d\mapsto\kappa_{\mathcal A}(d)$ from CT2
propagates to every downstream quantity, yielding exponentially sharper phase
transitions for the merge architecture (Theorem~\ref{thm:arch-phase-ds}) and
exponentially smaller maximum resilient depth
(Theorem~\ref{thm:depth-resilience-duality}).
All coding theorems transfer across synonymous representations
(Theorem~\ref{thm:CT-transfer}).

The paper is organized as follows.
Section~\ref{sec:model} develops the system model: irredundant cores, derivation
depth, the two Datalog architectures, coding theorems CT1--CT3, and representation
invariance.
Section~\ref{sec:noise-resilience} introduces premise erasure and proves CT4 via
the structural rigidity theorem.
Section~\ref{sec:precise} establishes the source--channel separation, the
$1/\varepsilon$ derivation penalty, the strong converse, and the dispersion
dichotomy.
Section~\ref{sec:phase-diagram} refines the error analysis to Bahadur--Rao
precision and assembles the complete phase diagram.
Section~\ref{sec:multi-query} extends the theory to multi-query joint caching.
Section~\ref{sec:arch-depth} re-parameterizes the phase diagram in depth space.
Section~\ref{sec:numerical} provides numerical validation.
Sections~\ref{sec:discussion} and~\ref{sec:conclusion} discuss implications,
limitations, and open problems.
Appendix~\ref{sec:substrate} contains the logical substrate and noise model;
Appendix~\ref{app:sc-exponent-calc} provides a supporting calculation.


\section{System Model and Preliminaries}
\label{sec:model}

The axiomatization of information---the logical system
$\mathcal{L}:=\mathrm{FO(LFP)}$ over finite relational
structures~\cite{immerman1999descriptive}, the two-domain model
$(S_O,S_C)$ with effective encodings, the enabling mechanism,
state synonymy~$\equiv_{\mathcal{L}}$, and the noisy semantic
base---is developed in Appendix~\ref{sec:substrate}.
All semantic state sets lie in a fixed ambient
universe~$\mathbb{S}_O$
(Assumption~\ref{assump:semantic-universe});
by the Immerman--Vardi
theorem~\cite{immerman1982relational,vardi1982complexity},
$\mathcal{L}$ captures exactly the polynomial-time properties of
ordered finite structures, and
Proposition~\ref{prop:noisy-exist} ensures that the
set-perturbation noise model of
Sections~\ref{sec:noise-resilience}--\ref{sec:arch-depth} entails
no loss of generality.
We fix an effective proof system~$\mathsf{PS}$ for~$\mathcal{L}$
with decidable proof checking
(Assumption~\ref{assump:proof-system}) and define the deductive
closure of any finite formula set~$\Gamma$ as
\begin{equation}\label{eq:Cn-def}
  \Cn(\Gamma)
  \;:=\;
  \bigl\{\varphi:\Gamma\vdash_{\mathcal{L}}\varphi\bigr\}.
\end{equation}
The remainder of this section introduces irredundant cores and
derivation depth (Section~\ref{subsec:atomic-derivation}),
derivation DAGs and strata (Section~\ref{subsec:dag-strata}),
two concrete Datalog architectures
(Section~\ref{subsec:architectures}), the coding theorems
CT1--CT2 (Section~\ref{subsec:coding-theorems}), the
storage--depth and storage--computation tradeoffs with CT3
(Sections~\ref{subsec:stratified-tradeoff}--\ref{subsec:tradeoff}),
and representation invariance (Section~\ref{subsec:rep-inv}).

\subsection{Semantic Dependencies, Irredundant Cores, and
  Derivation Depth}
\label{subsec:atomic-derivation}

This subsection introduces: (i)~an irredundant semantic core
$\Atom(S_O)$ (via $\Cn$, not $P_O$); (ii)~a computable predecessor
operator $P_O$ defining $\Dd(\cdot\mid B)$.

A finite set $B\subseteq\mathbb S_O$ is an \emph{available premise
base}: its elements are depth-$0$ premises.

\begin{assumption}[Finite and effectively listable knowledge bases]
\label{assump:finite-so}
The knowledge bases $S_O$ considered are finite and effectively
listable under a fixed canonical order.
\end{assumption}

\begin{assumption}[Effective redundancy test]
\label{assump:core-extractable}
The predicate $s\in\Cn(\Gamma)$ is decidable whenever
$\Gamma\subseteq S_O$ is finite and $s\in S_O$.
This holds in Datalog/Horn-style settings and bounded-domain
theories~\cite{abiteboul1995foundations,dantsin2001complexity}.
\end{assumption}

\begin{definition}[Semantic atomic core]
\label{def:atom-so}
$\Atom(S_O)$ is the output of the deterministic procedure:
initialize $A\gets S_O$; scan in canonical order; for each
$s\in S_O$, if $s\in\Cn(A\setminus\{s\})$ then
$A\gets A\setminus\{s\}$; output $A$.
\end{definition}

\begin{proposition}[Core correctness]
\label{prop:atom-core-correct}
Under Assumptions~\ref{assump:finite-so}
and~\ref{assump:core-extractable}, $A:=\Atom(S_O)$ satisfies:
\textup{(i)}~$\Cn(A)=\Cn(S_O)$;
\textup{(ii)}~$a\notin\Cn(A\setminus\{a\})$ for all $a\in A$;
\textup{(iii)}~$A$ is uniquely determined by $S_O$ and the
canonical order.
\end{proposition}

\begin{proof}
Each removal preserves $\Cn$; by induction,
$\Cn(A)=\Cn(S_O)$.
For irredundancy: when $a$ was scanned, $a$ was retained because
$a\notin\Cn(A_t\setminus\{a\})$ where $A_t\supseteq A$;
monotonicity of $\Cn$ gives
$a\notin\Cn(A\setminus\{a\})$.
Canonicality: determinism.
\end{proof}

\begin{definition}[Intrinsic vs.\ operational premise bases]
\label{def:intrinsic-operational-bases}
$A:=\Atom(S_O)$ (core premises),
$J:=S_O\setminus A$ (stored shortcuts),
$A\cup J=S_O$ (operational base).
\end{definition}

\paragraph{Predecessor structure for depth.}

\begin{axiom}[Finite and computable predecessors]
\label{ax:finite-predecessors}
For every $s\in\mathbb S_O$, $P_O(s)$ is finite and effectively
computable from $\enc_O(s)$.
\end{axiom}

\begin{axiom}[Well-foundedness]
\label{ax:wellfounded-predecessor}
For every $s\in\mathbb S_O$, the backward unfolding of $s$ along
$P_O$ contains no infinite chain.
\end{axiom}

\begin{proposition}[Finite acyclic predecessor unfolding]
\label{prop:no-cycle-dag}
Under Axioms~\ref{ax:finite-predecessors}
and~\ref{ax:wellfounded-predecessor}, for any $s\in\mathbb S_O$
the backward unfolding is a finite computable
DAG~\cite{kunen2014set}.
\end{proposition}

\begin{proof}
Consider the rooted predecessor unfolding graph obtained by repeatedly expanding
each node \(v\) to its finite predecessor set \(P_O(v)\) (Axiom~\ref{ax:finite-predecessors}).
This unfolding is a finitely-branching rooted tree when nodes are kept with
multiplicity, and it contains a rooted path for every backward chain.

If the unfolding were infinite, then by K\H{o}nig's lemma (every infinite finitely
branching tree has an infinite path~\cite{konig1990theory}), it would contain an infinite backward chain,
contradicting Axiom~\ref{ax:wellfounded-predecessor}.
Hence the unfolding is finite. Collapsing repeated states yields a finite
computable DAG.
\end{proof}

\begin{assumption}[Alignment of dependency and inference]
\label{assump:alignment}
Fix a finite $B\subseteq\mathbb S_O$.
For every $s\in\mathbb S_O\setminus B$, $P_O(s)$ coincides with
the immediate premises in a fixed normal-form single-step
inference concluding $s$.
Elements of $B$ need not be inferred.
\end{assumption}

\begin{definition}[Base-relative derivation depth]
\label{def:derivation-depth}
Let $B\subseteq\mathbb S_O$ be finite.
Under Assumption~\ref{assump:alignment}:
\begin{equation}\label{eq:Dd-inductive}
  \Dd(s\mid B)
  :=
  \begin{cases}
    0, & s\in B,\\[0.3ex]
    1+\displaystyle\max_{s'\in P_O(s)}\Dd(s'\mid B),
      & \text{otherwise},
  \end{cases}
\end{equation}
with $\max\varnothing:=0$.
\end{definition}

\begin{theorem}[Well-definedness and computability of $\Dd$]
\label{thm:computability-Dd}
Let $B\subseteq\mathbb S_O$ be finite and effectively decidable.
Under Assumption~\ref{assump:alignment} and
Axioms~\ref{ax:finite-predecessors}--\ref{ax:wellfounded-predecessor},
$\Dd(s\mid B)$ is a unique, finite, computable non-negative
integer for every $s\in\mathbb S_O$.
\end{theorem}

\begin{proof}
Proposition~\ref{prop:no-cycle-dag} gives a finite DAG;
bottom-up dynamic programming computes $\Dd$.
\end{proof}

\begin{definition}[Intrinsic and operational depths]
\label{def:int-op-depth}
For $A:=\Atom(S_O)$ and $q\in\Cn(A)$:
$n_{\mathrm{int}}(q):=\Dd(q\mid A)$,
$n_{\mathrm{op}}(q):=\Dd(q\mid S_O)$.
Since $A\subseteq S_O$, $n_{\mathrm{op}}(q)\le n_{\mathrm{int}}(q)$.
\end{definition}

\subsection{Derivation DAGs, Strata, and Base-Fact Dependencies}
\label{subsec:dag-strata}

\begin{definition}[Canonical derivation DAG]
\label{def:canonical-dag}
For $q\in\Cn(B)$, the DAG $G(q,B)$ has vertex set $V(G)$
(the smallest set containing $q$ and closed under
predecessors), edge set
$E(G):=\{(v,v'):v\notin B,\;v'\in P_O(v)\}$,
root $q$, and leaves $V_0(G):=V(G)\cap B$.
\end{definition}

\begin{definition}[Depth strata and width profile]
\label{def:strata}
For $G=G(q,B)$ with $d:=\Dd(q\mid B)$:
depth stratum
$V_\ell:=\{v\in V(G):\Dd(v\mid B)=\ell\}$;
layer width $w_\ell:=|V_\ell|$;
maximum width $W_{\max}:=\max_{1\le\ell\le d}w_\ell$;
total work $W:=\sum_{\ell=1}^{d}w_\ell$.
\end{definition}

\begin{definition}[Dependency count vs.\ head arity]
\label{def:dep-count}
Let \(q\in\Cn(B)\) and \(G(q,B)\) be the canonical derivation DAG
(Definition~\ref{def:canonical-dag}).
\begin{enumerate}[label=\textup{(\roman*)}]
\item The \emph{distinct base-fact dependency count} is
\[
  \kappa(q,B)\;:=\;|V_0(G(q,B))|\;=\;|V(G(q,B))\cap B|.
\]
This counts \emph{distinct} EDB facts in \(B\) that appear as leaves.

\item For the concrete query families of
Definitions~\ref{def:tuple-assembly}--\ref{def:balanced-merge},
each IDB query \(q\) has a fixed head tuple length (arity), denoted \(a(q)\).
For depth \(d\), write \(a(\Pi,d)\) for the arity of depth-\(d\) queries under
architecture \(\Pi\).
For \(\Pi_k\), \(a(\Pi_k,d)=k+d-1\); for \(\Pi_k^{\parallel}\),
\(a(\Pi_k^{\parallel},d)=k\cdot 2^{d-1}\).

\item Always \(\kappa(q,B)\le a(q)\). Equality holds for the
\emph{distinct-coordinate} subfamily where all head coordinates are pairwise
distinct.
\end{enumerate}
\end{definition}

\subsection{Derivation Architectures}
\label{subsec:architectures}

The paper's quantitative results are established for two
Datalog programs over a unary EDB $B:=\{A(i):i\in[m]\}$
with $m:=|B|\ge 2$.
Both programs use rule set~$R$ with $r:=|R|=2$ and maximum
arity~$k$.

\begin{definition}[Tuple-assembly chain program $\Pi_k$]
\label{def:tuple-assembly}
Fix $k\ge 2$.
IDB relations $T_d$ of arity $d\ge k$.
\emph{Base rule:}
$T_k(x_1,\ldots,x_k)\leftarrow A(x_1),\ldots,A(x_k)$.
\emph{Extension rule:}
$T_{d+1}(\vec x,x_{d+1})\leftarrow T_d(\vec x),A(x_{d+1})$.
\end{definition}

\begin{definition}[Balanced-merge program $\Pi_k^{\parallel}$]
\label{def:balanced-merge}
Fix $k\ge 2$.
IDB relations $R_d$ of arity $k\cdot 2^{d-1}$.
\emph{Base rule:}
$R_1(x_1,\ldots,x_k)\leftarrow A(x_1),\ldots,A(x_k)$.
\emph{Merge rule:}
$R_{d+1}(\vec x,\vec y)\leftarrow R_d(\vec x),R_d(\vec y)$,
with $|\vec x|=|\vec y|=k\cdot 2^{d-1}$.
\end{definition}

\begin{definition}[Faithful rule]
\label{def:faithful-rule}
A Datalog rule is \emph{faithful} if every body variable
appears in the head, so the head tuple uniquely determines
the body instantiation.
\end{definition}

\begin{lemma}[Faithfulness of $\Pi_k$]
\label{lem:faithfulness}
Every rule of\/ $\Pi_k$ is faithful.
If $T_d(i_1,\ldots,i_d)$ is derivable from~$B$, its
derivation trace is uniquely determined by the tuple
$(i_1,\ldots,i_d)$.
\end{lemma}

\begin{proof}
Each rule places all body variables in the head.
Uniqueness follows by induction on~$d$: the head tuple
determines the body instantiation
$(T_{d-1}(i_1,\ldots,i_{d-1}),\,A(i_d))$, and the trace
of $T_{d-1}(\ldots)$ is unique by hypothesis.
\end{proof}

\begin{lemma}[Faithfulness of $\Pi_k^{\parallel}$]
\label{lem:merge-faithfulness}
Every rule of $\Pi_k^{\parallel}$ is faithful.
If $R_d(i_1,\ldots,i_{k\cdot 2^{d-1}})$ is derivable
from~$B$, its derivation trace is uniquely determined by
the tuple.
\end{lemma}

\begin{proof}
Analogous to Lemma~\ref{lem:faithfulness}: the merge rule
places all body variables in the head, and uniqueness follows
by induction on~$d$ since each head tuple determines both
body instantiations.
\end{proof}

\begin{proposition}[Counting and structure for \(\Pi_k\)]
\label{prop:tuple-counting}
For \(\Pi_k\) with \(|B|=m\ge 2\) and depth \(d\ge 1\), define the depth-\(d\)
query set
\[
  \mathcal{Q}_d:=\{T_{k+d-1}(i_1,\ldots,i_{k+d-1}): (i_1,\ldots,i_{k+d-1})\in[m]^{k+d-1}\}.
\]
Then:
\begin{enumerate}[label=\textup{(\roman*)}]
\item \(|\mathcal{Q}_d|=m^{a}\) with head-arity \(a=a(\Pi_k,B,d)=k+d-1\).
\item Every derivation DAG \(G(q,B)\) is a path with \(w_\ell=1\).
\item The distinct dependency count satisfies \(\kappa(q,B)\le a\), with equality
whenever \(i_1,\ldots,i_a\) are pairwise distinct.
\end{enumerate}
\end{proposition}

\begin{proof}
A chain of $d$ steps builds $T_{k+d-1}$-tuples over
$[m]^{k+d-1}$; faithfulness
(Lemma~\ref{lem:faithfulness}) gives uniqueness and
minimality.
\end{proof}

\begin{lemma}[Distinct-coordinate fraction]
\label{lem:distinct-coordinates}
Let \(a\ge 1\) and draw \((I_1,\ldots,I_a)\) uniformly from \([m]^a\).
Then
\[
  \Pr[\exists\,r<s:\ I_r=I_s]\le \binom{a}{2}\frac{1}{m}.
\]
Hence if \(a=o(\sqrt m)\), the distinct-coordinate fraction tends to \(1\) as
\(m\to\infty\). In particular, for a uniformly random depth-\(d\) query from
\(\mathcal Q_d\) (chain) or \(\mathcal Q_d^{\parallel}\) (merge), we have
\(\kappa(q,B)=a(q)\) with probability \(1-o(1)\) whenever \(a(\Pi,d)=o(\sqrt m)\).
\end{lemma}

\begin{proof}
Union bound over all \(\binom{a}{2}\) pairs: \(\Pr[I_r=I_s]=1/m\).
\end{proof}

\begin{definition}[Non-colliding query]
\label{def:non-colliding}
A query $q=R_d(\vec i)\in\mathcal{Q}_d^{\parallel}$ is
\emph{non-colliding} if at every merge step the two body
instantiations are distinct ground facts.
\end{definition}

\begin{lemma}[Non-colliding fraction]
\label{lem:non-colliding-fraction}
For $\Pi_k^{\parallel}$ with $|B|=m\ge 2$ and $d\ge 2$,
$|\mathcal{Q}_d^{\parallel}\setminus
\mathcal{Q}_d^{\mathrm{nc}}|
/|\mathcal{Q}_d^{\parallel}|\le 2^d/m^k$.
For $d\le c\log_2 m$ with $c<k$, the non-colliding fraction
tends to~$1$.
\end{lemma}

\begin{proof}
The derivation tree has $2^{d-1}-1$ merge nodes; at each,
collision probability is $\le m^{-k}$; apply the union
bound.
\end{proof}

\begin{proposition}[Counting and structure for
  $\Pi_k^{\parallel}$]
\label{prop:merge-counting}
For $\Pi_k^{\parallel}$ with $|B|=m$ and $d\ge 1$:
\textup{(i)}~$|\mathcal{Q}_d^{\parallel}|
=m^{k\cdot 2^{d-1}}$.
\textup{(ii)}~For non-colliding queries, $G(q,B)$ is a
complete binary tree of height~$d$ with
$w_\ell=2^{d-\ell}$.
\textup{(iii)}~For all queries, $w_\ell\le 2^{d-\ell}$.
\textup{(iv)}~For non-colliding queries,
$W_{\max}=2^{d-1}$ and $W=2^d-1$.
\textup{(v)}~\(\kappa(q,B)\le a\) where \(a=a(\Pi_k^{\parallel},B,d)=k\cdot 2^{d-1}\),
with equality if the head coordinates are pairwise distinct.
\end{proposition}

\begin{definition}[Quantitative genericity]
\label{def:quantitative-generic}
A query $q$ in $\mathcal{Q}$ with $|\mathcal{Q}|=M$ is
\emph{$\eta$-generic} if
$K(q\mid\langle B\rangle)\ge\log M-\log(1/\eta)$.
At least a $1{-}\eta$ fraction of $\mathcal{Q}$ is
$\eta$-generic.
\end{definition}

\begin{corollary}[Joint typicality]
\label{cor:joint-typicality}
The fraction of $\mathcal{Q}_d^{\parallel}$ simultaneously
non-colliding and $(1/m)$-generic
\textup{(Definition~\ref{def:quantitative-generic})} is
$\ge 1-1/m-2^d/m^k$, tending to~$1$ when
$d=o(k\log_2 m)$ and $m\to\infty$.
\end{corollary}

\subsection{Coding Theorems: Per-Step Information and Capacity
  Separation}
\label{subsec:coding-theorems}

\begin{definition}[Canonical self-delimiting encoding of a finite base]
\label{def:canonical-encoding}
Fix the canonical order on finite knowledge bases
(Assumption~\ref{assump:finite-so}).
For any finite premise base \(B=\{b_1,\ldots,b_m\}\subseteq\mathbb S_O\),
let \(\langle B\rangle\) denote a fixed \emph{prefix-free} binary encoding of the
ordered list \((\enc_O(b_1),\ldots,\enc_O(b_m))\), using a standard self-delimiting
scheme (e.g., length-prefixing each \(\enc_O(b_i)\) and then concatenating).
We analogously encode any finite noisy base \(\tilde B\) as \(\langle \tilde B\rangle\).

All Kolmogorov complexities \(K(\cdot\mid\cdot)\) in this paper are taken with
respect to a fixed universal prefix machine and the fixed encoding convention
\(\langle\cdot\rangle\).
\end{definition}

Throughout, $\log:=\log_2$ and $K(\cdot\mid\cdot)$ denotes
prefix-free conditional Kolmogorov complexity with respect to a
fixed universal machine~\cite{li2008introduction}.

\begin{definition}[Information-rich regime]
\label{def:info-rich}
A family of instances $(q,B)$ lies in the
\emph{information-rich regime} if
$K(q\mid\langle B\rangle)=\omega(\log m)$ as $m\to\infty$.
\end{definition}

\begin{lemma}[Chain-form encoding]
\label{lem:chain-encoding}
A chain-form trace of length $n\ge 1$ from~$B$ with
$|B|=m$ admits a prefix-free encoding of length
$\kappa\log m+O(n)$,
where $\kappa=k{+}n{-}1$ is the effective base-selection
count.
\end{lemma}

\begin{proof}
Each step encodes a rule identifier
($\lceil\log r\rceil$ bits) and the base-premise pointers
($k_j\lceil\log m\rceil$ bits); the IDB pointer is
deterministic.
\end{proof}

\begin{theorem}[Tight coding theorem for chain derivations
  \textup{(CT1)}]
\label{thm:tight-datalog}
For $\Pi_k$ with $|B|=m\ge 2$, $q\in\mathcal{Q}_d$ with
$d\ge 1$, and $\kappa:=k{+}d{-}1$:
\begin{equation}\label{eq:tight-datalog-bounds}
  \kappa\log m-\log m-O(1)
  \;\le\;K(q\mid\langle B\rangle)
  \;\le\;\kappa\log m+O(d).
\end{equation}
The lower bound holds for at least a $1{-}1/m$ fraction.
For $(1/m)$-generic queries in the regime
$\log m,d\to\infty$:
$K(q\mid\langle B\rangle)=(1{+}o(1))\kappa\log m$,
with each derivation step contributing exactly $\log m$
bits.
\end{theorem}

\begin{proof}
\emph{Upper bound.}\;
Lemma~\ref{lem:chain-encoding} gives
$K\le\kappa\log m+O(d)$.

\emph{Lower bound.}\;
Since $|\mathcal{Q}_d|=m^{\kappa}$
(Proposition~\ref{prop:tuple-counting}) and
$|\{x:K(x\mid y)<\ell\}|<2^\ell$, setting
$\ell=\kappa\log m-\log m$ shows that fewer than
$m^{\kappa-1}$ queries have complexity below~$\ell$,
giving a compressible fraction $<1/m$.
\end{proof}

\begin{assumption}[Serializable depth witnesses]
\label{assump:serializable}
$c_{\min}\,\Dd(q\mid B)\le N(q\mid B)
\le c_{\max}\,\Dd(q\mid B)$
for constants $c_{\min},c_{\max}>0$, where $N(q\mid B)$ is
the minimal derivation trace length.
This holds with $c_{\min}=c_{\max}=1$ for both $\Pi_k$ and
$\Pi_k^{\parallel}$.
\end{assumption}

\begin{theorem}[Derivation depth as information metric]
\label{thm:derivation-depth-info-metric}
Under Assumption~\textup{\ref{assump:serializable}}, for
generic queries with $\Dd(q\mid B)\ge 1$:
$K(q\mid\langle B\rangle)
=\Theta(\Dd(q\mid B)\cdot\log(m{+}\Dd(q\mid B)))$.
\end{theorem}

\begin{proof}
By Theorem~\ref{thm:tight-datalog},
$K=(1{+}o(1))\kappa\log m$ for generic queries.
Assumption~\ref{assump:serializable} gives
$N=\Theta(\Dd)$; for $\Pi_k$,
$\kappa=k{+}\Dd{-}1=\Theta(\Dd)$
and the encoding cost per step is
$\Theta(\log m)=\Theta(\log(m{+}\Dd))$ when
$\Dd=O(\poly(m))$.
\end{proof}

\begin{theorem}[Depth-reachability capacity separation
  \textup{(CT2)}]
\label{thm:exponential-capacity}
Define
$C(\Pi,B,d):=\log|\{q:\Dd(q\mid B)=d\}|$.
Then:
\begin{enumerate}[label=\textup{(\Roman*)}]
\item \emph{Chain.}\;
$C(\Pi_k,B,d)=(k{+}d{-}1)\log m$.

\item \emph{Merge.}\;
$C(\Pi_k^{\parallel},B,d)=k\cdot 2^{d-1}\log m$.

\item \emph{Exponential separation.}\;
$C(\Pi_k^{\parallel},B,d)/C(\Pi_k,B,d)
=k\cdot 2^{d-1}/(k{+}d{-}1)\to\infty$.

\item \emph{Conservation at equal work.}\;
Both architectures yield $\Theta(W\log m)$ bits for
total work~$W$; the gap is in information per unit of
\emph{depth}, not per unit of work.
\end{enumerate}
\end{theorem}

\begin{proof}
\textbf{(I) Chain.}
By Proposition~\ref{prop:tuple-counting}(i), the depth-\(d\) head-arity is
\(a=k+d-1\) and \(|\mathcal Q_d|=m^{a}\). Hence
\(C(\Pi_k,B,d)=\log|\mathcal Q_d|=a\log m=(k+d-1)\log m\).

\textbf{(II) Merge.}
By Proposition~\ref{prop:merge-counting}(i), the depth-\(d\) head-arity is
\(a=k\cdot 2^{d-1}\) and \(|\mathcal Q_d^{\parallel}|=m^{a}\). Hence
\(C(\Pi_k^{\parallel},B,d)=a\log m=k\cdot 2^{d-1}\log m\).

\textbf{(III) Separation.}
Immediate from the ratio \(k\cdot 2^{d-1}/(k+d-1)\to\infty\) as \(d\to\infty\).

\textbf{(IV) Equal-work remark.}
For \(\Pi_k\), \(W=\Theta(d)\) and \(C=\Theta(W\log m)\).
For non-colliding \(\Pi_k^{\parallel}\), Proposition~\ref{prop:merge-counting}(iv)
gives \(W=2^d-1\) while \(C=\Theta(2^d\log m)=\Theta(W\log m)\).
\end{proof}

The derivation depth $\Dd(q\mid B)$ coincides with the
\emph{parallel evaluation depth}: the number of synchronous
rounds in bottom-up evaluation from~$B$.

\begin{proposition}[$\Dd$ equals parallel evaluation depth]
\label{prop:dd-parallel}
Under Assumption~\ref{assump:alignment}, the parallel
evaluation depth of~$q$ from~$B$ equals $\Dd(q\mid B)$.
\end{proposition}

\begin{proof}
Both quantities satisfy the same recurrence: a fact in~$B$ is
evaluated in round~$0$, and a non-base fact is evaluated in the
round after all its predecessors are available.
The unique solution is~\eqref{eq:Dd-inductive}.
\end{proof}

\subsection{Stratified Storage--Depth Tradeoff}
\label{subsec:stratified-tradeoff}

\begin{definition}[Layer-prefix caching]
\label{def:layer-prefix-cache}
For $G=G(q,B)$ with $d=\Dd(q\mid B)$ and
$\ell_0\in\{0,\ldots,d\}$:
$S_{\ell_0}:=\{v\in V(G):1\le\Dd(v\mid B)\le\ell_0\}$.
\end{definition}

\begin{proposition}[Depth reduction]
\label{prop:layer-prefix-depth}
$\Dd(q\mid B\cup S_{\ell_0})=d-\ell_0$.
\end{proposition}

\begin{theorem}[Stratified Pareto curve]
\label{thm:stratified-tradeoff}
Let $r:=|R|$ be the rule count and
$M_\ell:=m+\sum_{j=1}^{\ell-1}w_j$ the cumulative pool
size at layer~$\ell$.
Define the per-layer information
$I_\ell:=w_\ell(k\log M_\ell+\log r)$ and the cumulative
information $I(\ell_0):=\sum_{\ell=1}^{\ell_0}I_\ell$.
The layer-prefix strategy with parameter $\ell_0$ achieves
storage
$\sigma(\ell_0)\le I(\ell_0)+O(W+d\log W_{\max})$
and residual depth $d{-}\ell_0$.
For $\Pi_k$ the tradeoff is linear
\textup{(}$\Theta(\log m)$ bits per depth unit\textup{)};
for $\Pi_k^{\parallel}$ it is exponentially front-loaded
\textup{(}caching layer~$1$ costs $\Theta(K/2)$ but
reduces depth by only~$1$\textup{)}.
\end{theorem}

\subsection{Storage--Computation Tradeoff}
\label{subsec:tradeoff}

\begin{definition}[Cost model]
\label{def:cost-model}
Storage costs $\rho_s:=\alpha/\beta$ per bit (normalized);
computation costs $1$ per derivation step.
For query~$q$ accessed $f_q$ times over a horizon:
caching costs $\rho_s\cdot\ell_q/f_q+O(1)$ per access,
where $\ell_q=K(q\mid\langle B\rangle)+O(1)$;
on-demand derivation costs $\Dd(q\mid B)$ per access.
\end{definition}

\begin{theorem}[Critical frequency \textup{(CT3)}]
\label{thm:freq_tradeoff}
Under Assumption~\textup{\ref{assump:serializable}} and the
information-rich regime, the break-even frequency is
\begin{equation}\label{eq:critical-freq}
  f_c=\Theta(\rho_s\cdot\log(m{+}d)).
\end{equation}
For $f_q\gg f_c$, caching is optimal; for $f_q\ll f_c$,
on-demand derivation is optimal.
\end{theorem}

\subsection{Representation Invariance}
\label{subsec:rep-inv}

A foundational requirement is that the coding theorems measure
information content rather than representation artifacts.

\begin{definition}[Inference system and aligned
  bi-interpretation]
\label{def:inference-system}
An \emph{inference system}
$\mathfrak{I}:=(S_O,P_O,B)$ pairs a state set with a
predecessor operator and a base.
An \emph{aligned bi-interpretation}
$\Phi=(\tau,\bar\tau):
\mathfrak{I}_1\rightleftharpoons\mathfrak{I}_2$
consists of compositional interpretations witnessing
synonymy $S_O^{(1)}\equiv_{\mathcal L}S_O^{(2)}$,
additionally satisfying base fidelity
\textup{(}$|B_1|=|B_2|=m$\textup{)}, query
correspondence, and derivation compatibility with
dilation $(\delta_\tau,\delta_{\bar\tau})$.
When $\delta_\tau=\delta_{\bar\tau}=1$, $\Phi$ is
\emph{isometric}.
\end{definition}

\begin{theorem}[Representation invariance]
\label{thm:rep-inv}
Under an aligned bi-interpretation with dilation
$(\delta_\tau,\bar\delta_\tau)$:
\begin{enumerate}[label=\textup{(\alph*)}]
\item \emph{Depth sandwich.}\;
$\Dd_{\Pi_1}(q\mid B_1)/\bar\delta_\tau
\le\Dd_{\Pi_2}(q^\tau\mid B_2)
\le\delta_\tau\,\Dd_{\Pi_1}(q\mid B_1)$.

\item \emph{Complexity invariance.}\;
$|K(q\mid\langle B_1\rangle)
-K(q^\tau\mid\langle B_2\rangle)|
\le O(|\tau|{+}|\bar\tau|)$.

\item \emph{Isometric sharpening.}\;
$\delta_\tau=\bar\delta_\tau=1$ gives exact depth equality.
\end{enumerate}
\end{theorem}

\begin{theorem}[Core theorem transfer]
\label{thm:CT-transfer}
All four coding theorems \textup{(}CT1--CT3 of this section
and CT4 of Theorem~\textup{\ref{thm:storage-resilience-capacity})}
transfer across synonymous representations with
dilation-bounded distortion: CT1's coding constant
transforms by $[\kappa_1/\delta_\tau,\kappa_1\bar\delta_\tau]$;
CT2's capacity satisfies
$C_1(\lfloor d/\delta_\tau\rfloor)\le C_2(d)
\le C_1(\bar\delta_\tau\cdot d)$;
CT3's critical frequency satisfies
$f_c^{(1)}/\delta_\tau\le f_c^{(2)}
\le\bar\delta_\tau\,f_c^{(1)}$;
CT4's cache size is modified by at most
$\delta_\tau\bar\delta_\tau$.
Under isometry all quantities are exactly preserved.
The $\Theta$-level laws and qualitative dichotomies are
canonical invariants of the synonymy class
$[\mathfrak{I}]_{\equiv_{\mathcal{L}}}$.
\end{theorem}

\paragraph{Noisy information.}
The noise model used in
Sections~\ref{sec:noise-resilience}--\ref{sec:arch-depth} is
the \emph{noisy semantic base} of
Definition~\ref{def:noisy}: for a given $S_O\subseteq\mathbb S_O$,
any set $\tilde S_O:=(S_O\setminus S_O^-)\cup S_O^+$ with
$S_O^-\subseteq S_O$ (lost) and
$S_O^+\subseteq\mathbb S_O\setminus S_O$ (spurious).
Under loss only ($S_O^+=\varnothing$),
$\Cn(\tilde S_O)\subseteq\Cn(S_O)$ and
$\Dd(q\mid\tilde S_O)\ge\Dd(q\mid S_O)$
(Theorem~\ref{thm:contraction-expansion}(I));
the representation-invariance results above extend to this
noisy setting via
Proposition~\ref{prop:noisy-exist}.


\section{Storage--Resilience Capacity under Premise Erasure}
\label{sec:noise-resilience}%

When the premise base is subject to stochastic loss, an online
derivation engine faces a fundamental tension: deeper queries
depend on more base facts, each of which may be erased, yet
storing all dependencies pre-emptively consumes the very
storage budget the derivation was meant to economize.
This section resolves that tension for
\emph{derivation-constrained} decoders---those that must
derive the query answer from the surviving premises and a
reliable cache of intermediate logical facts.

The main result is the \emph{Storage--Resilience Capacity
Theorem} (Theorem~\ref{thm:storage-resilience-capacity}, CT4):
under i.i.d.\ premise erasure with rate~$\varepsilon$ and
resilience target~$\delta$, the minimum reliable cache for a
generic query with $\kappa$ base-fact dependencies is
$\sigma^*=(\kappa-N^*)\log m+O(\kappa)$,
where $N^*\approx\delta/\varepsilon$.
Here and throughout CT4, \(\kappa\) denotes the \emph{distinct} dependency count
\(\kappa(q,B)\) (Definition~\ref{def:dep-count}).
The residual $N^*\log m$ bits represent the ``channel
capacity'' of the erasure-prone base---the information that
can be transmitted reliably without caching.
Section~\ref{sec:precise} will compare this derivation-constrained
optimum to the coded-caching optimum
$\sigma^*_{\mathrm{code}}=\varepsilon\kappa\log m$, revealing
the $1/\varepsilon$ derivation penalty.

Throughout, every noisy base $\tilde B$ used to evaluate
$\Dd(\cdot\mid\tilde B)$ is finite with decidable membership, so
Theorem~\ref{thm:computability-Dd} applies.

\subsection{Premise Perturbations and the Contraction--Expansion
  Dichotomy}
\label{subsec:noisy-setup}

We begin by grounding the noise model in the formal framework of
Section~\ref{sec:model} and then establish the monotone behavior of
$\Cn(\cdot)$ and $\Dd(\cdot\mid\cdot)$ under set-level
perturbations.

\begin{proposition}[Noisy premise bases as instances of the
  formal noise framework]
\label{prop:noisy-base-grounding}
Let\/ $\mathcal{I}$ be an information instance
\textup{(Definition~\ref{def:info-instance})} with finite
$S_O$, $B_0\subseteq S_O\subseteq\mathbb{S}_O$ a finite
premise base, $B_0^-\subseteq B_0$ and
$B_0^+\subseteq\mathbb{S}_O\setminus B_0$ finite, and
$\tilde B_0:=(B_0\setminus B_0^-)\cup B_0^+$.
Then:
\textup{(i)}~$\tilde B_0$ is a noisy semantic base
\textup{(Definition~\ref{def:noisy})};
\textup{(ii)}~$\tilde B_0$ is finite and effectively decidable;
\textup{(iii)}~$\Dd(s\mid\tilde B_0)$ is well-defined and
computable for every $s\in\mathbb{S}_O$;
\textup{(iv)}~$\Cn(\tilde B_0)$ is well-defined and the query
partition
\textup{(Definition~\ref{def:query-partition})} is determined by
monotonicity of $\Cn(\cdot)$.
\end{proposition}

\begin{proof}
(i)~follows from
Assumption~\ref{assump:semantic-universe}.
(ii)~uses finiteness of $B_0,B_0^-,B_0^+$ and decidable
membership in $B_0$
(Axiom~\ref{ax:state-representation}(R1)).
(iii)~applies Theorem~\ref{thm:computability-Dd} with
$B:=\tilde B_0$.
(iv)~uses Assumption~\ref{assump:proof-system} and monotonicity
of $\Cn$.
\end{proof}

\begin{definition}[Noisy premise base (baseline-parametric)]
\label{def:noisy-base}
Fix a finite baseline $B_0\subseteq\mathbb S_O$.
A \emph{noisy premise base} is
$\tilde B_0:=(B_0\setminus B_0^-)\cup B_0^+$
with $B_0^-\subseteq B_0$ (lost) and
$B_0^+\subseteq\mathbb S_O\setminus B_0$ (spurious).
Write $B_\cap:=B_0\setminus B_0^-$, $\ell:=|B_0^-|$,
$p:=|B_0^+|$, $\mathcal N:=\ell+p$.
\end{definition}

\begin{assumption}[Bounded and effectively describable noise]
\label{assump:bounded-noise}
There exists a computable family $\{\mathcal V_s\}\subseteq
\mathbb S_O$ with $|\mathcal V_s|\le\poly(s)$ and
$B_0\cup B_0^+\subseteq\mathcal V_{m+\mathcal N}$;
given either $\langle B_0\rangle$ or
$\langle\tilde B_0\rangle$, the sets $B_0^-,B_0^+$ are
describable by index lists of total length
$O(\mathcal N\log(m+\mathcal N))$.
\end{assumption}

\begin{definition}[Sound, spurious, and lost query sets]
\label{def:query-partition}
$\mathcal{Q}_{\mathrm{sound}}
 :=\Cn(\tilde B_0)\cap\Cn(B_0)$,\;
$\mathcal{Q}_{\mathrm{spur}}
 :=\Cn(\tilde B_0)\setminus\Cn(B_0)$,\;
$\mathcal{Q}_{\mathrm{lost}}
 :=\Cn(B_0)\setminus\Cn(\tilde B_0)$.
\end{definition}

\begin{definition}[Consistency]
\label{def:consistency-class}
A finite premise base $B$ is \emph{consistent} if
$\Cn(B)$ is not the set of all well-formed formulas.
\end{definition}

\paragraph{Monotonicity lemmas.}

\begin{lemma}[Monotonicity in the premise base]
\label{lem:premise-monotonicity-sec3}
If $B_1\subseteq B_2$, then $\Dd(s\mid B_2)\le\Dd(s\mid B_1)$
for every $s\in\mathbb S_O$.
\end{lemma}

\begin{proof}
By well-founded induction on $\Dd(s\mid B_1)$.
If $s\in B_1$, then $s\in B_2$ as well, so
$\Dd(s\mid B_2)=0=\Dd(s\mid B_1)$.
If $s\in B_2\setminus B_1$, then
$\Dd(s\mid B_2)=0\le\Dd(s\mid B_1)$.
If $s\notin B_2$, then $s\notin B_1$ and
\[
  \Dd(s\mid B_2)
  =1+\max_{s'\in P_O(s)}\Dd(s'\mid B_2)
  \le 1+\max_{s'\in P_O(s)}\Dd(s'\mid B_1)
  =\Dd(s\mid B_1),
\]
where the inequality uses the induction hypothesis applied to
each $s'\in P_O(s)$ (with $\Dd(s'\mid B_1)<\Dd(s\mid B_1)$).
\end{proof}

\begin{lemma}[Base-conversion description length]
\label{lem:base-conversion}
Under Assumption~\ref{assump:bounded-noise},
$K(\langle\tilde B_0\rangle\mid\langle B_0\rangle)
=O(\mathcal N\log(m+\mathcal N))$
and
$K(\langle B_0\rangle\mid\langle\tilde B_0\rangle)
=O(\mathcal N\log(m+\mathcal N))$.
\end{lemma}

\begin{proof}
Describe $B_0^-\cup B_0^+$ by index lists over
$\mathcal V_{m+\mathcal N}$.
\end{proof}

\paragraph{The contraction--expansion dichotomy.}

\begin{theorem}[Contraction--expansion dichotomy]
\label{thm:contraction-expansion}
Let $B_0$ be finite and consistent,
$\tilde B_0=(B_0\setminus B_0^-)\cup B_0^+$.

\begin{enumerate}[label=\textup{(\Roman*)}]
\item \emph{Loss-only ($B_0^+=\varnothing$).}
$\Cn(\tilde B_0)\subseteq\Cn(B_0)$;
$\mathcal{Q}_{\mathrm{spur}}=\varnothing$;
$\tilde B_0$ is consistent;
$\Dd(q\mid\tilde B_0)\ge\Dd(q\mid B_0)$ for
$q\in\Cn(\tilde B_0)$;
$q\notin\Cn(\tilde B_0)$ for every
$q\in\mathcal{Q}_{\mathrm{lost}}$
\textup{(irrecoverable without external premises)}.

\item \emph{Pollution-only ($B_0^-=\varnothing$).}
$\Cn(\tilde B_0)\supseteq\Cn(B_0)$;
$\mathcal{Q}_{\mathrm{lost}}=\varnothing$;
$\Dd(q\mid\tilde B_0)\le\Dd(q\mid B_0)$ for
$q\in\Cn(B_0)$.
If $B_0\cup\{b^+\}$ is inconsistent for some
$b^+\in B_0^+$, then $\Cn(\tilde B_0)$ is the set of all
wffs.

\item \emph{Combined noise.}
Neither monotonicity direction is guaranteed; all three query
classes may be nonempty.
\end{enumerate}
\end{theorem}

\begin{proof}
(I)~$\tilde B_0=B_\cap\subseteq B_0$; monotonicity of $\Cn$
gives $\Cn(\tilde B_0)\subseteq\Cn(B_0)$.
Consistency: any model of $B_0$ satisfies $\tilde B_0$.
Depth increase: Lemma~\ref{lem:premise-monotonicity-sec3}.
(II)~$\tilde B_0\supseteq B_0$; monotonicity of $\Cn$ gives
expansion.
Depth decrease: Lemma~\ref{lem:premise-monotonicity-sec3}.
Consistency catastrophe: ex falso quodlibet.
(III)~Competing effects.
\end{proof}

\subsection{Loss--Computation Duality}
\label{subsec:loss-computation}

This subsection establishes a position-sensitive exchange rate
between lost storage and increased computation.
The duality applies when lost premises are reconstructible from
the surviving base (e.g., operational shortcuts
$J\subseteq\Cn(A)$ with core $A\subseteq B_\cap$);
when irredundant core premises are lost, resilience requires
reliable pre-storage, as will be formalized by the fragility
proposition and CT4 in Section~\ref{subsec:stochastic-erasure}.

\begin{definition}[Global reconstruction depth]
\label{def:rec-depth-preserved}
Let $B_0^-\subseteq B_0$ and $B_\cap:=B_0\setminus B_0^-$.
The \emph{reconstruction depth} is
$d_{\mathrm{rec}}:=\max_{b\in B_0^-}\Dd(b\mid B_\cap)$,
with $d_{\mathrm{rec}}:=0$ if $B_0^-=\varnothing$ and
$d_{\mathrm{rec}}:=\infty$ if some lost premise is not
derivable from $B_\cap$.
\end{definition}

\begin{theorem}[Loss--computation duality]
\label{thm:loss-comp-duality}
Let $B_0^-\subseteq B_0$, $B_\cap=B_0\setminus B_0^-$,
$q\in\Cn(B_\cap)$.

\begin{enumerate}[label=\textup{(\Roman*)}]
\item \emph{Depth shift.}
\begin{equation}\label{eq:depth-shift}
  \Dd(q\mid B_0)
  \;\le\;
  \Dd(q\mid B_\cap)
  \;\le\;
  \Dd(q\mid B_0)+d_{\mathrm{rec}}.
\end{equation}

\item If $B_0^-\cap V_0(G(q,B_0))=\varnothing$
\textup{(no lost fact appears in the dependency set)}, then
$\Dd(q\mid B_\cap)=\Dd(q\mid B_0)$.

\item Under combined noise
$\tilde B_0=(B_0\setminus B_0^-)\cup B_0^+$:
$\Dd(q\mid\tilde B_0)\le\Dd(q\mid B_\cap)
 \le\Dd(q\mid B_0)+d_{\mathrm{rec}}$.
\end{enumerate}
\end{theorem}

\begin{proof}
(I)~The lower bound is
Lemma~\ref{lem:premise-monotonicity-sec3} ($B_\cap\subseteq B_0$).
For the upper bound we prove
$\Dd(v\mid B_\cap)\le\Dd(v\mid B_0)+d_{\mathrm{rec}}$
for every $v$ in the backward unfolding of~$q$,
by well-founded induction on $\Dd(v\mid B_0)$.

\emph{Base.}\;$v\in B_0$.
If $v\in B_\cap$, then
$\Dd(v\mid B_\cap)=0\le 0+d_{\mathrm{rec}}$.
If $v\in B_0^-$, then
$\Dd(v\mid B_\cap)\le d_{\mathrm{rec}}
 =0+d_{\mathrm{rec}}=\Dd(v\mid B_0)+d_{\mathrm{rec}}$
by Definition~\ref{def:rec-depth-preserved}.

\emph{Step.}\;$v\notin B_0$.
Each $v'\in P_O(v)$ satisfies
$\Dd(v'\mid B_0)<\Dd(v\mid B_0)$; by the induction
hypothesis,
$\Dd(v'\mid B_\cap)\le\Dd(v'\mid B_0)+d_{\mathrm{rec}}$.
Hence
$\Dd(v\mid B_\cap)
 =1+\max_{v'}\Dd(v'\mid B_\cap)
 \le\Dd(v\mid B_0)+d_{\mathrm{rec}}$.
Specializing to $v:=q$ yields~\eqref{eq:depth-shift}.

(II)~The derivation of $q$ from $B_0$ uses only premises in
$B_\cap$; equality follows from
Lemma~\ref{lem:premise-monotonicity-sec3}.

(III)~$B_\cap\subseteq\tilde B_0$ gives the first inequality by
Lemma~\ref{lem:premise-monotonicity-sec3}; the second is from~(I).
\end{proof}

\begin{theorem}[Information-theoretic form]
\label{thm:loss-comp-info}
Under the hypotheses of Theorem~\ref{thm:loss-comp-duality},
the information-rich regime
\textup{(Definition~\ref{def:info-rich})}, and
Assumption~\ref{assump:serializable}, for generic queries:
\begin{equation}\label{eq:loss-comp-info}
  K(q\mid\langle B_\cap\rangle)
  \;=\;
  \Theta\!\bigl(\Dd(q\mid B_\cap)\cdot
    \log(\tilde m+\Dd(q\mid B_\cap))\bigr)
  \;=\;
  K_q+O(\mathcal N\log(m+\mathcal N)),
\end{equation}
where $\tilde m=|B_\cap|$ and
$K_q:=K(q\mid\langle B_0\rangle)$.
\end{theorem}

\begin{proof}
The first equality is
Theorem~\ref{thm:derivation-depth-info-metric} applied to
$(q,B_\cap)$; the second is Lemma~\ref{lem:base-conversion}.
\end{proof}

\subsection{Tight Noisy Coding Theorem}
\label{subsec:noisy-coding}

\begin{definition}[Noisy query sets]
\label{def:noisy-query-sets}
Let $\Pi_k$ have base $B$ with $|B|=m\ge 2$,
$B^-\subseteq B$ ($|B^-|=\ell$),
$B^+\subseteq\mathbb S_O\setminus B$ ($|B^+|=p$),
$\tilde B:=(B\setminus B^-)\cup B^+$ with
$\tilde m:=m-\ell+p\ge 2$.
For $n\ge 1$:
$\tilde{\mathcal{Q}}_n
 :=\{T_{k+n-1}(\vec i):A(i_j)\in\tilde B\;\forall j\}$
and
$\mathcal{Q}_n^{\mathrm{sound}}
 :=\{q\in\tilde{\mathcal{Q}}_n:
   \text{all }A(i_j)\in B\setminus B^-\}$.
\end{definition}

\begin{theorem}[Tight noisy coding theorem]
\label{thm:noisy-coding-combined}
Under Definition~\ref{def:noisy-query-sets} with
$\tilde\kappa:=k+n-1$ and $\varepsilon:=\ell/m$:

\begin{enumerate}[label=\textup{(\Roman*)}]
\item $|\tilde{\mathcal{Q}}_n|=\tilde m^{\tilde\kappa}$;
for generic queries,
$K(q\mid\langle\tilde B\rangle)
 =(1{+}o(1))\tilde\kappa\log\tilde m$.

\item $\log\tilde m
 =\log m+\log(1{-}\varepsilon)
  +\log((m{-}\ell{+}p)/(m{-}\ell))$.

\item $|\mathcal{Q}_n^{\mathrm{sound}}|
 /|\tilde{\mathcal{Q}}_n|
 =((m{-}\ell)/\tilde m)^{\tilde\kappa}$.

\item $C(\tilde B,n)-C(B,n)
 =\tilde\kappa\log(\tilde m/m)$.
\end{enumerate}
\end{theorem}

\begin{proof}
Apply Theorem~\ref{thm:tight-datalog} to $\Pi_k$ with base
$\tilde B$ of size $\tilde m$; counting via
Proposition~\ref{prop:tuple-counting}.
\end{proof}

\begin{corollary}[Noise-balanced regime]
\label{cor:noise-balanced}
When $\ell=p$: $\tilde m=m$, capacity unchanged, but the sound
fraction is $(1{-}\varepsilon)^{\tilde\kappa}$---exponentially
small for fixed $\varepsilon>0$.
\end{corollary}

\subsection{Stochastic Erasure and the Storage--Resilience
  Capacity Theorem}
\label{subsec:stochastic-erasure}

\paragraph{Base-fact dependencies.}
We use the distinct dependency count \(\kappa(q,B)\) and the head arity \(a(q)\)
as in Definition~\ref{def:dep-count}.
For the two architectures, the depth-\(d\) head arities are
\(a(\Pi_k,d)=k+d-1\) and \(a(\Pi_k^{\parallel},d)=k\cdot 2^{d-1}\).
When we invoke i.i.d.\ erasure survival probabilities \((1-\varepsilon)^{\kappa(q,B)}\),
we explicitly assume the dependencies are \emph{distinct} (i.e., \(\kappa(q,B)\)
is the number of distinct leaves), or we restrict to the distinct-coordinate
subfamily where \(\kappa(q,B)=a(q)\).

\begin{lemma}[Minimality of $\kappa$ for faithful programs]
\label{lem:kappa-minimal}
Let $\Pi$ be faithful with unique derivation traces
\textup{(}satisfied by $\Pi_k$ and $\Pi_k^{\parallel}$
by Lemmas~\textup{\ref{lem:faithfulness}}
and~\textup{\ref{lem:merge-faithfulness})}.
Then $G(q,B)$ is the unique canonical derivation DAG of~$q$,
$\kappa(q,B)=|V_0(G(q,B))|$, and the values in
Propositions~\textup{\ref{prop:tuple-counting}}
and~\textup{\ref{prop:merge-counting}(v)} hold.
\end{lemma}

\begin{proof}
Faithfulness and unique defining rules make the trace
uniquely determined by the head tuple, by induction on depth.
\end{proof}

\paragraph{The i.i.d.\ erasure channel and resilience.}

\begin{definition}[i.i.d.\ premise erasure channel]
\label{def:iid-erasure}
Each $b_i\in B$ is independently erased with probability
$\varepsilon\in(0,1)$, yielding $\tilde B\subseteq B$ with
$\E[|\tilde B|]=(1{-}\varepsilon)m$.
A query $q$ \emph{natively survives} if
$V_0(G(q,B))\subseteq\tilde B$.
The decoder observes cache $S$ (reliable) and $\tilde B$.
Each realization is a loss-only noisy base
\textup{(Proposition~\ref{prop:noisy-base-grounding})},
so Theorem~\ref{thm:contraction-expansion}\textup{(I)}
ensures soundness.
\end{definition}

\begin{definition}[Resilience threshold]
\label{def:noise-threshold}
$N^*(\varepsilon,\delta)
 :=\lfloor\ln(1/(1{-}\delta))/\ln(1/(1{-}\varepsilon))\rfloor$.
For small $\varepsilon,\delta$: $N^*\approx\delta/\varepsilon$.
\end{definition}

\begin{proposition}[Survival probability]
\label{prop:survival-prob}
Under i.i.d.\ erasure with rate $\varepsilon$, a query with
$\kappa$ distinct base-fact dependencies survives natively with
probability $(1{-}\varepsilon)^\kappa$ and is natively
$\delta$-resilient iff $\kappa\le N^*(\varepsilon,\delta)$.
\end{proposition}

\begin{definition}[$\delta$-resilience]
\label{def:N-resilience}
A cache $S\subseteq\Cn(B)$ makes $T\subseteq\Cn(B)$
\emph{$\delta$-resilient} if
$\Pr[T\subseteq\Cn(\tilde B\cup S)]\ge 1{-}\delta$.
\end{definition}

\begin{proposition}[Fragility of irredundant cores]
\label{prop:fragility}
If $a\in A:=\Atom(S_O)$ and $T\ni a$, then
$\Pr[T\subseteq\Cn(\tilde A)]\le 1{-}\varepsilon$
by irredundancy
\textup{(Proposition~\ref{prop:atom-core-correct}(ii))}.
Hence $T$ cannot be $\delta$-resilient with
$\delta<\varepsilon$ unless $S$ covers $a$.
\end{proposition}

\paragraph{Structural caching rigidity.}

The following theorem is the key structural ingredient for the
converse: it shows that only cache facts lying \emph{within the
derivation DAG of the target query} contribute to resilience.

\begin{theorem}[Structural caching rigidity]
\label{thm:ancestral-relevance}
Let $\Pi$ be faithful with unique traces and satisfy
Assumption~\ref{assump:alignment}, $q\in\Cn(B)$,
$\tilde B\subseteq B$, $S\subseteq\Cn(B)$.
Then
\begin{equation}\label{eq:ancestral-relevance}
  q\in\Cn(\tilde B\cup S)
  \;\Longleftrightarrow\;
  q\in\Cn(\tilde B\cup(S\cap V(G(q,B)))).
\end{equation}
\end{theorem}

\begin{proof}
$(\Leftarrow)$ is immediate.
$(\Rightarrow)$: by induction on $d:=\Dd(q\mid B)$.

\emph{Base} ($d=0$): $q\in B$.
In both $\Pi_k$ and $\Pi_k^{\parallel}$, the EDB predicate~$A$
appears only in rule bodies, never in rule heads; hence $q$
cannot be derived by any rule and must already belong to
$\tilde B\cup S$.
If $q\in S$ then $q\in V(G(q,B))$, so
$q\in\tilde B\cup(S\cap V(G(q,B)))$.

\emph{Step} ($d\ge 1$): let $P_O(q)=\{v_1,\ldots,v_a\}$.
By faithfulness and unique traces, any derivation of $q$ from
$\tilde B\cup S$ must apply the same rule with body
instantiation $v_1,\ldots,v_a$; in particular, each $v_i$ must
lie in $\Cn(\tilde B\cup S)$.
For each $v_i$: if $v_i\in\tilde B$, then
$v_i\in\Cn(\tilde B\cup(S\cap V(G(q,B))))$;
if $v_i\in S$, then
$v_i\in P_O(q)\subseteq V(G(q,B))$, so
$v_i\in S\cap V(G(q,B))$;
if $v_i\notin\tilde B\cup S$, the induction hypothesis
applies since $\Dd(v_i\mid B)\le d{-}1$ and
$G(v_i,B)\subseteq G(q,B)$, giving
$v_i\in\Cn(\tilde B\cup(S\cap V(G(v_i,B))))
\subseteq\Cn(\tilde B\cup(S\cap V(G(q,B))))$.
Applying the unique rule yields the conclusion.
\end{proof}

\begin{corollary}[Absorption characterization]
\label{cor:absorption-char}
Define the useful cache $U(q,S):=S\cap V(G(q,B))$,
the absorbed leaf set
$D_{\mathrm{abs}}(q,S)
 :=\bigcup_{f\in U}V_0(G(f,B))$,
and the exposed leaf set
$D_{\mathrm{exp}}(q,S)
 :=V_0(G(q,B))\setminus D_{\mathrm{abs}}(q,S)$.
Then:
\textup{(i)}~$q\in\Cn(\tilde B\cup S)$ implies
$D_{\mathrm{exp}}\subseteq\tilde B$;
\textup{(ii)}~$\Pr[q\in\Cn(\tilde B\cup S)]
\le(1{-}\varepsilon)^{|D_{\mathrm{exp}}|}$.
\end{corollary}

\begin{proof}
(i)~By Theorem~\ref{thm:ancestral-relevance}, only $U$ matters;
exposed leaves are not supplied by $U$ and must lie in
$\tilde B$.
(ii)~Since $D_{\mathrm{exp}}\subseteq V_0(G(q,B))\subseteq B$
and each element of~$B$ is independently erased,
$\Pr[D_{\mathrm{exp}}\subseteq\tilde B]
=(1{-}\varepsilon)^{|D_{\mathrm{exp}}|}$.
Combining with~(i) gives the bound.
\end{proof}

\paragraph{Conditional encoding and the structural converse.}

\begin{lemma}[Conditional encoding given an ancestral cache]
\label{lem:conditional-encoding}
Under the hypotheses of Theorem~\ref{thm:ancestral-relevance}, let \(q\) have
distinct dependency count \(\kappa=\kappa(q,B)\), and let \(S\subseteq\Cn(B)\).
Let \(j:=|D_{\mathrm{abs}}(q,S)|\) be the absorbed leaf count
(Corollary~\ref{cor:absorption-char}).
Then
\[
  K(q\mid\langle B\rangle,S)
  \;\le\;
  (\kappa-j)\log m + O(\kappa).
\]
\end{lemma}

\begin{proof}
By Theorem~\ref{thm:ancestral-relevance}, we may assume w.l.o.g.\ that
\(S\subseteq V(G(q,B))\), since replacing \(S\) by \(S\cap V(G(q,B))\) does not
change \(q\in\Cn(\tilde B\cup S)\) and can only reduce \(K(S\mid\langle B\rangle)\).

Fix a canonical enumeration of the \(\kappa\) distinct leaves
\(V_0(G(q,B))=\{A(u_1),\ldots,A(u_\kappa)\}\) under the fixed canonical order on
\(B=\{A(i):i\in[m]\}\) (Assumption~\ref{assump:finite-so}).
A description of \(q\) given \(\langle B\rangle\) and \(S\) consists of:
\begin{enumerate}[label=\textup{(\roman*)}]
\item a \(\kappa\)-bit mask indicating which of the \(\kappa\) leaves are exposed
(thus \(\kappa-j\) positions), costing \(O(\kappa)\) bits;
\item the identities (indices in \([m]\)) of the exposed leaves, costing
\((\kappa-j)\lceil\log m\rceil\) bits;
\item an \(O(\kappa)\)-bit tie-breaking string that selects, among all candidates
consistent with (i)--(ii) and with the cached facts \(S\), the intended query
\(q\) (e.g., by specifying the required left/right instantiations at symmetric
merge steps in a fixed traversal order of the derivation tree).
\end{enumerate}
Given \(\langle B\rangle\), \(S\), and the above data, a universal decoder can
search (unboundedly) over candidates and output the unique target specified by
the tie-breaking string. The total description length is
\((\kappa-j)\log m+O(\kappa)\).
\end{proof}

\begin{theorem}[Structural converse for faithful programs]
\label{thm:structural-converse}
Under i.i.d.\ erasure with rate $\varepsilon$ and target
$\delta$, let $N^*:=N^*(\varepsilon,\delta)$.
Let $q$ be $(1/m)$-generic
\textup{(Definition~\ref{def:quantitative-generic})} with
$\kappa>N^*$ dependencies.
Any $\delta$-resilient cache $S\subseteq\Cn(B)$ satisfies
\begin{equation}\label{eq:structural-converse}
  K(S\mid\langle B\rangle)
  \ge(\kappa-N^*)\log m-O(\kappa+\log m).
\end{equation}
The structural rigidity
(Theorem~\ref{thm:ancestral-relevance}) is essential: without
it, coded strategies could defy coordinate-level accounting.
A Fano-type converse~\cite[Ch.~2]{cover2006elements} yields only
$\Omega(\kappa\varepsilon\log m)$, weaker by a factor
of~$\varepsilon$.
\end{theorem}

\begin{proof}
\emph{Step~1.}
Let $j:=|D_{\mathrm{abs}}(q,S)|$.
By Corollary~\ref{cor:absorption-char}(ii),
$(1{-}\varepsilon)^{\kappa-j}\ge 1{-}\delta$, so
$j\ge\kappa-N^*$.

\emph{Step~2.}
The chain rule for prefix-free
complexity~\cite{li2008introduction}:
$K(q\mid\langle B\rangle)
 \le K(S\mid\langle B\rangle)
   +K(q\mid\langle B\rangle,S)
   +O(\log K(q\mid\langle B\rangle))$.
By Lemma~\ref{lem:conditional-encoding},
$K(q\mid\langle B\rangle,S)\le(\kappa{-}j)\log m+O(\kappa)$.
By $(1/m)$-genericity,
$K(q\mid\langle B\rangle)\ge\kappa\log m-\log m-O(1)$.

\emph{Step~3.}
$K(S\mid\langle B\rangle)
 \ge(\kappa\log m-\log m-O(1))
   -((\kappa{-}j)\log m+O(\kappa))
   -O(\log(\kappa\log m))
 =j\log m-O(\kappa+\log m)
 \ge(\kappa-N^*)\log m-O(\kappa+\log m)$.
\end{proof}

\paragraph{The storage--resilience capacity theorem.}

\begin{theorem}[Storage--resilience capacity theorem
  \textup{(CT4)}]
\label{thm:storage-resilience-capacity}
Let $\Pi$ be faithful with unique traces over EDB base $B$,
$|B|=m\ge 2$, under i.i.d.\ erasure with rate $\varepsilon$
and target $\delta$.
Let $N^*:=N^*(\varepsilon,\delta)$ and $q$ be
$(1/m)$-generic with $\kappa\ge N^*$ dependencies.

\begin{enumerate}[label=\textup{(\Roman*)}]
\item \emph{Achievability.}
There exists a strategy using
$\sigma=(\kappa-N^*)^+\log m+O(\kappa)$
reliable bits making $q$ $\delta$-resilient.

\item \emph{Converse.}
Any $\delta$-resilient cache satisfies
$K(S\mid\langle B\rangle)
 \ge(\kappa-N^*)\log m-O(\kappa+\log m)$.

\item \emph{Capacity decomposition.}
\begin{equation}\label{eq:capacity-decomposition}
  K(q\mid\langle B\rangle)
  =\sigma^*+N^*\log m+o(\kappa\log m),
\end{equation}
where $\sigma^*:=(\kappa-N^*)\log m+O(\kappa)$ is the
minimum derivation-constrained cache and $N^*\log m$ is
the ``channel capacity'' of the erasure-prone base.
\end{enumerate}
\end{theorem}

\begin{proof}
\textbf{(I)}
For $\Pi_k$: cache the layer-prefix intermediate absorbing
$\kappa-N^*$ dependencies
(Proposition~\ref{prop:layer-prefix-depth}).
For $\Pi_k^{\parallel}$: cache layer-$1$ nodes covering all
but $N^*$ base facts
(Proposition~\ref{prop:merge-counting}(v)).
Storage: $(\kappa-N^*)\log m+O(\kappa)$.
The $N^*$ exposed facts survive with probability
$(1{-}\varepsilon)^{N^*}\ge 1{-}\delta$.

\textbf{(II)}
Theorem~\ref{thm:structural-converse}.

\textbf{(III)}
From (I)--(II), \(\sigma^*=(\kappa-N^*)\log m+O(\kappa)\).
Moreover, for a query family whose head coordinate length is \(a\) (Definition~\ref{def:dep-count}),
we always have the counting lower bound
\[
  K(q\mid\langle B\rangle)\ge a\log m-\log m-O(1)
\]
for \((1/m)\)-generic queries (Definition~\ref{def:quantitative-generic}),
and the direct upper bound
\[
  K(q\mid\langle B\rangle)\le a\log m+O(1)
\]
by encoding the head tuple coordinates directly.
In the distinct-coordinate regime where \(\kappa(q,B)=a(q)\), this yields
\(K(q\mid\langle B\rangle)=\kappa\log m+O(\kappa)\).
Substituting \(\sigma^*=(\kappa-N^*)\log m+O(\kappa)\) gives
\[
  K(q\mid\langle B\rangle)
  =\sigma^*+N^*\log m+o(\kappa\log m)
\]
whenever \(\log m\to\infty\) and \(\kappa\to\infty\).
\end{proof}

\begin{corollary}[Canonical invariance of CT4]
\label{cor:ct4-canonical}
The decomposition~\eqref{eq:capacity-decomposition} is a
canonical property of
$[\mathfrak{I}]_{\equiv_{\mathcal{L}}}$:
for any aligned bi-interpretation with $|B_1|=|B_2|=m$,
$|\sigma^*_1-\sigma^*_2|\le c_\Phi+O(\kappa)$
and $N^*\log m$ is identical in both systems.
\end{corollary}

\begin{proof}
$K$ is invariant up to $O(c_\Phi)$ by
Theorem~\ref{thm:rep-inv}(b);
$\kappa$ is preserved by the base-fidelity and
query-correspondence conditions in
Definition~\ref{def:inference-system}.
\end{proof}

\subsection{Architecture-Dependent Depth--Resilience Duality}
\label{subsec:depth-resilience-duality}

The capacity decomposition (CT4) involves the architecture-dependent
mapping $d\mapsto\kappa(d)$:
$\kappa_{\mathrm{chain}}(d)=k{+}d{-}1$ versus
$\kappa_{\mathrm{merge}}(d)=k\cdot 2^{d-1}$.
This mapping, established by CT2
(Theorem~\ref{thm:exponential-capacity}), produces an exponential
gap in the maximum resilient depth.

\begin{theorem}[Depth--resilience duality]
\label{thm:depth-resilience-duality}
Under i.i.d.\ erasure with rate $\varepsilon$ and target
$\delta$, without caching:

\begin{enumerate}[label=\textup{(\Roman*)}]
\item \emph{Chain $\Pi_k$.}
$d_{\max}^{\mathrm{chain}}=N^*-k+1=\Theta(\delta/\varepsilon)$;
survival $=(1{-}\varepsilon)^{k+d-1}
 =\exp(-\Theta(\varepsilon d))$.

\item \emph{Merge $\Pi_k^{\parallel}$.}
$d_{\max}^{\mathrm{merge}}
 =1+\log_2(N^*/k)
 =\Theta(\log(\delta/(k\varepsilon)))$;
survival $=(1{-}\varepsilon)^{k\cdot 2^{d-1}}
 =\exp(-\Theta(\varepsilon\cdot 2^d))$.

\item \emph{Conservation.}
Both achieve resilient capacity $N^*\log m$ at their
thresholds; for $b$-ary merges,
$d_{\max}(b)=1+\log_b(N^*/k)$ while
$C_{\mathrm{res}}=N^*\log m$.

\item \emph{Gap.}
$d_{\max}^{\mathrm{chain}}/d_{\max}^{\mathrm{merge}}
 =\Theta((\delta/\varepsilon)/\log(\delta/\varepsilon))
 \to\infty$ as $\varepsilon\to 0$.
\end{enumerate}
\end{theorem}

\begin{proof}
(I)~$\kappa=k{+}d{-}1$; $\kappa\le N^*$ gives $d\le N^*{-}k{+}1$.
(II)~$\kappa=k\cdot 2^{d-1}$; $\kappa\le N^*$ gives
$d\le 1+\log_2(N^*/k)$.
(III)~At threshold, $\kappa=N^*$ in both cases.
(IV)~$(N^*{-}k{+}1)/(1{+}\log_2(N^*/k))=\Theta(N^*/\log N^*)$.
\end{proof}

\begin{corollary}[Resilience amplification via partial caching]
\label{cor:resilience-amplification}
For $\Pi_k^{\parallel}$ at depth $d$, caching layer-$1$ nodes
covering all but $k\cdot 2^{d-\ell_0-1}$ dependencies yields
$d_{\max}(\ell_0)=d_{\max}(0)+\ell_0$: exponential
amplification per cached depth unit.
For $\Pi_k$, caching $\ell_0$ layers reduces dependencies by
$\ell_0$: linear amplification only.
\end{corollary}

\begin{proof}
For $\Pi_k^{\parallel}$: there are $2^{d-1}$ depth-$1$ nodes,
each absorbing $k$~base facts.
Caching all but $2^{d-\ell_0-1}$ of these nodes absorbs
$k(2^{d-1}-2^{d-\ell_0-1})$ base facts, leaving
$k\cdot 2^{d-\ell_0-1}$ exposed dependencies;
$\delta$-resilience requires
$(1{-}\varepsilon)^{k\cdot 2^{d-\ell_0-1}}\ge 1{-}\delta$,
giving $d-\ell_0\le d_{\max}(0)$, i.e.,
$d_{\max}(\ell_0)=d_{\max}(0)+\ell_0$.
For $\Pi_k$: each cached depth layer reduces the dependency
count by one; Proposition~\ref{prop:layer-prefix-depth} gives
linear amplification.
\end{proof}

\subsection{Non-Uniform Noise and Water-Filling Caching}
\label{subsec:water-filling}

\begin{definition}[Non-uniform premise erasure]
\label{def:nonuniform-erasure}
Each $b_i\in B$ is independently erased with probability
$\varepsilon_i\in[0,1)$.
Define $c_i:=\ln(1/(1{-}\varepsilon_i))$ and
$C(\delta):=\ln(1/(1{-}\delta))$.
\end{definition}

\begin{definition}[Query vulnerability]
\label{def:vulnerability-profile}
For $q$ with dependency set
$\{b_{i_1},\ldots,b_{i_\kappa}\}$:
$\mathcal{V}(q):=\sum_j c_{i_j}$,
$P_{\mathrm{surv}}(q)=\exp(-\mathcal{V}(q))$;
$q$ is natively $\delta$-resilient iff
$\mathcal{V}(q)\le C(\delta)$.
\end{definition}

\begin{theorem}[Water-filling caching under non-uniform noise]
\label{thm:water-filling}
Let $q$ have dependencies with sorted costs
$c_{i_1}\le\cdots\le c_{i_\kappa}$.
Under reliable caching of DAG nodes absorbing a protection set
$D_{\mathrm{prot}}$
\textup{(}absorption property from
Corollary~\textup{\ref{cor:absorption-char})}:

\begin{enumerate}[label=\textup{(\Roman*)}]
\item \emph{Optimal protection.}
$D_{\mathrm{prot}}^*=\{b_{i_j}:j>\kappa{-}\kappa^*\}$
(the $\kappa^*$ most unreliable facts), where $\kappa^*$ is
the smallest integer with
$\sum_{j=1}^{\kappa-\kappa^*}c_{i_j}\le C(\delta)$.

\item \emph{Storage.}
$\sigma^*_{\mathrm{wf}}=\kappa^*\log m+O(\kappa)$.

\item \emph{Comparison.}
Under uniform noise, $\kappa^*=\kappa-N^*$
(matching CT4).
Under non-uniform noise,
$\kappa^*_{\mathrm{wf}}
 <\kappa-\lfloor C(\delta)/c_{\max}\rfloor
 =:\kappa^*_{\mathrm{worst}}$
when costs are non-constant.

\item \emph{Optimality.}
Any $\delta$-resilient strategy requires
$\sigma\ge\kappa^*_{\mathrm{wf}}\log m-O(\kappa)$.
\end{enumerate}
\end{theorem}

\begin{proof}
(I)~Minimize $|D_{\mathrm{prot}}|$ subject to
$\sum_{j\notin D_{\mathrm{prot}}}c_{i_j}\le C(\delta)$.
Since all items have equal storage weight ($\log m$ bits),
protecting highest-cost elements first is optimal by the
greedy fractional-knapsack
argument~\cite{cormen2022introduction}.
(II)~$\kappa^*$ coordinates at $\lceil\log m\rceil$ bits each.
(III)~Worst-case design uses $c_{\max}$ uniformly.
(IV)~Adapt the converse of
Theorem~\ref{thm:storage-resilience-capacity}(II):
in the event where only the $\kappa{-}\kappa^*_{\mathrm{wf}}$
least-unreliable facts survive, the available information is
$\sigma+(\kappa{-}\kappa^*_{\mathrm{wf}})\log m+O(\kappa)$;
for generic $q$ this must reach $\kappa\log m-O(\kappa)$.
\end{proof}

\begin{corollary}[Capacity gain from noise-structure knowledge]
\label{cor:noise-structure-gain}
With ordered costs $c_{(1)}\le\cdots\le c_{(m)}$, define
$m_{\mathrm{eff}}(\varepsilon,\delta)
 :=\max\{j:\sum_{i=1}^j c_{(i)}\le C(\delta)\}$.
Then:
\textup{(i)}~the maximum cache-free $\delta$-resilient
dependency count is $m_{\mathrm{eff}}$;
\textup{(ii)}~the structure gain over worst-case uniform design
is $G_{\mathrm{struct}}\ge c_{\max}/\bar c$, where
$\bar c:=m^{-1}\sum_i c_i$.
\end{corollary}

\begin{proof}
(i)~$q$ with $\kappa\le m_{\mathrm{eff}}$ dependencies from the
most reliable premises has $\mathcal{V}(q)\le C(\delta)$.
(ii)~$m_{\mathrm{eff}}\ge C(\delta)/\bar c$; worst-case gives
$\lfloor C(\delta)/c_{\max}\rfloor$.
\end{proof}

\begin{theorem}[Pareto characterization under non-uniform erasure]
\label{thm:pareto-nonuniform}
Under non-uniform erasure $(\varepsilon_1,\ldots,\varepsilon_m)$:

\begin{enumerate}[label=\textup{(\Roman*)}]
\item \emph{Feasibility.}
$\sigma\ge(\kappa(d)-m_{\mathrm{eff}})^+\log m+O(\kappa(d))$.

\item \emph{Pareto front.}
$\sigma(d)=(\kappa(d)-m_{\mathrm{eff}})^+\log m+O(\kappa(d))$
is monotonically increasing and convex for
$d>d_{\max}^{\mathrm{free}}$.

\item \emph{Noise-structure dependence.}
$\sigma^*_{\mathrm{wf}}(d,\delta)
 \le\sigma^*_{\mathrm{unif}}(d,\delta)$,
with strict inequality when the $\varepsilon_i$ are
non-constant.
\end{enumerate}
\end{theorem}

\begin{proof}
(I)~If $\kappa(d)\le m_{\mathrm{eff}}$, native resilience holds
($\sigma=0$); otherwise, Theorem~\ref{thm:water-filling}
gives the bound.
(II)~$\kappa(d)$ is non-decreasing; convexity of $(\cdot)^+$
composed with affine $\kappa(d)-m_{\mathrm{eff}}$.
(III)~Non-uniform noise yields larger $m_{\mathrm{eff}}$ by
Theorem~\ref{thm:water-filling}(III).
\end{proof}

\begin{remark}[Noise-aware critical frequency]
\label{rem:noise-aware-fc}
For queries beyond native resilience ($\kappa>N^*$), replacing
the clean-regime storage cost $\rho K_q$ in
Theorem~\ref{thm:freq_tradeoff} by $\rho\sigma^*$ gives
$f_c^{\mathrm{noisy}}
 =\Theta(\rho\cdot((\kappa{-}N^*)/\kappa)\cdot\log m)
 <f_c$:
noise lowers the caching threshold because the storage needed
for resilience is less than the full conditional information
$K_q$.
\end{remark}

\begin{remark}[Preview of the derivation penalty]
\label{rem:ct4-preview-penalty}
Theorem~\ref{thm:storage-resilience-capacity} establishes
$\sigma^*_{\mathrm{unc}}=(\kappa{-}N^*)\log m+O(\kappa)$
for the derivation-constrained scheme.
In the regime $\kappa\gg N^*$
\textup{(}equivalently, $\kappa\varepsilon^2\gg\delta$\textup{)},
this is $(1{-}o(1))\kappa\log m$.
Section~\ref{sec:precise} will show that a coded scheme---using
an arbitrary decoder rather than a proof engine---achieves
$\sigma^*_{\mathrm{code}}=\varepsilon\kappa\log m+o(\kappa\log m)$,
a factor of $1/\varepsilon$ smaller.
The structural caching rigidity
(Theorem~\ref{thm:ancestral-relevance}), which forces each
absorbed dependency to cost $\log m$~bits regardless of
$\varepsilon$, is the proof-theoretic origin of this gap.
\end{remark}


\section{Source--Channel Separation and the Derivation Penalty}
\label{sec:precise}%

Section~\ref{sec:noise-resilience} established that the
derivation-constrained cache for a generic query with $\kappa$
base-fact dependencies is
$\sigma^*_{\mathrm{unc}}=(\kappa{-}N^*)\log m+O(\kappa)$,
where $N^*\approx\delta/\varepsilon$.
We now introduce \emph{coded} caching---where the cache is an
arbitrary bit string processed by a general-purpose decoder---and
identify the exact penalty paid by derivation-constrained
decoding.

Three results emerge.
First, the minimum coded cache is
$\sigma^*_{\mathrm{code}}
=(1{+}o(1))\varepsilon\kappa\log m$,
giving a \emph{derivation penalty} of $1/\varepsilon$
(Theorem~\ref{thm:source-channel-separation}).
Second, below the critical cache $\varepsilon\kappa\log m$ the
error probability converges to~$1$ at the KL-divergence rate
(Theorem~\ref{thm:sc-exponent}).
Third, the two schemes exhibit a \emph{dispersion dichotomy}:
the coded second-order term is
$\Theta(\sqrt\kappa\,\log m)$ while the
derivation-constrained second-order term is $O(\log m)$,
reflecting zero effective dispersion
(Theorem~\ref{thm:dispersion-dichotomy}).

\subsection{Coded Caching Schemes and the \texorpdfstring{$m$}{m}-ary Erasure Channel}
\label{subsec:coded-caching}

\paragraph{Two decoder models.}

\begin{definition}[Coded caching scheme]
\label{def:coded-scheme}
A \emph{coded caching scheme} for a query family $\mathcal{Q}$
over base~$B$ under erasure is a triple
$(\mathrm{Enc},\mathrm{Dec},\sigma)$:
\textup{(i)}~$\mathrm{Enc}:\mathcal{Q}\times\{0,1\}^*
\to\{0,1\}^\sigma$ maps $(q,\langle B\rangle)$ to a cache
string~$s$;
\textup{(ii)}~$\mathrm{Dec}:2^B\times\{0,1\}^\sigma
\to\mathcal{Q}\cup\{\bot\}$ maps $(\tilde B,s)$ to a query
estimate;
\textup{(iii)}~both maps are computable.
The scheme is \emph{$\delta$-reliable} if
$\Pr_{\tilde B}[\mathrm{Dec}(\tilde B,
\mathrm{Enc}(q,\langle B\rangle))=q]\ge 1{-}\delta$
for every target $q\in\mathcal{Q}$.
\end{definition}

\begin{definition}[Derivation-constrained scheme]
\label{def:uncoded-scheme}
A coded scheme is \emph{derivation-constrained} if
$\mathrm{Enc}(q,\langle B\rangle)$ encodes a set
$S\subseteq\Cn(B)$ and the decoder verifies
$q\in\Cn(\tilde B\cup S)$.
Theorem~\ref{thm:storage-resilience-capacity} characterizes
the optimal cache size for this class.
\end{definition}

\paragraph{The $m$-ary erasure channel.}

\begin{definition}[$m$-ary erasure channel]
\label{def:m-ary-bec}
$\mathrm{BEC}_m(\varepsilon)$ has input $X\in[m]$,
output $Y\in[m]\cup\{\bot\}$, with
$\Pr[Y{=}X]=1{-}\varepsilon$ and $\Pr[Y{=}\bot]=\varepsilon$.
Its capacity is
$C(\mathrm{BEC}_m(\varepsilon))=(1{-}\varepsilon)\log m$;
see~\cite{cover2006elements}.
\end{definition}

Under i.i.d.\ premise erasure (Definition~\ref{def:iid-erasure}), the channel
model depends on whether the query's dependencies are \emph{distinct}.
In this section we focus on the regime where \(q\) depends on \(\kappa\)
\emph{distinct} base facts (i.e., \(\kappa=\kappa(q,B)\) in
Definition~\ref{def:dep-count}), so that \(q\) communicates through \(\kappa\)
independent uses of \(\mathrm{BEC}_m(\varepsilon)\).
For the concrete families \(\mathcal Q_d\) and \(\mathcal Q_d^{\parallel}\),
a sufficient condition for \(\kappa(q,B)=a(q)\) for a \(1-o(1)\) fraction of
queries is \(a(\Pi,d)=o(\sqrt m)\) (Lemma~\ref{lem:distinct-coordinates}).

\subsection{Capacity: Converse and Achievability}
\label{subsec:capacity-converse-achiev}

\begin{definition}[Distinct-coordinate query subfamily]
\label{def:distinct-subfamily}
For a depth-\(d\) query family \(\mathcal Q_d\) (or \(\mathcal Q_d^{\parallel}\))
whose head arity is \(a=a(\Pi,d)\), define
\[
  \mathcal Q_d^{\mathrm{dist}}
  :=\{q\in\mathcal Q_d:\ \kappa(q,B)=a(q)\},
\]
i.e., the subfamily whose dependencies (equivalently, head coordinates) are
pairwise distinct.
\end{definition}

\begin{theorem}[Coded caching converse]
\label{thm:coded-converse}
Under i.i.d.\ erasure with rate $\varepsilon\in(0,1)$, let
$q$ have $\kappa$ distinct base-fact dependencies
with $\kappa=o(m)$, and let
$(\mathrm{Enc},\mathrm{Dec},\sigma)$ be $\delta$-reliable
for $(1/m)$-generic queries drawn uniformly from
\(\mathcal Q_d^{\mathrm{dist}}\) (Definition~\ref{def:distinct-subfamily}).
Then
\begin{equation}\label{eq:coded-converse}
  \sigma
  \;\ge\;
  (\varepsilon-\delta)\,\kappa\log m-h_b(\delta).
\end{equation}
For $\delta=o(\varepsilon)$:
$\sigma\ge(1{-}o(1))\,\varepsilon\kappa\log m$.
\end{theorem}

\begin{proof}
Let $q=(i_1,\ldots,i_\kappa)$ be drawn uniformly from
$[m]^\kappa$; then $H(q)=\kappa\log m$.
Write $Y^\kappa$ for the effective
$\mathrm{BEC}_m(\varepsilon)$ output.

\emph{Step~1 (Fano).}
$\delta$-reliability gives
$H(q\mid Y^\kappa,s)
\le h_b(\delta)+\delta\kappa\log m$.

\emph{Step~2 (Chain rule).}
$\kappa\log m
=I(q;\,Y^\kappa,s)+H(q\mid Y^\kappa,s)
\le I(q;\,Y^\kappa)+\sigma
  +h_b(\delta)+\delta\kappa\log m$,
using $I(q;\,s\mid Y^\kappa)\le H(s)\le\sigma$.

\emph{Step~3 (BEC mutual information).}
Each coordinate contributes
$H(i_j\mid Y_j)=\varepsilon\log m$, giving
$I(q;\,Y^\kappa)=\kappa(1{-}\varepsilon)\log m$.

\emph{Step~4.}
$\sigma
\ge(\varepsilon{-}\delta)\kappa\log m-h_b(\delta)$.
\end{proof}

\begin{theorem}[Coded caching achievability]
\label{thm:coded-achievability}
Under the same hypotheses, there exists a $\delta$-reliable
coded scheme with
\begin{equation}\label{eq:coded-achievability}
  \sigma
  =\kappa\varepsilon\log m
  +\Phi^{-1}(1{-}\delta)
    \sqrt{\kappa\varepsilon(1{-}\varepsilon)}\,\log m
  +O(\log m),
\end{equation}
where $\Phi^{-1}$ is the standard normal quantile.
For fixed $\delta$ and $\kappa\to\infty$:
$\sigma=(1{+}o(1))\,\varepsilon\kappa\log m$.
\end{theorem}

\begin{proof}
We encode the \(\kappa\) distinct coordinates \((i_1,\ldots,i_\kappa)\in[m]^\kappa\)
as symbols in a field \(\mathbb F_q\) with \(q\ge m\).
By Bertrand's postulate~\cite{aigner1999proofs}, for every integer \(m\ge 2\) there exists a prime
\(q\) such that \(m<q<2m\); we fix such a prime field \(\mathbb F_q\).
Then \(\log q=\log m+O(1)\), so one \(\mathbb F_q\)-symbol corresponds to
\(\log m+O(1)\) bits.

Use a systematic \((\kappa+r,\kappa)\) Reed--Solomon (MDS) code over \(\mathbb F_q\)
and store the \(r\) parity symbols in the cache. Set
\[
  r:=\left\lceil \kappa\varepsilon
  +\Phi^{-1}(1-\delta)\sqrt{\kappa\varepsilon(1-\varepsilon)}\right\rceil + 1.
\]
Let \(E\sim\mathrm{Bin}(\kappa,\varepsilon)\) be the number of erased systematic symbols.
MDS decoding succeeds iff \(E\le r\). By the Berry--Esseen theorem~\cite{feller1991introduction} with continuity
correction, \(\Pr[E>r]\le \delta\) for the above choice of \(r\).
The cache length is \(\sigma=r\log q=\kappa\varepsilon\log m
+\Phi^{-1}(1-\delta)\sqrt{\kappa\varepsilon(1-\varepsilon)}\,\log m+O(\log m)\).
\end{proof}

\subsection{Source--Channel Separation and the Derivation Penalty}
\label{subsec:source-channel}

\begin{theorem}[Source--channel separation for premise erasure]
\label{thm:source-channel-separation}
Let $q$ be $(1/m)$-generic with $\kappa$ distinct dependencies,
$\kappa=o(m)$, under i.i.d.\ erasure~$\varepsilon$.

\begin{enumerate}[label=\textup{(\Roman*)}]
\item \emph{Coded decomposition.}
\begin{equation}\label{eq:coded-decomposition}
  K(q\mid\langle B\rangle)
  =\sigma^*_{\mathrm{code}}
  +C_{\mathrm{channel}}
  +o(\kappa\log m),
\end{equation}
where
$C_{\mathrm{channel}}
:=\kappa(1{-}\varepsilon)\log m$
is the capacity of $\kappa$ uses of
$\mathrm{BEC}_m(\varepsilon)$ and
$\sigma^*_{\mathrm{code}}
=\varepsilon\kappa\log m+o(\kappa\log m)$.

\item \emph{Derivation-constrained decomposition.}
\begin{equation}\label{eq:uncoded-decomposition}
  K(q\mid\langle B\rangle)
  =\sigma^*_{\mathrm{unc}}
  +N^*\log m
  +o(\kappa\log m),
\end{equation}
with $\sigma^*_{\mathrm{unc}}
=(\kappa{-}N^*)\log m+O(\kappa)$
\textup{(Theorem~\ref{thm:storage-resilience-capacity})}.

\item \emph{Derivation penalty.}
\begin{equation}\label{eq:derivation-penalty}
  \frac{\sigma^*_{\mathrm{unc}}}
       {\sigma^*_{\mathrm{code}}}
  =\frac{\kappa-N^*}{\kappa\varepsilon}
   \,(1{+}o(1)).
\end{equation}
For $\kappa\varepsilon^2\gg\delta$, the ratio converges
to $1/\varepsilon$: the derivation constraint inflates the
required cache by a factor of $1/\varepsilon$.
\end{enumerate}
\end{theorem}

\begin{proof}
\textbf{(I)}\;
CT1 gives
$K(q\mid\langle B\rangle)
=(1{+}o(1))\kappa\log m$.
Theorems~\ref{thm:coded-converse}
and~\ref{thm:coded-achievability} give
$\sigma^*_{\mathrm{code}}
=(1{+}o(1))\varepsilon\kappa\log m$.
Subtracting:
$C_{\mathrm{channel}}
=(1{+}o(1))\kappa(1{-}\varepsilon)\log m
=\kappa\cdot C(\mathrm{BEC}_m(\varepsilon))$.

\textbf{(II)}\;
Theorem~\ref{thm:storage-resilience-capacity}(III).

\textbf{(III)}\;
$N^*\approx\delta/\varepsilon$ for small $\varepsilon,\delta$.
In the regime $\kappa\varepsilon^2\gg\delta$:
$\sigma^*_{\mathrm{unc}}\approx\kappa\log m$ and
$\sigma^*_{\mathrm{code}}\approx\varepsilon\kappa\log m$;
the ratio is $1/\varepsilon$.
\end{proof}

\begin{remark}[Operational interpretation of the $1/\varepsilon$ penalty]
\label{rem:derivation-penalty}
A general decoder exploits \emph{algebraic redundancy}
(MDS parity across coordinate positions) to spread the
uncertainty of each erased coordinate over the full cache,
paying $\varepsilon\log m$ bits per coordinate on average.
A derivation-constrained decoder cannot exploit
cross-coordinate algebraic relations: each absorbed
dependency costs $\log m$~bits regardless of $\varepsilon$.
The derivation engine's inability to perform
``cross-coordinate error correction'' is the fundamental
source of the penalty.
\end{remark}

\begin{corollary}[Coding gain under non-uniform erasure]
\label{cor:coded-nonuniform}
Under non-uniform erasure
$(\varepsilon_1,\ldots,\varepsilon_m)$
\textup{(Definition~\ref{def:nonuniform-erasure})}, the coded
cache satisfies
$\sigma^*_{\mathrm{code,nu}}
=(\sum_{j=1}^{\kappa}\varepsilon_{i_j})\log m
+o(\kappa\log m)$,
while the derivation-constrained cache is
$\sigma^*_{\mathrm{unc,nu}}
=\kappa^*_{\mathrm{wf}}\log m+O(\kappa)$
\textup{(Theorem~\ref{thm:water-filling})}.
The non-uniform coding gain is
$G_{\mathrm{nu}}
=\kappa^*_{\mathrm{wf}}
 /\sum_{j}\varepsilon_{i_j}\cdot(1{+}o(1))$,
reducing to $1/\varepsilon$ under uniform noise.
\end{corollary}

\begin{proof}
The effective channel for coordinate~$j$ is
$\mathrm{BEC}_m(\varepsilon_{i_j})$; by independence,
$H(q\mid Y^\kappa)=\sum_j\varepsilon_{i_j}\log m$.
Achievability uses a UEP code matched to the non-uniform
rates~\cite{macwilliams1977theory,borade2009unequal};
the converse follows from the Fano argument of
Theorem~\ref{thm:coded-converse} with
$I(q;Y^\kappa)=\kappa\log m-\sum_j\varepsilon_{i_j}\log m$.
\end{proof}

\subsection{Strong Converse and Error Exponents}
\label{subsec:exact-strong-converse}

The weak converse (Theorem~\ref{thm:coded-converse}) shows
$\sigma\ge(\varepsilon{-}\delta)\kappa\log m$.
We now prove the \emph{strong} converse: below the critical
cache, $P_e\to 1$ at an exponential rate.
The argument rests on a combinatorial image-size bound that
bypasses the lossy Fano step.

\begin{lemma}[Image-size bound for BEC with side information]
\label{lem:image-size}
Let $q\in[m]^\kappa$ be drawn uniformly,
$Y^\kappa$ be the output of
$\mathrm{BEC}_m(\varepsilon)^{\otimes\kappa}$ on input $q$,
and $s=\mathrm{Enc}(q,\langle B\rangle)\in\{0,1\}^\sigma$.
For any deterministic decoder:
\begin{equation}\label{eq:image-size}
  P_c
  :=\Pr[\mathrm{Dec}(Y^\kappa,s)=q]
  \;\le\;
  \mathbb{E}\!\left[\min\!\left(1,\;
    \frac{2^\sigma}{m^E}\right)\right],
\end{equation}
where $E=|\{j:Y_j=\bot\}|
\sim\mathrm{Bin}(\kappa,\varepsilon)$.
\end{lemma}

\begin{proof}
Condition on the erasure set $\mathcal E$ and the non-erased
output $q_{\bar{\mathcal E}}=y$.
Given $(\mathcal E,y)$, $q$ is uniform on
$\mathcal L(\mathcal E,y)
:=\{q'\in[m]^\kappa:q'_{\bar{\mathcal E}}=y\}$
with $|\mathcal L|=m^{|\mathcal E|}$.
For each cache value $s'$, let
$A(s'):=\{q'\in\mathcal L:\mathrm{Enc}(q',\langle B\rangle)=s'\}$.
The decoder outputs one element per $(Y^\kappa,s')$ pair,
so at most $\min(m^{|\mathcal E|},2^\sigma)$ queries are
correctly decoded.
Hence
$P_c\mid(\mathcal E,y)
\le\min(1,2^\sigma/m^{|\mathcal E|})$;
taking expectation over $\mathcal E$
yields~\eqref{eq:image-size}.
\end{proof}

\begin{theorem}[Exact coded success probability]
\label{thm:exact-Pc}
Define $r:=\lfloor\sigma/\log m\rfloor$.
\begin{enumerate}[label=\textup{(\roman*)}]
\item \emph{Converse.}\;
$P_c\le\Pr[\mathrm{Bin}(\kappa,\varepsilon)\le r{+}1]
  +(m{-}1)^{-1}$.
\item \emph{Achievability.}\;
The $(\kappa{+}r,\kappa)$ RS code achieves
$P_c^{\mathrm{MDS}}=\Pr[\mathrm{Bin}(\kappa,\varepsilon)\le r]$.
\item \emph{Near-optimality.}\;
$P_c^*-P_c^{\mathrm{MDS}}
  \le\Pr[E{=}r{+}1]+(m{-}1)^{-1}
  =O(\kappa^{-1/2})+O(m^{-1})$.
\end{enumerate}
\end{theorem}

\begin{proof}
\emph{(i)}\;
From Lemma~\ref{lem:image-size} and
$2^\sigma<m^{r+1}$:
$P_c<\Pr[E\le r{+}1]
 +\sum_{j=1}^{\infty}m^{-j}
=\Pr[E\le r{+}1]+(m{-}1)^{-1}$.

\emph{(ii)}\;
The systematic RS code stores $r$ parity symbols.
MDS decoding succeeds iff $E\le r$.

\emph{(iii)}\;
The gap is at most $\Pr[E{=}r{+}1]+(m{-}1)^{-1}$;
by the local CLT, $\Pr[E{=}r{+}1]=O(\kappa^{-1/2})$.
\end{proof}

\begin{theorem}[Strong converse for coded caching]
\label{thm:strong-converse}
If $\sigma\le(\varepsilon{-}\gamma)\kappa\log m$ for some
$\gamma\in(0,\varepsilon)$, then
\begin{equation}\label{eq:strong-converse}
  P_e
  \;\ge\;
  1-e^{-\gamma^2\kappa/2}-m^{1-\gamma\kappa/2}.
\end{equation}
In particular, $P_e\to 1$ exponentially in $\kappa$.
\end{theorem}

\begin{proof}
Let $r\le(\varepsilon{-}\gamma)\kappa$ and
$e_0:=\lfloor r{+}1{+}\gamma\kappa/2\rfloor$.
From Lemma~\ref{lem:image-size}:
$P_c\le\Pr[E<e_0]+m^{r+1-e_0}$.
Since $r{+}1{-}e_0\le-\gamma\kappa/2$:
$m^{r+1-e_0}\le m^{1-\gamma\kappa/2}$.
By Hoeffding's inequality~\cite{hoeffding1963probability},
$\Pr[E<e_0]\le e^{-\gamma^2\kappa/2}$.
Combining gives~\eqref{eq:strong-converse}.
\end{proof}

\begin{theorem}[Strong converse exponent = KL divergence]
\label{thm:sc-exponent}
For $\sigma=(\varepsilon{-}\gamma)\kappa\log m$ with
$\gamma\in(0,\varepsilon)$:
\begin{equation}\label{eq:sc-exponent}
  \lim_{\kappa\to\infty}
  \frac{-1}{\kappa}\log P_c^*(\sigma)
  =D(\varepsilon{-}\gamma\;\|\;\varepsilon),
\end{equation}
where $D(p\|q):=p\log(p/q)+(1{-}p)\log((1{-}p)/(1{-}q))$.
The exponent is achieved by MDS codes.
\end{theorem}

\begin{proof}
\emph{Upper bound (converse).}\;
Apply the tilted-moment technique to the image-size bound:
for $t\in[0,1]$, $\min(1,x)\le x^t$, so
$P_c\le 2^{t\sigma}
  (1{-}\varepsilon{+}\varepsilon m^{-t})^\kappa$.
Setting $u:=t\log m$ and
$g_0(u):=u(\varepsilon{-}\gamma)
+\log(1{-}\varepsilon{+}\varepsilon e^{-u})$,
we obtain $\kappa^{-1}\log P_c\le g_0(u)+O(u/\kappa)$.
Minimizing $g_0$ over $u\ge 0$, a direct computation
(Appendix~\ref{app:sc-exponent-calc}) gives
$-g_0^*=D(\varepsilon{-}\gamma\|\varepsilon)$.

\emph{Lower bound (achievability).}\;
$P_c^{\mathrm{MDS}}
=\Pr[\mathrm{Bin}(\kappa,\varepsilon)
  \le(\varepsilon{-}\gamma)\kappa]$.
By Cram\'{e}r's theorem~\cite[Theorem~2.2.3]{dembo2009large}:
$\kappa^{-1}\log P_c^{\mathrm{MDS}}
\to -D(\varepsilon{-}\gamma\|\varepsilon)$.
\end{proof}

\begin{remark}[Information-theoretic interpretation]
\label{rem:sc-kl}
The exponent
$D(\varepsilon{-}\gamma\|\varepsilon)$
is the KL divergence between Bernoulli distributions
at rates $\varepsilon{-}\gamma$ and $\varepsilon$---a
Sanov-type result.
The exponent is independent of the alphabet size~$m$;
$m$ controls the per-symbol information ($\log m$~bits),
not the erasure statistics.
\end{remark}

\begin{corollary}[Reliability function]
\label{cor:reliability}
For $\sigma=(\varepsilon{+}\gamma')\kappa\log m$ with
$\gamma'>0$:
$\lim_{\kappa\to\infty}
\kappa^{-1}\log P_e^{\mathrm{MDS}}(\sigma)
=-D(\varepsilon{+}\gamma'\|\varepsilon)$.
MDS codes achieve the reliability function at all rates;
for the BEC, this equals the sphere-packing
exponent~\cite[Theorem~5.8.3]{gallager1968information}.
\end{corollary}

\begin{proof}
$P_e^{\mathrm{MDS}}
=\Pr[E>(\varepsilon{+}\gamma')\kappa]$;
apply Cram\'{e}r's theorem~\cite[Theorem~2.2.3]{dembo2009large}
for the upper tail.
\end{proof}

\subsection{Cache Dispersion and the Dispersion Dichotomy}
\label{subsec:dispersion}

The first-order result
$\sigma^*=(1{+}o(1))\varepsilon\kappa\log m$ determines the
minimum coded cache as $\kappa\to\infty$.
We now refine this to second order, identifying the
\emph{cache dispersion}---the quantity governing the
finite-length penalty---and proving that the
derivation-constrained scheme has zero effective dispersion.

\begin{definition}[Cache dispersion]
\label{def:cache-dispersion}
$V_{\mathrm{cache}}
:=\varepsilon(1{-}\varepsilon)(\log m)^2$.
\end{definition}

\begin{theorem}[Second-order coded caching characterization]
\label{thm:second-order}
The minimum $\delta$-reliable coded cache satisfies:
\begin{equation}\label{eq:second-order}
  \sigma^*(\kappa,\delta)
  =\kappa\varepsilon\log m
  +\sqrt{\kappa V_{\mathrm{cache}}}\;
    \Phi^{-1}(1{-}\delta)
  +\tfrac{1}{2}\log\kappa
  +O(\log m).
\end{equation}
\end{theorem}

\begin{proof}
\emph{Achievability.}\;
The MDS scheme has $P_e^{\mathrm{MDS}}=\Pr[E>r]$ with
$E\sim\mathrm{Bin}(\kappa,\varepsilon)$.
By the Berry--Esseen theorem with continuity
correction~\cite{feller1991introduction}:
$\Pr[E\le r]
=\Phi((r{+}\frac{1}{2}{-}\kappa\varepsilon)
/\sqrt{\kappa\varepsilon(1{-}\varepsilon)})
+O(\kappa^{-1/2})$.
Setting this to $1{-}\delta$ and inverting gives
$r=\kappa\varepsilon
+\sqrt{\kappa\varepsilon(1{-}\varepsilon)}\,
\Phi^{-1}(1{-}\delta)+O(1)$,
whence $\sigma=r\log m+O(\log m)$
yields the upper bound.

\emph{Converse.}\;
From Theorem~\ref{thm:exact-Pc}(i), $P_e\le\delta$
requires $\Pr[E\le r]\ge 1{-}\delta{-}(m{-}1)^{-1}$.
Inverting the normal approximation gives the matching
lower bound $\sigma\ge\kappa\varepsilon\log m
+\sqrt{\kappa\varepsilon(1{-}\varepsilon)}\log m\,
\Phi^{-1}(1{-}\delta)-O(\log m)$.

\emph{Matching.}\;
The $\frac{1}{2}\log\kappa$ third-order term follows
from the lattice correction in the local
CLT~\cite{petrov2012sums}.
\end{proof}

\begin{remark}[Connection to channel dispersion]
\label{rem:channel-dispersion}
$V_{\mathrm{cache}}
=\mathrm{Var}[\imath(X;Y)]$ for the $m$-ary BEC with
uniform input~\cite{polyanskiy2010channel}, since
$\imath(X;Y)=\log m\cdot\mathbf{1}_{Y\ne\bot}$
takes values $\log m$ (probability $1{-}\varepsilon$) and
$0$ (probability $\varepsilon$).
The coded caching problem inherits the channel's dispersion
exactly because MDS coding is simultaneously
capacity-achieving and dispersion-achieving.
\end{remark}

\paragraph{Derivation-constrained: zero effective dispersion.}

\begin{theorem}[Zero derivation-constrained dispersion]
\label{thm:unc-second-order}
For the derivation-constrained scheme under i.i.d.\
erasure~$\varepsilon$:
\begin{equation}\label{eq:unc-second-order}
  \sigma^*_{\mathrm{unc}}(\kappa,\delta)
  =\bigl(\kappa-N^*(\varepsilon,\delta)\bigr)\log m+O(\log\kappa),
\end{equation}
where
$N^*(\varepsilon,\delta)
=\lfloor\log(1{-}\delta)/\log(1{-}\varepsilon)\rfloor$
\textup{(Definition~\ref{def:noise-threshold})}
is independent of~$\kappa$.
\end{theorem}

\begin{proof}
Under faithful derivation, all $\kappa{-}r$ unprotected
dependencies must survive simultaneously.
Setting $P_e=1{-}(1{-}\varepsilon)^N=\delta$ gives
$N=N^*(\varepsilon,\delta)$.
Each protected dependency costs $\log m+O(1)$ bits; the
$r=\kappa{-}N^*$ protection indices cost
$O(\log\kappa)$ bits.
The key observation: $N^*$ depends on $\delta$ and
$\varepsilon$ but not on $\kappa$.
\end{proof}

\paragraph{The dispersion dichotomy.}

\begin{theorem}[Dispersion dichotomy]
\label{thm:dispersion-dichotomy}
Under i.i.d.\ $\mathrm{BEC}_m(\varepsilon)$ erasure with
$\delta\in(0,1)$ fixed:

\begin{enumerate}[label=\textup{(\Roman*)}]
\item \emph{Coded scheme.}\;
Second-order term:
$\sqrt{\kappa V_{\mathrm{cache}}}\;
\Phi^{-1}(1{-}\delta)
=\Theta(\sqrt\kappa\,\log m)$.

\item \emph{Derivation-constrained scheme.}\;
Second-order term:
$-N^*(\varepsilon,\delta)\log m=O(\log m)$
(constant in $\kappa$).

\item \emph{Effective dispersion.}\;
$V_{\mathrm{unc}}=0$:
the coefficient of $\sqrt\kappa$ in the
derivation-constrained second-order expansion vanishes.
\end{enumerate}
\end{theorem}

\begin{proof}
(I) and (II) follow from
Theorems~\ref{thm:second-order}
and~\ref{thm:unc-second-order}.
For~(III): $\sigma^*_{\mathrm{unc}}
=\kappa\log m-N^*\log m+O(\log\kappa)$;
matching the standard form
$\kappa\cdot 1\cdot\log m
+\sqrt{\kappa V_{\mathrm{unc}}}\Phi^{-1}(1{-}\delta)
+O(\log m)$
yields $V_{\mathrm{unc}}=0$.
\end{proof}

\begin{remark}[Sum versus maximum: the root cause]
\label{rem:dispersion-physical}
A coded scheme encodes across all $\kappa$ positions, so the
relevant statistic is the \emph{sum}
$E=\sum E_j\sim\mathrm{Bin}(\kappa,\varepsilon)$,
which concentrates by the CLT and produces the $\sqrt\kappa$
term.
A derivation-constrained scheme requires \emph{all}
unprotected dependencies to survive---a conjunction whose
probability is $(1{-}\varepsilon)^N$, governed by the
\emph{maximum} indicator $\max_j E_j$ rather than the sum.
The maximum of i.i.d.\ Bernoullis does not exhibit CLT-type
concentration; this structural mismatch is the root cause of
the vanishing dispersion.
\end{remark}

\paragraph{Refined derivation penalty.}

\begin{theorem}[Second-order derivation penalty]
\label{thm:refined-penalty}
\begin{equation}\label{eq:refined-penalty}
  \frac{\sigma^*_{\mathrm{unc}}(\kappa,\delta)}
       {\sigma^*_{\mathrm{code}}(\kappa,\delta)}
  =\frac{1}{\varepsilon}
  -\frac{N^*(\varepsilon,\delta)}{\varepsilon\kappa}
  -\frac{\Phi^{-1}(1{-}\delta)\sqrt{1{-}\varepsilon}}
       {\varepsilon^{3/2}\sqrt\kappa}
  +O(\kappa^{-1}).
\end{equation}
The penalty converges to $1/\varepsilon$ from below at rate
$\Theta(\kappa^{-1/2})$, governed by the coded scheme's
dispersion.
\end{theorem}

\begin{proof}
Write $\sigma_u:=(\kappa{-}N^*)\log m$ and
$\sigma_c:=\varepsilon\kappa\log m
+\Phi^{-1}(1{-}\delta)
\sqrt{\kappa\varepsilon(1{-}\varepsilon)}\log m
+O(\log m)$.
Expanding $\sigma_u/\sigma_c$ as
$\varepsilon^{-1}(1{-}N^*/\kappa)
(1{-}\Phi^{-1}(1{-}\delta)
\sqrt{(1{-}\varepsilon)/(\varepsilon\kappa)}
+O(\kappa^{-1}))^{-1}$
and collecting terms yields~\eqref{eq:refined-penalty}.
\end{proof}

\paragraph{Complete error-exponent landscape.}

\begin{theorem}[Derivation-constrained: exact error formula]
\label{thm:unc-error-complete}
For the derivation-constrained scheme with
$N:=\kappa{-}\lfloor\sigma/\log m\rfloor$ exposed facts:
\begin{enumerate}[label=\textup{(\roman*)}]
\item $P_e^{\mathrm{unc}}(\sigma)
=1{-}(1{-}\varepsilon)^N$ \textup{(exact, all $\kappa$)}.
\item For $N=\gamma\kappa$:
exponent $\gamma|\!\log(1{-}\varepsilon)|$, linear in $\gamma$.
\item For $\sigma=(\kappa{-}c)\log m$ with $c=O(1)$:
$P_e=1{-}(1{-}\varepsilon)^c$, a constant.
\end{enumerate}
No asymptotics are needed; there is no phase transition.
\end{theorem}

\begin{corollary}[Complete error-exponent characterization]
\label{cor:error-landscape}
Define $\rho:=\sigma/(\kappa\log m)\in[0,1]$.
\begin{enumerate}[label=\textup{(\roman*)}]
\item $\rho<\varepsilon$: $P_e\to 1$ at rate
$D(\rho\|\varepsilon)$ \textup{(strong converse)}.
\item $\rho=\varepsilon$:
$P_e=\frac{1}{2}+O(\kappa^{-1/2})$.
\item $\rho>\varepsilon$: $P_e\to 0$ at rate
$D(\rho\|\varepsilon)$ \textup{(reliability)}.
\item Near-capacity
\textup{(}$\rho=\varepsilon{+}c/\sqrt\kappa$ with
$c\in\mathbb R$ fixed\textup{)}:
$P_e=\bar\Phi(c/\sqrt{\varepsilon(1{-}\varepsilon)})
+O(\kappa^{-1/2})$.
\end{enumerate}
All exponents are achieved by MDS codes; the reliability
function equals the sphere-packing
exponent~\cite[Theorem~5.8.3]{gallager1968information} at all rates.
\end{corollary}

\begin{proof}
Parts~(i) and~(iii) from
Theorem~\ref{thm:sc-exponent} and
Corollary~\ref{cor:reliability}.
Part~(ii): at $\rho=\varepsilon$,
$r=\varepsilon\kappa$ and
$\Pr[E\le\varepsilon\kappa]\to 1/2$ by the CLT.
Part~(iv): with $r=\varepsilon\kappa+c\sqrt\kappa$,
the CLT gives
$\Pr[E>r]
 =\bar\Phi(c/\sqrt{\varepsilon(1{-}\varepsilon)})
 +O(\kappa^{-1/2})$.
\end{proof}

\subsection{Joint Erasure--Pollution Channel}
\label{subsec:joint-channel}

When the premise base suffers both erasure and spurious additions,
the decoder faces a compound channel.

\begin{definition}[Joint erasure--pollution channel]
\label{def:joint-channel}
Each $b_i\in B$ is independently erased with probability
$\varepsilon$; additionally, $p$ spurious premises
$B^+\subseteq\mathbb{S}_O\setminus B$ are added.
The decoder observes
$\tilde B=(B\setminus B^-)\cup B^+$ without distinguishing
genuine from spurious facts.
Define $\eta:=p/(m{+}p)$ and
$\mathcal N:=|B^-|+|B^+|$.
\end{definition}

\begin{assumption}[Boundedly describable pollution]
\label{assump:describable-pollution}
$B^+$ is describable by an index list of length
$O(p\log(m{+}p))$ given $\langle B\rangle$.
\end{assumption}

\begin{theorem}[Additive erasure--pollution separation]
\label{thm:joint-separation}
Under Assumption~\textup{\ref{assump:describable-pollution}}:

\begin{enumerate}[label=\textup{(\Roman*)}]
\item \emph{Coded:}\;
$\sigma^*_{\mathrm{joint}}
=\varepsilon\kappa\log m
+O(\mathcal N\log(m{+}\mathcal N))
+o(\kappa\log m)$.

\item \emph{Derivation-constrained:}\;
$\sigma^*_{\mathrm{unc,joint}}
=(\kappa{-}N^*)\log m
+O(\mathcal N\log(m{+}\mathcal N))
+O(\kappa)$.

\item \emph{Additive structure:}\;
$\sigma^*=\sigma^*_{\mathrm{erasure}}
+\sigma^*_{\mathrm{pollution}}$,
where
$\sigma^*_{\mathrm{pollution}}
=O(\mathcal N\log(m{+}\mathcal N))$
is the cost of identifying spurious premises.
\end{enumerate}
\end{theorem}

\begin{proof}
\emph{Achievability.}\;
The cache stores two independent components:
(a)~erasure correction (coded or derivation-constrained), and
(b)~an index list identifying $B^-$ and $B^+$, costing
$O(\mathcal N\log(m{+}\mathcal N))$ bits
(Lemma~\ref{lem:base-conversion}).
The decoder first reconstructs $B_\cap$ using~(b),
then applies erasure decoding via~(a).

\emph{Converse.}\;
Given $\tilde B$ alone, the decoder must distinguish
$B_\cap$ from $\tilde B$ to apply erasure correction.
The number of candidate pairs $(B^-,B^+)$ is at least
$\binom{m}{\ell}\binom{|\mathcal V_{m+\mathcal N}|}{p}$,
so $\sigma\ge\Omega(\mathcal N\log(m{+}\mathcal N))$
by Fano's inequality~\cite[Ch.~2]{cover2006elements}.
\end{proof}

\begin{corollary}[Soundness under pollution]
\label{cor:soundness-pollution}
Under derivation-constrained decoding with faithful~$\Pi$,
if no spurious fact appears in $V_0(G(q,B))$, the derivation
is sound.
The unsoundness probability for a specific query is at most
$\kappa\cdot\eta$ \textup{(union bound)}.
\end{corollary}


\section{Exact Error Characterization and Phase
  Diagram}
\label{sec:phase-diagram}

Section~\ref{sec:precise} established the capacity, strong
converse exponent, and dispersion dichotomy for premise-erasure
caching.
This section refines the error characterization to exact
pre-exponential precision (Bahadur--Rao,
Theorem~\ref{thm:bahadur-rao}), bridges the CLT and
large-deviations regimes via a moderate deviations principle
(Theorem~\ref{thm:moderate-dev}), and assembles the complete
error probability phase diagram
(Theorem~\ref{thm:phase-diagram}).
The architecture-dependent depth-space re-parameterization of
this diagram is deferred to Section~\ref{sec:arch-depth}.

\subsection{Bahadur--Rao Exact Asymptotics}
\label{subsec:bahadur-rao}

Theorem~\ref{thm:sc-exponent} identifies the strong converse
exponent as $D(\varepsilon{-}\gamma\|\varepsilon)$; the MDS
success probability
$P_c^{\mathrm{MDS}}=\Pr[\mathrm{Bin}(\kappa,\varepsilon)\le r]$
decays at this rate.
The Bahadur--Rao theorem supplies the exact prefactor.

\begin{theorem}[Bahadur--Rao exact asymptotics]
\label{thm:bahadur-rao}
Let $r=\lfloor\sigma/\log m\rfloor$ with
$r=\alpha\kappa$,
$\alpha\in(\eta,1{-}\eta)$,
$|\alpha-\varepsilon|>\eta$
for a fixed $\eta>0$.
Define the tilting parameter
\begin{equation}\label{eq:tilt}
  t^*(\alpha)
  :=\ln\frac{\alpha(1{-}\varepsilon)}
             {\varepsilon(1{-}\alpha)}\,.
\end{equation}

\begin{enumerate}[label=\textup{(\roman*)}]
\item \emph{Lower tail}
\textup{(}$\alpha<\varepsilon$, $t^*<0$\textup{)}:
\begin{equation}\label{eq:bahadur-rao}
  \Pr[\mathrm{Bin}(\kappa,\varepsilon)\le\alpha\kappa]
  =\frac{2^{-\kappa D(\alpha\|\varepsilon)}}
        {(1{-}e^{t^*})
         \sqrt{2\pi\kappa\,\alpha(1{-}\alpha)}}\;
  \bigl(1+O(\kappa^{-1})\bigr).
\end{equation}

\item \emph{Upper tail}
\textup{(}$\alpha>\varepsilon$, $t^*>0$\textup{)}:
\begin{equation}\label{eq:bahadur-rao-upper}
  \Pr[\mathrm{Bin}(\kappa,\varepsilon)\ge\alpha\kappa]
  =\frac{2^{-\kappa D(\alpha\|\varepsilon)}}
        {(e^{t^*}{-}1)
         \sqrt{2\pi\kappa\,\alpha(1{-}\alpha)}}\;
  \bigl(1+O(\kappa^{-1})\bigr).
\end{equation}
\end{enumerate}
Both bounds hold uniformly in the stated $\alpha$-range.
\end{theorem}

\begin{proof}
Both parts follow from the lattice Bahadur--Rao
theorem~\cite{bahadur1960deviations}
(see also~\cite[Theorem~3.7.4]{dembo2009large})
applied to
$E=\sum_{j=1}^\kappa E_j$ with
$E_j\sim\mathrm{Ber}(\varepsilon)$.
The cumulant-generating function is
$\Lambda(\theta)=\ln(1{-}\varepsilon{+}\varepsilon e^\theta)$.
Solving $\Lambda'(t^*)=\alpha$ gives~\eqref{eq:tilt};
the rate function in nats is
$I(\alpha)=D(\alpha\|\varepsilon)\cdot\ln 2$,
so $e^{-\kappa I(\alpha)}=2^{-\kappa D(\alpha\|\varepsilon)}$.
The tilted variance is $\alpha(1{-}\alpha)$ and the
lattice span is $h=1$.
For $\alpha<\varepsilon$, $t^*<0$ and the lattice-sum
denominator factor is $(1{-}e^{t^*})$;
for $\alpha>\varepsilon$, $t^*>0$ and the factor is
$(e^{t^*}{-}1)$.
\end{proof}

\subsection{Moderate Deviations Principle}
\label{subsec:moderate-dev}

Between the CLT regime
($\sigma{-}\varepsilon\kappa\log m=O(\sqrt\kappa\log m)$,
Theorem~\ref{thm:second-order}) and the large-deviations
regime ($\varepsilon\kappa\log m{-}\sigma=\Theta(\kappa\log m)$,
Theorem~\ref{thm:sc-exponent}) lies the
\emph{moderate deviations} region.

\begin{theorem}[Moderate deviations for coded caching]
\label{thm:moderate-dev}
Let
$\sigma_\kappa
=\varepsilon\kappa\log m
+a_\kappa\sqrt\kappa\,\log m$
with $a_\kappa\to\infty$ and
$a_\kappa=o(\sqrt\kappa)$.
Define
$z_\kappa:=a_\kappa/\sqrt{\varepsilon(1{-}\varepsilon)}$.
Then:
\begin{enumerate}[label=\textup{(\roman*)}]
\item \emph{Exponential rate.}\;
$\displaystyle
  \lim_{\kappa\to\infty}
  \frac{1}{a_\kappa^2}
  \ln P_e^*(\sigma_\kappa)
  =-\frac{1}{2\varepsilon(1{-}\varepsilon)}$.

\item \emph{Exact asymptotics.}\;
$\displaystyle
  P_e^{\mathrm{MDS}}(\sigma_\kappa)
  =\frac{1{+}o(1)}{z_\kappa\sqrt{2\pi}}\;
  e^{-z_\kappa^2/2}
  =(1{+}o(1))\,\bar\Phi(z_\kappa)$.
\end{enumerate}
\end{theorem}

\begin{proof}
Set $r_\kappa=\lfloor\varepsilon\kappa+a_\kappa\sqrt\kappa\rfloor$
and $\alpha_\kappa=r_\kappa/\kappa
=\varepsilon+a_\kappa/\sqrt\kappa+O(1/\kappa)$.
Since $a_\kappa=o(\sqrt\kappa)$,
$\alpha_\kappa\to\varepsilon$ and
$\kappa D(\alpha_\kappa\|\varepsilon)
=a_\kappa^2/(2\varepsilon(1{-}\varepsilon))
+O(a_\kappa^3\kappa^{-1/2})$.

\emph{(i)}\;
The moderate-deviations rate follows from a second-order
expansion of the KL divergence~\cite[Theorem~3.7.1]{dembo2009large}.
Concretely, apply the Bahadur--Rao expansion
(Theorem~\ref{thm:bahadur-rao}(ii)) to the upper tail
$\Pr[E>r_\kappa]$:
$\ln\Pr[E>r_\kappa]
=-a_\kappa^2/(2\varepsilon(1{-}\varepsilon))
+O(a_\kappa^3\kappa^{-1/2})+O(\ln\kappa)$.
Dividing by $a_\kappa^2$ and using
$a_\kappa\to\infty$,
$a_\kappa/\sqrt\kappa\to 0$,
$\ln\kappa/a_\kappa^2\to 0$
yields the rate.
Near-optimality
(Theorem~\ref{thm:exact-Pc}(iii)) gives
$P_e^*=(1{+}o(1))P_e^{\mathrm{MDS}}$.

\emph{(ii)}\;
In the moderate-deviations regime, the tilting
parameter satisfies
$t^*_\kappa
=a_\kappa/(\sqrt\kappa\,\varepsilon(1{-}\varepsilon))
\cdot(1+O(a_\kappa/\sqrt\kappa))$,
and the Bahadur--Rao prefactor reduces to
$\Pr[E>r_\kappa]
=(1{+}o(1))\cdot
\frac{\sqrt{\varepsilon(1{-}\varepsilon)}}
     {a_\kappa\sqrt{2\pi}}
\exp(-a_\kappa^2/(2\varepsilon(1{-}\varepsilon)))
=\frac{1{+}o(1)}{z_\kappa\sqrt{2\pi}}e^{-z_\kappa^2/2}$.
Mill's ratio gives
$\bar\Phi(z)=(1{+}O(z^{-2}))\phi(z)/z$,
confirming the match with $\bar\Phi(z_\kappa)$.
\end{proof}

\begin{remark}[Interpolation between CLT and large deviations]
\label{rem:interpolation}
When $a_\kappa=c$ (constant),
$P_e\approx\bar\Phi(c/\sqrt{\varepsilon(1{-}\varepsilon)})$
recovers the CLT of Theorem~\ref{thm:second-order}.
When $a_\kappa=\gamma\sqrt\kappa$,
$z_\kappa^2/2\approx\kappa D(\varepsilon{+}\gamma\|\varepsilon)$,
recovering the large-deviations exponent of
Corollary~\ref{cor:reliability}.
In the moderate-deviations zone, the error decays as a
\emph{stretched exponential}
$P_e=\exp(-\Theta(a_\kappa^2))$,
governed by the inverse normalized dispersion
$1/(2\varepsilon(1{-}\varepsilon))$.
Here $\ln$ is used in the rate formula; the corresponding
$\log_2$-rate is $1/(2\varepsilon(1{-}\varepsilon)\ln 2)$.
\end{remark}

\subsection{Complete Error Probability Phase Diagram}
\label{subsec:complete-phase}

We assemble the full asymptotic landscape.

\begin{theorem}[Complete phase diagram]
\label{thm:phase-diagram}
Let $\rho:=\sigma/(\kappa\log m)\in[0,1]$.
As $\kappa\to\infty$:

\noindent\emph{Coded scheme
\textup{(}MDS/Reed--Solomon\textup{)}:}

\begin{enumerate}[label=\textup{C\arabic*.}]
\item
\emph{Deep sub-capacity}
($\rho<\varepsilon{-}\eta$, $\eta>0$ fixed):
$P_e=1{-}\Theta(\kappa^{-1/2})
  e^{-\kappa D(\rho\|\varepsilon)}$.

\item
\emph{Near-capacity from below}
($\rho=\varepsilon{-}a_\kappa/\sqrt\kappa$,
$a_\kappa\to\infty$, $a_\kappa=o(\sqrt\kappa)$):
$P_e=1{-}(1{+}o(1))\,
  \bar\Phi(a_\kappa/\sqrt{\varepsilon(1{-}\varepsilon)})$.

\item
\emph{Critical window}
($\rho=\varepsilon{+}c/\sqrt\kappa$, $c\in\mathbb R$):
$P_e=\bar\Phi(c/\sqrt{\varepsilon(1{-}\varepsilon)})
  +O(\kappa^{-1/2})$.

\item
\emph{Near-capacity from above}
($\rho=\varepsilon{+}a_\kappa/\sqrt\kappa$,
$a_\kappa\to\infty$, $a_\kappa=o(\sqrt\kappa)$):
$P_e=(1{+}o(1))\,
  \bar\Phi(a_\kappa/\sqrt{\varepsilon(1{-}\varepsilon)})$.

\item
\emph{Deep super-capacity}
($\rho>\varepsilon{+}\eta$):
$P_e=\Theta(\kappa^{-1/2})
  e^{-\kappa D(\rho\|\varepsilon)}$.
\end{enumerate}

\noindent\emph{Derivation-constrained scheme:}

\begin{enumerate}[label=\textup{U\arabic*.}]
\item
\emph{Sub-full cache} ($\rho<1{-}\eta$):
$P_e=1{-}(1{-}\varepsilon)^{(1{-}\rho)\kappa}$.

\item
\emph{Near-full cache}
($\rho=1{-}c/\kappa$, $c>0$):
$P_e=1{-}(1{-}\varepsilon)^c$.

\item
\emph{Full cache} ($\rho=1$): $P_e=0$.
\end{enumerate}
\end{theorem}

\begin{proof}
Regimes C1 and C5 follow from the Bahadur--Rao expansion
(Theorem~\ref{thm:bahadur-rao}(i) and~(ii), respectively);
C2 and C4 from the moderate deviations principle
(Theorem~\ref{thm:moderate-dev});
C3 from the second-order expansion
(Theorem~\ref{thm:second-order}).
Regimes U1--U3 follow from the exact formula
$P_e^{\mathrm{unc}}=1{-}(1{-}\varepsilon)^N$
(Theorem~\ref{thm:unc-error-complete}).
\end{proof}

\begin{remark}[Absence of phase transition in the
  derivation-constrained scheme]
\label{rem:no-transition}
The coded scheme exhibits a sharp phase transition at
$\rho=\varepsilon$ with five asymptotic regimes.
The derivation-constrained scheme has no phase transition:
$P_e=1{-}(1{-}\varepsilon)^N$ is a smooth, monotone function
of $N$, reflecting the absence of collective statistical
effects---each dependency contributes independently, with no
threshold behavior.
\end{remark}

\begin{remark}[Depth-space instantiation]
\label{rem:forward-depth}
The phase diagram of Theorem~\ref{thm:phase-diagram} is stated
in terms of the normalized cache rate $\rho$ and the dependency
count $\kappa$.
Section~\ref{sec:arch-depth} re-parameterizes this diagram in
depth space via the architecture-dependent mapping
$d\mapsto\kappa_{\mathcal A}(d)$ from CT2
\textup{(Theorem~\ref{thm:exponential-capacity})}, exposing the
exponentially sharper phase transition of the merge architecture
\textup{(Theorem~\ref{thm:arch-phase-ds})} and the maximum
error-exponent gap between coded and derivation-constrained
schemes \textup{(Theorem~\ref{thm:max-gap-ds})}.
\end{remark}


\section{Multi-Query Joint Caching}
\label{sec:multi-query}

This section extends the single-query theory to $L\ge 2$ queries
sharing base-fact dependencies.
The shared structure creates joint coding gains at both first
order (overlap deduplication) and second order (statistical
pooling).
The main conclusions are: (i)~the minimum joint coded cache
depends only on the effective number of distinct dependencies
$n_{\mathrm{eff}}$ (Theorem~\ref{thm:joint-coded-mq});
(ii)~the $1/\varepsilon$ derivation penalty is universal,
independent of $L$, overlap, and $\delta$
(Theorem~\ref{thm:multi-penalty-mq});
(iii)~the derivation-constrained scheme has zero joint
dispersion, inheriting the single-query dichotomy
(Theorem~\ref{thm:multi-dispersion-mq}).

\subsection{Multi-Query Model and Overlap Structure}
\label{subsec:mq-model}

\begin{definition}[Multi-query caching instance]
\label{def:multi-query-mq}
A \emph{multi-query instance}
$\mathcal Q=(L,m,\varepsilon,(S_\ell)_{\ell=1}^L)$
consists of $L$ queries over a common base $B$ with
$|B|=m$, where query $\ell$ depends on the base-fact
subset $S_\ell\subseteq B$ with $|S_\ell|=\kappa$.
Each base fact is independently erased with probability
$\varepsilon$.
Define the \emph{dependency union}
$S:=\bigcup_\ell S_\ell$ with
$n_{\mathrm{eff}}:=|S|$ and the \emph{overlap index}
\begin{equation}\label{eq:overlap-index-mq}
  \Omega:=\frac{L\kappa}{n_{\mathrm{eff}}}\ge 1,
\end{equation}
with equality iff the $S_\ell$ are pairwise disjoint.
A $\delta$-reliable scheme must recover all $L$ query
answers jointly with probability $\ge 1{-}\delta$.
\end{definition}

\begin{definition}[Pairwise overlap coefficient]
\label{def:pairwise-overlap-mq}
For a symmetric instance with
$|S_\ell\cap S_{\ell'}|=\alpha\kappa$ for all
$\ell\ne\ell'$:
$n_{\mathrm{eff}}=(L{-}(L{-}1)\alpha)\kappa$ and
$\Omega=L/(L{-}(L{-}1)\alpha)$.
These formulas hold under the common-core overlap model
\textup{(}all pairwise overlap arises from a single shared
subset\textup{)}; for $L=2$ they hold without additional
assumptions.
\end{definition}

\begin{remark}[Derivation-structural origin of overlap]
\label{rem:overlap-origin}
For faithful programs with unique traces, the multi-query
derivation DAG $G_L:=\bigcup_\ell G(q_\ell,B)$ determines
$n_{\mathrm{eff}}=|V_0(G_L)|$.
Under $\Pi_k^{\parallel}$ (balanced merge), a single shared
intermediate node at depth $j$ creates base-fact overlap of
size $k\cdot 2^{j-1}$---exponential in $j$---while under
$\Pi_k$ (chain), the overlap per shared node is
$k{+}j{-}1=\Theta(j)$, linear.
This \emph{overlap amplification separation} is the
multi-query manifestation of CT2's capacity separation.
\end{remark}

\subsection{Joint Capacity: Coded and Derivation-Constrained}
\label{subsec:mq-capacity}

\begin{theorem}[Joint coded capacity]
\label{thm:joint-coded-mq}
For the multi-query instance $\mathcal Q$:
\begin{enumerate}[label=\textup{(\Roman*)}]
\item \emph{First order.}\;
$\sigma^*_{\mathrm{code}}(\mathcal Q)
=\varepsilon\,n_{\mathrm{eff}}\log m
+O(\sqrt{n_{\mathrm{eff}}}\,\log m)$.

\item \emph{Second order.}\;
\begin{equation}\label{eq:joint-coded-2nd-mq}
  \sigma^*_{\mathrm{code}}(\mathcal Q,\delta)
  =\varepsilon\,n_{\mathrm{eff}}\log m
  +\sqrt{n_{\mathrm{eff}}\,V_{\mathrm{cache}}}\;
    \Phi^{-1}(1{-}\delta)
  +\tfrac{1}{2}\log n_{\mathrm{eff}}
  +O(\log m).
\end{equation}

\item \emph{Strong converse.}\;
If $\sigma<(\varepsilon{-}\gamma)n_{\mathrm{eff}}\log m$:
$\lim_{n_{\mathrm{eff}}\to\infty}
\frac{-1}{n_{\mathrm{eff}}}\log P_c^*
=D(\varepsilon{-}\gamma\|\varepsilon)$.
\end{enumerate}
\end{theorem}

\begin{proof}
Joint recovery of all $L$ queries reduces to recovering the
$n_{\mathrm{eff}}$ distinct facts in $S$: once all facts are
available, every query answer is determined.
Apply an $(n_{\mathrm{eff}}{+}r,n_{\mathrm{eff}})$ MDS code
over $\mathbb F_q$ with $q\ge m$ to the facts in $S$.
The number of erasures is
$E\sim\mathrm{Bin}(n_{\mathrm{eff}},\varepsilon)$, and the
analysis of
Theorems~\ref{thm:coded-converse}--\ref{thm:coded-achievability}
and~\ref{thm:sc-exponent} carries over with $\kappa$ replaced
by $n_{\mathrm{eff}}$.
For the converse, the image-size bound
(Lemma~\ref{lem:image-size}) applied to the joint problem
yields
$P_c\le\mathbb E[\min(1,2^\sigma/m^E)]$ with
$E\sim\mathrm{Bin}(n_{\mathrm{eff}},\varepsilon)$.
\end{proof}

\begin{theorem}[Joint derivation-constrained capacity]
\label{thm:joint-unc-mq}
\begin{equation}\label{eq:joint-unc-exact-mq}
  \sigma^*_{\mathrm{unc}}(\mathcal Q,\delta)
  =\bigl(n_{\mathrm{eff}}-N^*(\varepsilon,\delta)\bigr)\log m
  +O(\log n_{\mathrm{eff}}).
\end{equation}
The $N^*$ unprotected facts may be placed arbitrarily among
the $n_{\mathrm{eff}}$ facts in $S$; the success probability
$(1{-}\varepsilon)^{N^*}$ depends only on $N^*$, not on the
overlap structure.
\end{theorem}

\begin{proof}
Joint success requires all $n_{\mathrm{eff}}$ facts to be
available.
Caching $n_{\mathrm{eff}}{-}N$ facts leaves $N$ unprotected;
$(1{-}\varepsilon)^N\ge 1{-}\delta$ gives
$N\le N^*(\varepsilon,\delta)$.
\end{proof}

\subsection{The Universal Derivation Penalty}
\label{subsec:mq-penalty}

\begin{theorem}[Multi-query derivation penalty]
\label{thm:multi-penalty-mq}
For the multi-query instance $\mathcal Q$:
\begin{enumerate}[label=\textup{(\Roman*)}]
\item \emph{First-order penalty.}\;
$\sigma^*_{\mathrm{unc}}/\sigma^*_{\mathrm{code}}
\to 1/\varepsilon$ as $\kappa\to\infty$,
independent of $L$ and $\alpha$.

\item \emph{Second-order refinement.}\;
\begin{equation}\label{eq:multi-penalty-2-mq}
  \frac{\sigma^*_{\mathrm{unc}}}
       {\sigma^*_{\mathrm{code}}}
  =\frac{1}{\varepsilon}
  -\frac{N^*(\varepsilon,\delta)}{\varepsilon\,n_{\mathrm{eff}}}
  -\frac{\Phi^{-1}(1{-}\delta)\sqrt{1{-}\varepsilon}}
        {\varepsilon^{3/2}\sqrt{n_{\mathrm{eff}}}}
  +O(n_{\mathrm{eff}}^{-1}).
\end{equation}
The approach rate is $\Theta(n_{\mathrm{eff}}^{-1/2})$,
depending on the overlap through $n_{\mathrm{eff}}$:
higher overlap \textup{(}smaller $n_{\mathrm{eff}}$ at
fixed $\kappa$\textup{)} yields a larger finite-length
correction, hence a wider gap from $1/\varepsilon$ at
any fixed~$\kappa$.
\end{enumerate}
\end{theorem}

\begin{proof}
The proof follows Theorem~\ref{thm:refined-penalty} with
$\kappa$ replaced by $n_{\mathrm{eff}}$.
For~(I): $\sigma_u/\sigma_c
\to n_{\mathrm{eff}}/(\varepsilon n_{\mathrm{eff}})
=1/\varepsilon$.
For~(II): expand $\sigma_u/\sigma_c$ as in
Theorem~\ref{thm:refined-penalty}.
\end{proof}

\begin{remark}[Universality of the $1/\varepsilon$ penalty]
\label{rem:universality}
The first-order penalty is $1/\varepsilon$ regardless of $L$,
$\alpha$, $\kappa$, or $\delta$: it is a pure function of the
channel parameter.
The coded scheme pays
$\varepsilon\log m$ bits per dependency on average (algebraic
redundancy across coordinates); the
derivation-constrained scheme pays $\log m$ bits per protected
dependency regardless of $\varepsilon$ (no cross-coordinate
correction).
\end{remark}

\subsection{Joint Coding Gain}
\label{subsec:mq-gain}

\begin{theorem}[Joint coding gain]
\label{thm:joint-gain-mq}
Let $\sigma^*_{\mathrm{sep}}$ denote the total cache under
separate coding \textup{(}each query coded independently
with reliability $\delta/L$\textup{)}.
\begin{enumerate}[label=\textup{(\Roman*)}]
\item \emph{First-order gain.}\;
$G_1:=\sigma^*_{\mathrm{sep}}/\sigma^*_{\mathrm{joint}}
\to\Omega=L\kappa/n_{\mathrm{eff}}$.

\item \emph{Second-order gain.}\;
$G_2
=L\sqrt\kappa\,\Phi^{-1}(1{-}\delta/L)\,/\,
(\sqrt{n_{\mathrm{eff}}}\,\Phi^{-1}(1{-}\delta))$.
For $L=2$, $\alpha=0$, $\delta=0.1$:
$G_2\approx 1.81$.
\end{enumerate}
\end{theorem}

\begin{proof}
\emph{(I)}\;
$\sigma^*_{\mathrm{sep}}
=(1{+}o(1))\varepsilon L\kappa\log m$;
$\sigma^*_{\mathrm{joint}}
=(1{+}o(1))\varepsilon n_{\mathrm{eff}}\log m$;
the ratio is $\Omega$.

\emph{(II)}\;
Under separate coding, the second-order total is
$L\sqrt{\kappa V_{\mathrm{cache}}}
\Phi^{-1}(1{-}\delta/L)$;
the joint second-order term is
$\sqrt{n_{\mathrm{eff}}V_{\mathrm{cache}}}
\Phi^{-1}(1{-}\delta)$.
\end{proof}

\begin{remark}[Decomposition of the joint coding gain]
\label{rem:gain-decomp-mq}
The total gain decomposes into a \emph{deduplication gain}
$\Omega$ (shared facts encoded once) and a \emph{pooling
gain} from two sources: variance reduction
($\sqrt{n_{\mathrm{eff}}}<L\sqrt\kappa$ when $\Omega>1$)
and reliability allocation
($\Phi^{-1}(1{-}\delta/L)>\Phi^{-1}(1{-}\delta)$ for
$L\ge 2$).
\end{remark}

\subsection{Multi-Query Dispersion and Pooling}
\label{subsec:mq-dispersion}

\begin{theorem}[Multi-query dispersion dichotomy]
\label{thm:multi-dispersion-mq}
For the multi-query instance $\mathcal Q$:
\begin{enumerate}[label=\textup{(\Roman*)}]
\item \emph{Coded dispersion.}\;
$V_{\mathrm{joint}}
=n_{\mathrm{eff}}\cdot V_{\mathrm{cache}}$;
per query:
$V_{\mathrm{joint}}/L
=(\kappa/\Omega)\cdot V_{\mathrm{cache}}
\le\kappa\cdot V_{\mathrm{cache}}$.

\item \emph{Derivation-constrained dispersion.}\;
$V_{\mathrm{unc}}=0$ \textup{(}identically, as in the
single-query case\textup{)}.

\item \emph{Pooling inequality.}\;
\begin{equation}\label{eq:pooling-ineq-mq}
  \sum_{\ell=1}^L
  \sqrt{|S_\ell|\,V_{\mathrm{cache}}}\;
    \Phi^{-1}(1{-}\delta/L)
  \;\ge\;
  \sqrt{n_{\mathrm{eff}}\,V_{\mathrm{cache}}}\;
    \Phi^{-1}(1{-}\delta),
\end{equation}
with equality iff $L=1$.
\end{enumerate}
\end{theorem}

\begin{proof}
(I) and (II) follow from the second-order expansions of
Theorems~\ref{thm:joint-coded-mq}
and~\ref{thm:joint-unc-mq}.
For~(III): since $n_{\mathrm{eff}}\le L\kappa$,
$L\sqrt\kappa\ge\sqrt{L\,n_{\mathrm{eff}}}
\ge\sqrt{n_{\mathrm{eff}}}$; combining with
$\Phi^{-1}(1{-}\delta/L)>\Phi^{-1}(1{-}\delta)$ for
$L\ge 2$ yields the inequality.
\end{proof}

\subsection{\texorpdfstring{Capacity Region and Large-$L$ Scaling}{Capacity Region and Large-L Scaling}}
\label{subsec:mq-scaling}

\begin{theorem}[Multi-query capacity region]
\label{thm:capacity-region-mq}
Under heterogeneous reliability
$(\delta_1,\ldots,\delta_L)$ with
$\delta_{\min}:=\min_\ell\delta_\ell$:

\noindent\emph{Coded:}\;
$\sigma\ge\varepsilon n_{\mathrm{eff}}\log m
+\sqrt{n_{\mathrm{eff}}V_{\mathrm{cache}}}\;
\Phi^{-1}(1{-}\delta_{\min})+O(\log m)$.

\noindent\emph{Derivation-constrained:}\;
$\sigma\ge(n_{\mathrm{eff}}{-}N^*(\varepsilon,\delta_{\min}))\log m
+O(\log n_{\mathrm{eff}})$.

\noindent
In both models, the binding constraint is the strictest
reliability requirement.
\end{theorem}

\begin{proof}
Joint recovery of all $n_{\mathrm{eff}}$ facts ensures all
queries succeed; the binding event is
$\delta_{\min}$.
For the coded scheme,
$\Pr[E\le r]\ge 1{-}\delta_{\min}$ determines $r$.
For the derivation-constrained scheme,
$(1{-}\varepsilon)^N\ge 1{-}\delta_{\min}$ gives
$N\le N^*(\varepsilon,\delta_{\min})$.
\end{proof}

\begin{theorem}[Scaling with number of queries]
\label{thm:large-L-mq}
Consider $L$ queries in the ``common core'' model where
$|\bigcap_\ell S_\ell|=\alpha\kappa$ and all remaining
facts are private.
As $L\to\infty$ with $\kappa$ fixed:
\begin{enumerate}[label=\textup{(\Roman*)}]
\item $n_{\mathrm{eff}}
=\alpha\kappa+L(1{-}\alpha)\kappa$.

\item \emph{Per-query coded cache:}\;
$\sigma_{\mathrm{code}}/L
\to\varepsilon(1{-}\alpha)\kappa\log m$.

\item \emph{Per-query uncoded cache:}\;
$\sigma_{\mathrm{unc}}/L
\to(1{-}\alpha)\kappa\log m$.

\item \emph{Per-query penalty:}\;
$(\sigma_{\mathrm{unc}}/L)\,/\,
(\sigma_{\mathrm{code}}/L)\to 1/\varepsilon$.
\end{enumerate}
\end{theorem}

\begin{proof}
(I)~In the common-core model:
$n_{\mathrm{eff}}=\alpha\kappa+L(1{-}\alpha)\kappa$.
(II)~$\sigma_{\mathrm{code}}/L
=\varepsilon[(1{-}\alpha)\kappa
+\alpha\kappa/L]\log m+O(\sqrt{\kappa/L}\log m)$.
(III)~$\sigma_{\mathrm{unc}}/L
=(1{-}\alpha)\kappa\log m
+(\alpha\kappa{-}N^*(\varepsilon,\delta))\log m/L
+O(\log n_{\mathrm{eff}}/L)$.
(IV)~Ratio $\to(1{-}\alpha)/(\varepsilon(1{-}\alpha))
=1/\varepsilon$.
\end{proof}

\begin{remark}[Exponent crossover at equal cache]
\label{rem:crossover-mq}
At equal normalized cache rate
$\rho=\sigma/(n_{\mathrm{eff}}\log m)$, the coded
reliability exponent $D(\rho\|\varepsilon)$
\textup{(}strictly convex, increasing from~$0$ at
$\rho=\varepsilon$\textup{)} and the
derivation-constrained success exponent
$(1{-}\rho)|\!\log(1{-}\varepsilon)|$
\textup{(}linear, decreasing to~$0$ at $\rho=1$\textup{)}
cross exactly once at some $\rho^*\in(\varepsilon,1)$.
Below $\rho^*$, the derivation-constrained scheme has the
larger exponent \textup{(}it operates far from its
capacity $\rho=1$\textup{)};
above $\rho^*$, the coded scheme dominates.
At equal \emph{reliability} $\delta$, however, the coded
scheme always uses less cache by the factor
$1/\varepsilon$.
\end{remark}


\section{Architecture-Parameterized Depth-Space Analysis}
\label{sec:arch-depth}

The phase diagram of Theorem~\ref{thm:phase-diagram} is stated
in terms of the normalized cache rate
$\rho=\sigma/(\kappa\log m)$, treating the number of
dependencies $\kappa$ as a free parameter.
We now re-parameterize the diagram in \emph{depth space} by
substituting the architecture-dependent mapping
$\kappa=\kappa_{\mathcal A}(d)$ from CT2
(Theorem~\ref{thm:exponential-capacity}):
$\kappa_{\mathrm{chain}}(d)=k{+}d{-}1$ and
$\kappa_{\mathrm{merge}}(d)=k\cdot 2^{d-1}$.
This reveals that the phase transition's location, width, and
sharpness in depth space are all architecture-dependent,
establishing a direct bridge from CT2 (capacity separation) to
the error probability landscape.

\subsection{Depth-Space Phase Transition}
\label{subsec:depth-phase}

\begin{definition}[Depth-space critical depth and transition
  width]
\label{def:depth-space-trans}
Fix cache budget $\sigma>0$, erasure rate
$\varepsilon\in(0,1)$, target reliabilities
$\delta_1<\delta_2\in(0,1)$, and architecture
$\mathcal A\in\{\mathrm{chain},\mathrm{merge}\}$.
The \emph{critical depth} $d^*(\sigma,\mathcal A)$ satisfies
$\sigma=\varepsilon\,\kappa_{\mathcal A}(d^*)\log m$.
The \emph{transition width}
$\Delta d(\sigma,\mathcal A)$ is the depth interval over
which $P_e$ transitions from $\delta_1$ to~$\delta_2$.
\end{definition}

\begin{theorem}[Architecture-dependent depth-space phase
  transition]
\label{thm:arch-phase-ds}
Fix $\sigma>0$, $\varepsilon\in(0,1)$, and let
$r:=\lfloor\sigma/\log m\rfloor$.

\smallskip\noindent
\textbf{\upshape(I) Critical depth.}
Chain: $d^*_c=r/\varepsilon-k+1$.
Merge: $d^*_m=1+\log_2(r/(\varepsilon k))$.
Both are real-valued; integer rounding introduces $O(1)$ error.

\smallskip\noindent
\textbf{\upshape(II) Coded transition width.}
\begin{enumerate}[label=\textup{(\alph*)}]
\item Chain:
$\Delta d_c^{\mathrm{code}}
=\Theta(\sqrt{d^*_c})$.
\item Merge:
$\Delta d_m^{\mathrm{code}}
=\Theta(2^{-d^*_m/2})$.
\end{enumerate}

\smallskip\noindent
\textbf{\upshape(III) Derivation-constrained transition
  width.}
\begin{enumerate}[label=\textup{(\alph*)}]
\item Chain:
$\Delta d_c^{\mathrm{unc}}=\Theta(1/\varepsilon)$.
\item Merge:
$\Delta d_m^{\mathrm{unc}}
=O(\log(1/\delta))$.
\end{enumerate}

\smallskip\noindent
\textbf{\upshape(IV) Sharpness ratio.}
$\Delta d_c^{\mathrm{code}}
/\Delta d_m^{\mathrm{code}}
=\Theta(\sqrt{d^*_c}\cdot 2^{d^*_m/2})
\to\infty$:
the merge architecture has an exponentially sharper
phase transition in depth space.
\end{theorem}

\begin{proof}
\emph{(I)}\;
Solve $r=\varepsilon\kappa_{\mathcal A}(d^*)+O(1)$
for $d^*$: chain gives $d^*_c=r/\varepsilon-k+1$;
merge gives $2^{d^*_m-1}=r/(\varepsilon k)$.

\emph{(II)}\;
At depth $d^*{+}\Delta d$ with fixed cache $\sigma$:
for the chain,
$\kappa(d^*{+}\Delta d)=\kappa(d^*)+\Delta d$, so the
CLT transition requires
$\Delta\kappa=O(\sqrt{\kappa(d^*)})$, giving
$\Delta d=O(\sqrt{d^*})$.
For the merge,
$\kappa(d^*{+}\Delta d)=\kappa(d^*)\cdot 2^{\Delta d}$;
linearizing for small $\Delta d$,
$\varepsilon\kappa(d^*)\Delta d\ln 2
=O(\sqrt{\varepsilon\kappa(d^*)})$,
giving $\Delta d=O(1/\sqrt{\varepsilon\kappa(d^*)})
=\Theta(2^{-d^*/2})$.

\emph{(III)}\;
For the chain: each depth unit adds one exposed
dependency, so
$\Delta d=\log((1{-}\delta_2)/(1{-}\delta_1))/
\log(1{-}\varepsilon)=\Theta(1/\varepsilon)$.
For the merge: exposed dependencies double per depth
unit, so
$\Delta d=\log_2(\log(1{-}\delta_2)/\log(1{-}\delta_1))
=O(\log(1/\delta))$.

\emph{(IV)}\;
Immediate from (II).
\end{proof}

\subsection{Exact Asymptotics in Depth Space}
\label{subsec:depth-exact}

\begin{corollary}[Depth-space Bahadur--Rao instantiation]
\label{cor:depth-br}
Fix \(\sigma>0\), \(\varepsilon\in(0,1)\), and let \(r:=\lfloor \sigma/\log m\rfloor\).
Let \(d^*(\sigma,\mathcal A)\) be defined by \(\sigma=\varepsilon\,\kappa_{\mathcal A}(d^*)\log m\)
(Definition~\ref{def:depth-space-trans}).
For any \(d>d^*\), define the effective normalized cache rate
\[
  \alpha_{\mathcal A}(d):=\frac{r}{\kappa_{\mathcal A}(d)}.
\]
Then \(\alpha_{\mathcal A}(d)<\varepsilon\) and hence the coded scheme is in the strong-converse
(lower-tail) regime. In particular, Theorem~\ref{thm:bahadur-rao}\textup{(i)} yields
\begin{equation}\label{eq:arch-br}
  P_c^{\mathrm{MDS}}(d,\mathcal A)
  =\frac{2^{-\kappa_{\mathcal A}(d)\,
    D(\alpha_{\mathcal A}(d)\|\varepsilon)}}
    {(1-e^{t^*(\alpha_{\mathcal A}(d))})
    \sqrt{2\pi\kappa_{\mathcal A}(d)\,
    \alpha_{\mathcal A}(d)(1-\alpha_{\mathcal A}(d))}}
  \bigl(1+O(\kappa_{\mathcal A}(d)^{-1})\bigr).
\end{equation}
\end{corollary}

\begin{proof}
Since \(d>d^*\) and \(\kappa_{\mathcal A}(d)\) is strictly increasing in \(d\) for both
architectures, we have \(\kappa_{\mathcal A}(d)>\kappa_{\mathcal A}(d^*)\).
Moreover, by definition of \(d^*\), \(r\approx \varepsilon\,\kappa_{\mathcal A}(d^*)\) up to an \(O(1)\)
integer-rounding term. Therefore
\[
  \alpha_{\mathcal A}(d)=\frac{r}{\kappa_{\mathcal A}(d)}
  < \frac{\varepsilon\,\kappa_{\mathcal A}(d^*)+O(1)}{\kappa_{\mathcal A}(d)}
  < \varepsilon
\]
for all sufficiently large \(\kappa_{\mathcal A}(d)\), which is the regime of interest in the
depth-space asymptotics. Substituting \(\kappa=\kappa_{\mathcal A}(d)\) and
\(\alpha=\alpha_{\mathcal A}(d)\) into Theorem~\ref{thm:bahadur-rao}\textup{(i)} gives
\eqref{eq:arch-br}.
\end{proof}

\begin{theorem}[Depth-space moderate deviations]
\label{thm:arch-mod-dev}
Fix a cache sequence
$\sigma_d=\varepsilon\kappa_{\mathcal A}(d)\log m
+a_d\sqrt{\kappa_{\mathcal A}(d)}\,\log m$
with $a_d\to\infty$ and
$a_d=o(\sqrt{\kappa_{\mathcal A}(d)})$.
By Theorem~\textup{\ref{thm:moderate-dev}}:
\begin{enumerate}[label=\textup{(\alph*)}]
\item Chain \textup{($\kappa=\Theta(d)$)}:
the constraint $a_d=o(\sqrt d)$ gives
$P_e\sim\exp(-\Theta(a_d^2))$ with
$a_d^2=o(d)$;
choosing $a_d=d^{1/4}$ yields a representative
stretched-exponential decay
$P_e\sim\exp(-\Theta(\sqrt d))$.

\item Merge \textup{($\kappa=\Theta(2^d)$)}:
the constraint $a_d=o(2^{d/2})$ gives
$P_e\sim\exp(-\Theta(a_d^2))$ with
$a_d^2=o(2^d)$;
choosing $a_d=2^{d/4}$ yields
$P_e\sim\exp(-\Theta(2^{d/2}))$.
\end{enumerate}
Per unit depth, the merge architecture achieves
exponentially faster error decay, inheriting CT2's
capacity separation.
\end{theorem}

\begin{proof}
Theorem~\ref{thm:moderate-dev}(i) gives
$\ln P_e\sim-a_d^2/(2\varepsilon(1{-}\varepsilon))$.
The architecture enters through
$\kappa_{\mathcal A}(d)$, which controls the feasible
range of $a_d=o(\sqrt\kappa)$: for the chain,
$a_d=o(\sqrt d)$; for the merge,
$a_d=o(2^{d/2})$.
\end{proof}

\subsection{Maximum Error Exponent Gap}
\label{subsec:max-gap-ds}

The preceding depth-space analysis quantifies how the coded
scheme's phase transition sharpens with architecture.
We now show that the \emph{gap} between coded and
derivation-constrained cache requirements, measured at the
error-exponent level, is maximized at low reliability and equals
the first-order derivation penalty.

\begin{theorem}[Maximum error exponent gap]
\label{thm:max-gap-ds}
Let $\varepsilon\in(0,1/2]$.
For exponent level $E\in(0,|\!\log(1{-}\varepsilon)|)$,
define the coded cache rate
$\rho_{\mathrm{code}}(E):=D^{-1}(E\|\varepsilon)$
\textup{(}upper branch, $\rho>\varepsilon$\textup{)} and
the derivation-constrained cache rate
$\rho_{\mathrm{unc}}(E):=1{-}E/|\!\log(1{-}\varepsilon)|$.
Then:
\begin{enumerate}[label=\textup{(\Roman*)}]
\item The ratio
$h(E):=\rho_{\mathrm{unc}}(E)/\rho_{\mathrm{code}}(E)$
is strictly decreasing on
$(0,\,|\!\log(1{-}\varepsilon)|)$.

\item At low exponent:
\begin{equation}\label{eq:max-ratio-ds}
  \lim_{E\to 0^+}h(E)=\frac{1}{\varepsilon}\,,
\end{equation}
recovering the first-order derivation penalty.

\item At high exponent:
$h(E)\to 0$ as
$E\to|\!\log(1{-}\varepsilon)|$,
since $\rho_{\mathrm{unc}}\to 0$ while
$\rho_{\mathrm{code}}$ remains bounded away from
zero.

\item There exists a unique crossover
$E^*\in(0,|\!\log(1{-}\varepsilon)|)$ with
$h(E^*)=1$: for $E<E^*$ the
derivation-constrained scheme requires strictly more
cache; for $E>E^*$ it requires strictly less.
\end{enumerate}
\end{theorem}

\begin{proof}
The upper branch of $D^{-1}(E\|\varepsilon)$ is well-defined on
$(0,D(1\|\varepsilon)]=(0,\log(1/\varepsilon)]$; since
$\varepsilon\le 1/2$, $\log(1/\varepsilon)\ge|\!\log(1{-}\varepsilon)|$,
so the stated domain is contained in the feasible range.

\emph{(I).}\;
$h'(E)=(\rho_u'\rho_c-\rho_u\rho_c')/\rho_c^2$
with $\rho_u'=-1/|\!\log(1{-}\varepsilon)|<0$ and
$\rho_c'=1/D'(\rho_c\|\varepsilon)>0$
(since $D'(\rho\|\varepsilon)>0$ for $\rho>\varepsilon$).
Both terms in the numerator are negative, giving
$h'(E)<0$.

\emph{(II).}\;
As $E\to 0^+$:
$\rho_{\mathrm{code}}\to\varepsilon$ and
$\rho_{\mathrm{unc}}\to 1$; the ratio converges to
$1/\varepsilon$.

\emph{(III).}\;
At $E=|\!\log(1{-}\varepsilon)|$: $\rho_u=0$;
$\rho_c=D^{-1}(|\!\log(1{-}\varepsilon)|\,\|\,\varepsilon)>0$
since the equation $D(\rho\|\varepsilon)=|\!\log(1{-}\varepsilon)|$
has a solution $\rho\in(\varepsilon,1]$ by the
feasibility established above.
Hence $h\to 0$.

\emph{(IV).}\;
Since $h$ is continuous, strictly decreasing from
$1/\varepsilon>1$ to~$0$, the
intermediate value theorem gives a unique $E^*$
with $h(E^*)=1$.
\end{proof}

\begin{remark}[Operational interpretation]
\label{rem:exp-gap-ds}
The derivation penalty is most severe at low exponent
levels (small $E$), where the coded scheme operates
near its capacity threshold $\rho=\varepsilon$ and
exploits statistical concentration, while the
derivation-constrained scheme cannot.
At high exponent levels ($E>E^*$), the relationship
reverses: the derivation-constrained scheme's linear
cache--exponent tradeoff becomes more efficient than
the coded scheme's convex tradeoff.
In depth space, the merge architecture's exponentially
larger $\kappa(d^*)$ amplifies both the advantage and
disadvantage by a factor of $\Theta(2^d)$ compared to
the chain.
\end{remark}

\subsection{The CT2--CT4 Bridge}
\label{subsec:ct2-ct4}

\begin{remark}[Complete logical chain from CT1 to the
  depth-space landscape]
\label{rem:ct2-ct4-ds}
The depth-space re-parameterization makes the following
chain explicit.
CT1 establishes the per-step information rate
$\log m$~bits, entering as the factor $\log m$ in
$V_{\mathrm{cache}}$.
CT2 determines the architecture-dependent mapping
$d\mapsto\kappa_{\mathcal A}(d)$: linear for chains,
exponential for merges.
The phase diagram of
Theorem~\textup{\ref{thm:phase-diagram}} applies at each
$\kappa$; the architecture maps this to a
depth-dependent landscape.
The transition width
$\Delta d\propto 1/\sqrt{\kappa'(d)}$ is the ``inverse''
of the capacity growth rate: $\kappa'_c=1$ versus
$\kappa'_m=\Theta(2^d)$;
the exponentially larger derivative for the merge
architecture yields the exponentially sharper transition.
CT4's depth--resilience duality
\textup{(Theorem~\ref{thm:depth-resilience-duality})} is
the special case $\sigma=0$ \textup{(}pure native
resilience\textup{)}:
$d^*_{\mathrm{chain}}=N^*-k+1=\Theta(\delta/\varepsilon)$,
$d^*_{\mathrm{merge}}
=1+\log_2(N^*/k)=\Theta(\log(\delta/\varepsilon))$.
The moderate-deviations result
\textup{(Theorem~\ref{thm:arch-mod-dev})} interpolates
between the CLT regime near $\Delta d=O(1)$ and the
large-deviations regime at $\Delta d=\Theta(d^*)$, with
the Bahadur--Rao prefactor
\textup{(Corollary~\ref{cor:depth-br})} providing exact
finite-$d$ evaluation.
\end{remark}


\section{Numerical Validation}
\label{sec:numerical}

This section validates the central finite-length predictions of
Sections~\ref{sec:noise-resilience}--\ref{sec:arch-depth} by exact computation
and Monte~Carlo simulation.
Unless stated otherwise, we use \(m=256\) (so \(\log_2 m=8\) bits) and \(k=2\).
All coded-caching numerical results treat \(\kappa\) as the number of
\emph{distinct} base-fact dependencies
(\(\kappa=\kappa(q,B)\), Definition~\ref{def:dep-count}), consistent with the
distinct-coordinate regime (Lemma~\ref{lem:distinct-coordinates}).

\subsection{Computational methodology (brief)}
\label{subsec:comp-methodology}

All experiments were executed in Python~3.11 using SciPy~1.12~\cite{virtanen2020scipy}
and NumPy~1.26.
For coded schemes, binomial tail probabilities
\(\Pr[\mathrm{Bin}(\kappa,\varepsilon)\le r]\) are computed via
\texttt{scipy.stats.binom.cdf} (exact to machine precision for the reported ranges),
and the minimal parity count \(r^*\) is obtained by exact CDF inversion.
For derivation-constrained schemes, the exact error probability
\(P_e^{\mathrm{unc}}=1-(1-\varepsilon)^N\) is evaluated in closed form.

\begin{figure}[ht]
\centering

\begin{minipage}{0.49\textwidth}
\centering
\begin{tikzpicture}
\begin{axis}[
  width=\textwidth,
  height=0.72\textwidth,
  xmode=log,
  log basis x=10,
  xtick={50,100,200,500,1000,5000},
  xmin=45, xmax=6000,
  ymin=3, ymax=11,
  grid=both,
  xlabel={\(\kappa\)},
  ylabel={Penalty ratio \(\sigma^*_{\mathrm{unc}}/\sigma^*_{\mathrm{code}}\)},
  title={Single-query penalty (Exp.~1)},
  legend style={
    at={(0.5,-0.30)},
    anchor=north,
    draw=none,
    fill=none,
    legend columns=2,
    font=\scriptsize,
    row sep=2pt,
    /tikz/every even column/.append style={column sep=6pt},
  },
  legend cell align=left,
]
\addplot+[mark=o, thick] coordinates {
  (50,6.12) (100,7.07) (200,7.65) (500,8.46) (1000,8.92) (5000,9.49)
};
\addlegendentry{\(\varepsilon=0.1\) Exact}

\addplot+[mark=triangle*, thick, dashed] coordinates {
  (50,4.36) (100,6.06) (200,7.23) (500,8.26) (1000,8.77) (5000,9.45)
};
\addlegendentry{\(\varepsilon=0.1\) Thm~\ref{thm:refined-penalty}}

\addplot+[mark=square, thick] coordinates {
  (50,3.57) (100,4.00) (200,4.26) (500,4.46) (1000,4.63) (5000,4.83)
};
\addlegendentry{\(\varepsilon=0.2\) Exact}

\addplot+[mark=diamond*, thick, dashed] coordinates {
  (50,3.19) (100,3.72) (200,4.09) (500,4.43) (1000,4.59) (5000,4.82)
};
\addlegendentry{\(\varepsilon=0.2\) Thm~\ref{thm:refined-penalty}}

\addplot[domain=45:6000, samples=2, thick, dotted, gray, forget plot] {10};
\addplot[domain=45:6000, samples=2, thick, dotted, gray, forget plot] {5};
\end{axis}
\end{tikzpicture}
\end{minipage}\hfill
\begin{minipage}{0.49\textwidth}
\centering
\begin{tikzpicture}
\begin{axis}[
  width=\textwidth,
  height=0.72\textwidth,
  xmin=450, xmax=5200,
  ymin=4.35, ymax=5.05,
  grid=both,
  xlabel={\(n_{\mathrm{eff}}\)},
  ylabel={Penalty ratio},
  title={Multi-query penalty (Exp.~5), \(\varepsilon=0.2\)},
  legend style={
    at={(0.5,-0.30)},
    anchor=north,
    draw=none,
    fill=none,
    legend columns=2,
    font=\scriptsize,
    row sep=2pt,
    /tikz/every even column/.append style={column sep=6pt},
  },
  legend cell align=left,
]
\addplot+[only marks, mark=*, mark size=2.1pt, thick] coordinates {
  (500,4.46) (1000,4.63) (850,4.59) (700,4.55)
  (2500,4.75) (1900,4.73) (5000,4.83) (2750,4.77)
};
\addlegendentry{Exact}

\addplot+[only marks, mark=x, mark size=2.6pt, thick, dashed] coordinates {
  (500,4.43) (1000,4.59) (850,4.56) (700,4.52)
  (2500,4.74) (1900,4.71) (5000,4.82) (2750,4.76)
};
\addlegendentry{\eqref{eq:multi-penalty-2-mq}}

\addplot+[domain=450:5200, samples=2, thick, dotted, gray] {5};
\addlegendentry{\(1/\varepsilon\)}
\end{axis}
\end{tikzpicture}
\end{minipage}

\caption{Penalty ratios and convergence to \(1/\varepsilon\) (Experiments~1 and~5).
All legends are placed outside the plotting regions to avoid covering curves.}
\label{fig:penalty}
\end{figure}
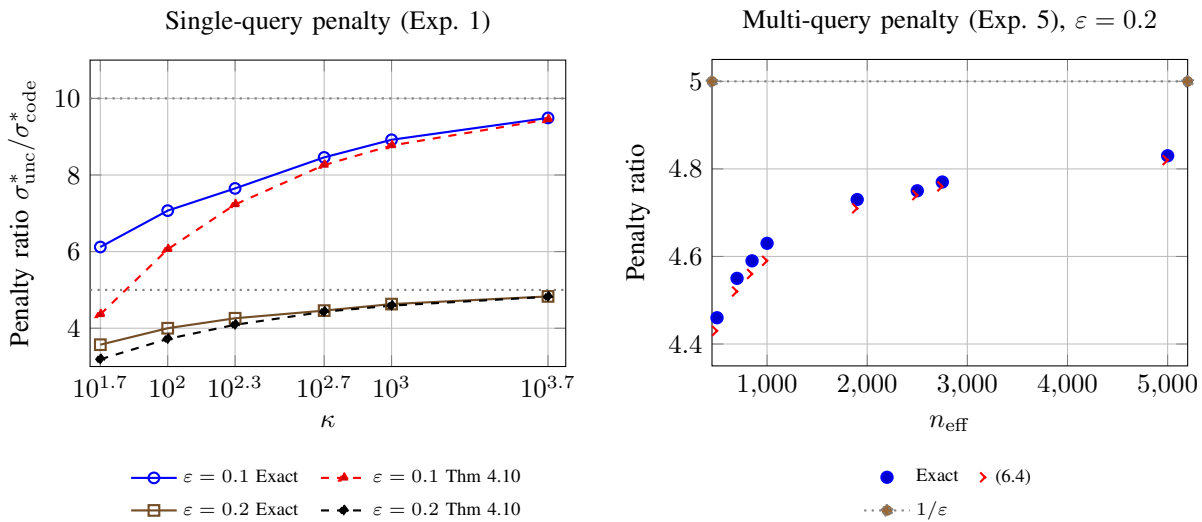

\subsection{Experiments 1 and 5: derivation-penalty convergence and universality}
\label{subsec:exp-penalty-unified}

Figure~\ref{fig:penalty} (left) validates
Theorems~\ref{thm:source-channel-separation}
and~\ref{thm:refined-penalty}.
At \(\varepsilon=0.1\) the exact ratio rises from \(6.12\) at
\(\kappa=50\) to \(9.49\) at \(\kappa=5000\), reaching within
\(6\%\) of the first-order limit \(1/\varepsilon=10\); the
closed-form second-order prediction~\eqref{eq:refined-penalty}
tracks the exact value with decreasing relative error
(\(29\%\) at \(\kappa=50\) versus \(0.4\%\) at \(\kappa=5000\)),
confirming the \(\Theta(\kappa^{-1/2})\) convergence rate.
At \(\varepsilon=0.2\) the pattern is analogous, with the exact
ratio reaching \(4.83\) at \(\kappa=5000\) against the limit
\(1/\varepsilon=5\).

Figure~\ref{fig:penalty} (right) confirms the multi-query
universality of Theorem~\ref{thm:multi-penalty-mq}: across eight
\((L,\alpha)\) configurations with
\(n_{\mathrm{eff}}\in[500,5000]\), the penalty ratio depends only
on \(n_{\mathrm{eff}}\)---not separately on \(L\) or
\(\alpha\)---and the second-order
prediction~\eqref{eq:multi-penalty-2-mq} matches the exact ratio
to within \(2\%\) for all entries.

\begin{figure}[ht]
\centering

\begin{minipage}{0.49\textwidth}
\centering
\begin{tikzpicture}
\begin{axis}[
  ylabel shift = -4pt,
  yticklabel style = {xshift=-2pt},
  width=\textwidth,
  height=0.72\textwidth,
  xmode=log,
  log basis x=10,
  xtick={100,200,500,1000,5000},
  xmin=90, xmax=6000,
  ymin=0.03, ymax=0.21,
  grid=both,
  xlabel={\(\kappa\)},
  ylabel={Empirical exponent \(\hat D(\kappa)\)},
  title={Strong converse exponent (Exp.~2), \(\varepsilon=0.3\)},
  legend style={
    at={(0.5,-0.30)},
    anchor=north,
    draw=none,
    fill=none,
    legend columns=2,
    font=\scriptsize,
    row sep=2pt,
    /tikz/every even column/.append style={column sep=6pt},
  },
  legend cell align=left,
]

\addplot+[mark=o, thick] coordinates {(100,0.0592) (200,0.0504) (500,0.0436) (1000,0.0409) (5000,0.0381)};
\addlegendentry{\(\gamma=0.10\) (\(\alpha=0.20\))}

\addplot+[mark=square, thick] coordinates {(100,0.1127) (200,0.1027) (500,0.0952) (1000,0.0922) (5000,0.0891)};
\addlegendentry{\(\gamma=0.15\) (\(\alpha=0.15\))}

\addplot+[mark=triangle*, thick] coordinates {(100,0.1929) (200,0.1828) (500,0.1751) (1000,0.1720) (5000,0.1689)};
\addlegendentry{\(\gamma=0.20\) (\(\alpha=0.10\))}

\addplot+[domain=90:6000, samples=2, thick, dashed, gray] {0.0371};
\addlegendentry{\(D(0.20\|0.3)\)}
\addplot+[domain=90:6000, samples=2, thick, dashed, gray!70] {0.0881};
\addlegendentry{\(D(0.15\|0.3)\)}
\addplot+[domain=90:6000, samples=2, thick, dashed, gray!40] {0.1678};
\addlegendentry{\(D(0.10\|0.3)\)}
\end{axis}
\end{tikzpicture}
\end{minipage}\hfill
\hspace{2mm}
\begin{minipage}{0.49\textwidth}
\centering
\begin{tikzpicture}
\begin{axis}[
  ylabel shift = -4pt,
  yticklabel style = {xshift=-2pt},
  width=\textwidth,
  height=0.72\textwidth,
  xmode=log,
  log basis x=10,
  xtick={100,200,500,1000,5000},
  xmin=90, xmax=6000,
  ymin=1.8, ymax=5.9,
  grid=both,
  xlabel={\(\kappa\)},
  ylabel={\(\Psi(\kappa)=\kappa(\hat D - D)\)},
  title={Bahadur--Rao prefactor scale (Exp.~2)},
  legend style={
    at={(0.5,-0.30)},
    anchor=north,
    draw=none,
    fill=none,
    legend columns=2,
    font=\scriptsize,
    row sep=2pt,
    /tikz/every even column/.append style={column sep=6pt},
  },
  legend cell align=left,
]
\addplot+[mark=o, thick] coordinates {(100,2.21) (200,2.66) (500,3.25) (1000,3.80) (5000,5.00)};
\addlegendentry{Obs \(\alpha=0.20\)}

\addplot+[mark=square, thick] coordinates {(100,2.46) (200,2.92) (500,3.55) (1000,4.10) (5000,5.00)};
\addlegendentry{Obs \(\alpha=0.15\)}

\addplot+[mark=triangle*, thick] coordinates {(100,2.51) (200,3.00) (500,3.65) (1000,4.20) (5000,5.50)};
\addlegendentry{Obs \(\alpha=0.10\)}

\addplot[thick, dashed, black, domain=90:6000, samples=200]
  {0.5*ln(x)/ln(2) - 1.26};
\addlegendentry{BR pred \(\alpha=0.20\)}

\addplot[thick, dash dot, black!60, domain=90:6000, samples=200]
  {0.5*ln(x)/ln(2) - 0.93};
\addlegendentry{BR pred \(\alpha=0.15\)}

\addplot[thick, loosely dashed, black!35, domain=90:6000, samples=200]
  {0.5*ln(x)/ln(2) - 0.84};
\addlegendentry{BR pred \(\alpha=0.10\)}
\end{axis}
\end{tikzpicture}
\end{minipage}

\caption{Strong converse exponent and Bahadur--Rao prefactor validation (Experiment~2).
The left panel shows convergence \(\hat D(\kappa)\to D(\alpha\|\varepsilon)\) as in Theorem~\ref{thm:sc-exponent};
the right panel validates the \(\tfrac{1}{2}\log_2\kappa\) prefactor scaling in Theorem~\ref{thm:bahadur-rao}.
Legends are placed outside the plotting regions.}
\label{fig:exponent}
\end{figure}
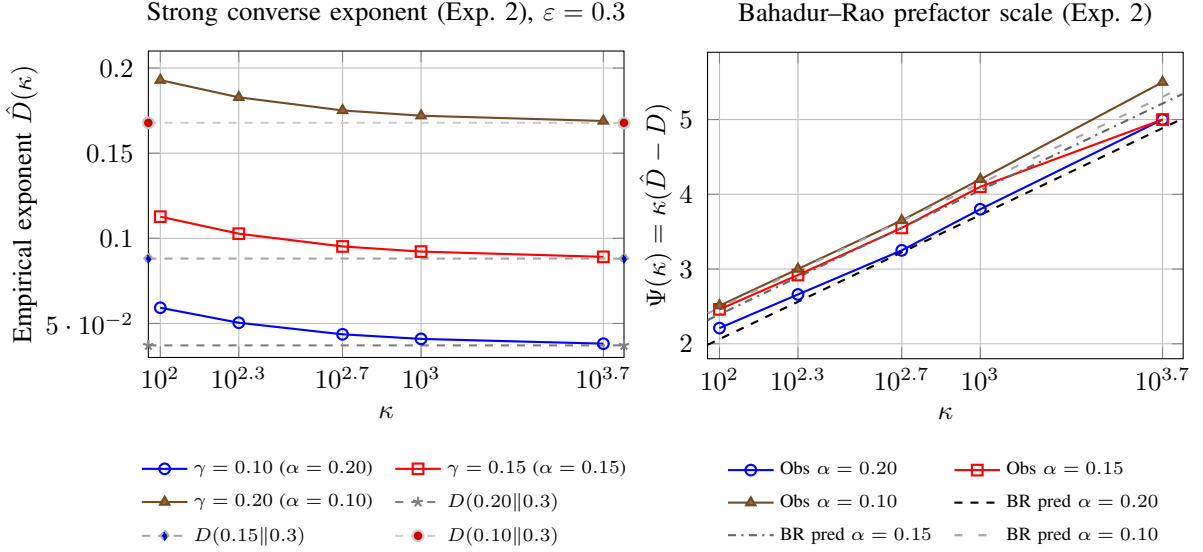

\subsection{Experiment 2: exponent and Bahadur--Rao prefactor}
\label{subsec:exp-exponent-fig}

Figure~\ref{fig:exponent} (left) plots the empirical exponent
\(\hat D(\kappa):=-\kappa^{-1}\log_2 P_c^{\mathrm{MDS}}\),
computed from the exact binomial CDF, against the KL divergence
\(D(\alpha\|\varepsilon)\) (horizontal dashed lines).
The convergence \(\hat D\to D\) from above is visible for all
three gap values; the relative error at \(\kappa=5000\) is
\(2.6\%\) (\(\gamma=0.10\)), \(1.2\%\) (\(\gamma=0.15\)), and
\(0.6\%\) (\(\gamma=0.20\)), decreasing with larger exponent as
the polynomial prefactor becomes a smaller fraction of the total.

Figure~\ref{fig:exponent} (right) shows the rescaled prefactor
quantity
\(\Psi(\kappa):=\kappa(\hat D-D)\).
The Bahadur--Rao expansion~\eqref{eq:bahadur-rao} predicts
\(\Psi_{\mathrm{BR}}(\kappa)
=\tfrac{1}{2}\log_2\kappa+c(\alpha)\)
with \(\alpha\)-dependent constants \(c(0.20)=-1.26\),
\(c(0.15)=-0.93\), \(c(0.10)=-0.84\).
The observed \(\Psi\) follows the predicted
\(\tfrac{1}{2}\log_2\kappa\) growth: the increase from
\(\kappa=100\) to \(\kappa=5000\) is within \(5\%\) of the
theoretical \(2.82\) for all three \(\alpha\) values, confirming
that the ``excess'' exponent at finite~\(\kappa\) is entirely
accounted for by the Bahadur--Rao prefactor.

\begin{figure}[ht]
\centering

\begin{minipage}{0.32\textwidth}
\centering
\begin{tikzpicture}
\begin{axis}[
  ylabel shift = -4pt,
  yticklabel style = {font=\tiny},
  width=\textwidth,
  height=0.84\textwidth,
  xmode=log,
  log basis x=10,
  xtick={50,100,200,500,1000,5000},
  xmin=45, xmax=6000,
  ymin=0.45, ymax=0.60,
  grid=both,
  xlabel={\(\kappa\)},
  ylabel={\(\Delta_c\)},
  title={Dispersion (Exp.~3), \(\varepsilon=0.2\)},
  xticklabel style = {font=\tiny},
  legend style={
    at={(0.5,-0.32)},
    anchor=north,
    draw=none,
    fill=none,
    legend columns=1,
    font=\scriptsize,
    row sep=2pt,
  },
  legend cell align=left,
]
\addplot+[mark=o, thick] coordinates {(50,0.566) (100,0.500) (200,0.495) (500,0.537) (1000,0.506) (5000,0.509)};
\addlegendentry{Obs \(\Delta_c=(r^*-\varepsilon\kappa)/\sqrt{\kappa}\)}
\addplot+[domain=45:6000, samples=2, thick, dashed, gray] {0.513};
\addlegendentry{Theory \(=\Phi^{-1}(0.9)\sqrt{0.16}\)}
\end{axis}
\end{tikzpicture}
\end{minipage}\hfill
\hspace{-8pt}
\begin{minipage}{0.32\textwidth}
\centering
\begin{tikzpicture}
\begin{axis}[
  ylabel shift = -4pt,
  yticklabel style = {xshift=-2pt},
  width=\textwidth,
  height=0.84\textwidth,
  xmin=0.19, xmax=0.41,
  ymode=log,
  ymin=5e-12, ymax=2,
  ytick={1,1e-2,1e-4,1e-6,1e-8,1e-10},
  grid=both,
  xlabel={\(\rho=\sigma/(\kappa\log m)\)},
  ylabel={\(P_e^{\mathrm{MDS}}\)},
  title={Phase transition (Exp.~3), \(\varepsilon=0.3\)},
  legend style={
    at={(0.5,-0.35)},
    anchor=north,
    draw=none,
    fill=none,
    legend columns=2,
    font=\scriptsize,
    row sep=2pt,
  },
  legend cell align=left,
]

\addplot+[mark=o, thick] coordinates {
  (0.20,0.984) (0.25,0.837) (0.28,0.623)
  (0.30,0.451) (0.32,0.289) (0.35,0.116) (0.40,0.012)
};
\addlegendentry{\(\kappa=100\)}

\addplot+[mark=square, thick] coordinates {
  (0.20,0.9999995) (0.25,0.992) (0.28,0.823)
  (0.30,0.478) (0.32,0.153) (0.35,0.007) (0.40,1e-5)
};
\addlegendentry{\(\kappa=500\)}

\addplot+[mark=triangle*, thick] coordinates {
  (0.20,0.99999999) (0.25,0.99999) (0.28,0.911)
  (0.30,0.484) (0.32,0.079) (0.35,1e-3) (0.40,1e-10)
};
\addlegendentry{\(\kappa=1000\)}

\addplot[thick, dotted, gray, forget plot]
  coordinates {(0.30,5e-12) (0.30,2)};
\end{axis}
\end{tikzpicture}
\end{minipage}\hfill
\hspace{8pt}
\begin{minipage}{0.32\textwidth}
\centering
\begin{tikzpicture}
\begin{axis}[
  ylabel shift = -4pt,
  yticklabel style = {xshift=-2pt},
  width=\textwidth,
  height=0.84\textwidth,
  xmode=log,
  log basis x=10,
  xtick={0.002,0.005,0.01,0.02},
  xmin=0.0018, xmax=0.022,
  ymin=0, ymax=30,
  grid=both,
  xlabel={\(\varepsilon\)},
  ylabel={\(d_{\max}^{\mathrm{chain}}/d_{\max}^{\mathrm{merge}}\)},
  title={Depth--resilience (Exp.~4), \(\delta=0.3\)},
  legend style={
    at={(0.5,-0.32)},
    anchor=north,
    draw=none,
    fill=none,
    legend columns=1,
    font=\scriptsize,
    row sep=2pt,
  },
  legend cell align=left,
]
\addplot+[mark=o, thick] coordinates {(0.02,4.00) (0.01,6.80) (0.005,11.67) (0.002,25.29)};
\addlegendentry{Exact ratio}
\addplot+[mark=none, thick, dashed, gray] coordinates {(0.02,4.16) (0.01,6.82) (0.005,11.55) (0.002,23.81)};
\addlegendentry{\(N^*/\log_2 N^*\) (Thm~\ref{thm:depth-resilience-duality})}
\end{axis}
\end{tikzpicture}
\end{minipage}

\caption{Experiments~3--4.
Left: coded second-order correction \(\Delta_c\) stabilizes at the
dispersion-predicted constant, while the derivation-constrained correction is
identically zero (Theorem~\ref{thm:unc-second-order}).
Middle: coded phase transition sharpens around \(\rho=\varepsilon=0.3\)
(vertical dotted line) as \(\kappa\) increases; the log-scale \(y\)-axis
reveals the exponential error decay in the super-capacity regime.
Right: architecture-dependent depth--resilience gap grows as
\(\varepsilon\to 0\).}
\label{fig:dispersion_phase_depth}
\end{figure}
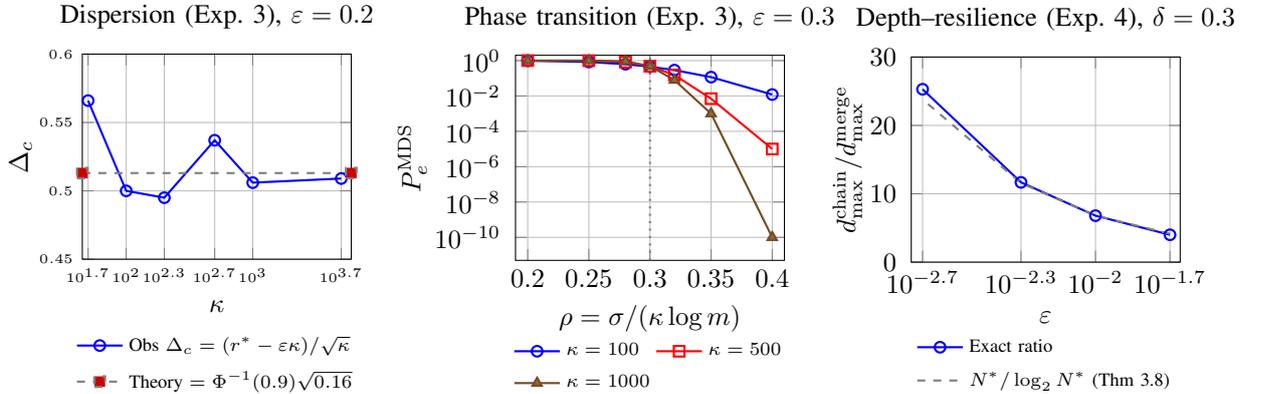

\subsection{Experiments 3 and 4: dispersion, phase transition, and depth--resilience}
\label{subsec:exp-dispersion-depth-fig}

Figure~\ref{fig:dispersion_phase_depth} (left) validates the
dispersion dichotomy
(Theorem~\ref{thm:dispersion-dichotomy}): the normalized coded
correction
\(\Delta_c:=(r^*-\varepsilon\kappa)/\sqrt\kappa\) oscillates
around the predicted value
\(\Phi^{-1}(0.9)\sqrt{0.16}=0.513\)
(small fluctuations arise from the integer constraint on~\(r^*\)),
while the derivation-constrained correction is identically zero
for all~\(\kappa\), reflecting \(V_{\mathrm{unc}}=0\)
(Theorem~\ref{thm:unc-second-order}).

Figure~\ref{fig:dispersion_phase_depth} (middle) illustrates the
coded phase transition on a logarithmic error-probability axis.
The window in which \(P_e\) drops from \(0.9\) to \(0.1\) narrows
from \(\Delta\rho\approx 0.15\) at \(\kappa=100\) to
\(\Delta\rho\approx 0.04\) at \(\kappa=1000\), and the exponential
error decay in the super-capacity regime (\(\rho>\varepsilon\))
becomes steeper with increasing~\(\kappa\), consistent with
Theorem~\ref{thm:phase-diagram} regimes C1--C5.

Figure~\ref{fig:dispersion_phase_depth} (right) confirms the
depth--resilience duality
(Theorem~\ref{thm:depth-resilience-duality}): as
\(\varepsilon\to 0\), the chain-to-merge maximum-depth ratio grows
as \(\Theta(N^*/\log N^*)\), matching the asymptotic prediction to
within~\(6\%\).
At \(\varepsilon=0.002\) the chain derives \(177\)~levels deep
while the merge is limited to~\(7\)---a \(25\times\) gap.

\subsection{Experiment 6: Monte Carlo validation (kept in Section VIII)}
\label{subsec:exp-mc}

To provide an independent check that does not rely on closed-form evaluation,
we perform Monte~Carlo simulation of the premise-erasure process under two parameter regimes.
In each trial, an erasure count \(E\sim\mathrm{Bin}(\kappa,\varepsilon)\) is drawn; the coded
scheme succeeds iff \(E\le r\), while the derivation-constrained scheme succeeds iff all \(N\)
exposed dependencies survive, i.e., \(E_{\mathrm{exp}}=0\) with \(E_{\mathrm{exp}}\sim\mathrm{Bin}(N,\varepsilon)\).
Table~\ref{tab:mc} reports \(N_{\mathrm{mc}}=10^6\) trials with \(95\%\) Clopper--Pearson intervals.

\begin{table}[ht]
\centering
\caption{Monte Carlo validation (\(N_{\mathrm{mc}}=10^6\)).
Panel~A: \(\varepsilon=0.2\), \(\kappa=500\). Panel~B: \(\varepsilon=0.3\), \(\kappa=100\).
``CI'' denotes the \(95\%\) Clopper--Pearson interval.}
\label{tab:mc}
\smallskip

\begin{minipage}[t]{0.48\textwidth}
\centering
\setlength{\tabcolsep}{2.5pt} 
\begin{tabular}{l|c|ccc}
\toprule
\multicolumn{5}{c}{\textbf{Panel~A}\;(\(\varepsilon=0.2\), \(\kappa=500\))} \\
\midrule
Scheme & Parameter & \(P_e\) (exact) & \(\hat P_e\) (MC) & \(95\%\) CI \\
\midrule
Coded  & \(r=90\)   & 0.856 & 0.857 & [0.856, 0.858] \\
Coded  & \(r=100\)  & 0.473 & 0.473 & [0.472, 0.474] \\
Coded  & \(r=111\)  & 0.100 & 0.100 & [0.100, 0.101] \\
Coded  & \(r=120\)  & 0.012 & 0.012 & [0.012, 0.013] \\
\midrule
Unc    & \(N=10\)   & 0.893 & 0.892 & [0.892, 0.893] \\
Unc    & \(N=3\)    & 0.488 & 0.488 & [0.487, 0.489] \\
Unc    & \(N=1\)    & 0.200 & 0.200 & [0.199, 0.201] \\
\bottomrule
\end{tabular}
\end{minipage}
\hfill
\begin{minipage}[t]{0.48\textwidth}
\centering
\setlength{\tabcolsep}{2.5pt} 
\begin{tabular}{l|c|ccc}
\toprule
\multicolumn{5}{c}{\textbf{Panel~B}\;(\(\varepsilon=0.3\), \(\kappa=100\))} \\
\midrule
Scheme & Parameter & \(P_e\) (exact) & \(\hat P_e\) (MC) & \(95\%\) CI \\
\midrule
Coded  & \(r=20\)   & 0.984 & 0.984 & [0.983, 0.984] \\
Coded  & \(r=25\)   & 0.837 & 0.837 & [0.837, 0.838] \\
Coded  & \(r=30\)   & 0.451 & 0.451 & [0.450, 0.452] \\
Coded  & \(r=35\)   & 0.116 & 0.116 & [0.116, 0.117] \\
\midrule
Unc    & \(N=5\)    & 0.832 & 0.832 & [0.831, 0.832] \\
Unc    & \(N=3\)    & 0.657 & 0.657 & [0.656, 0.658] \\
Unc    & \(N=1\)    & 0.300 & 0.300 & [0.299, 0.301] \\
\bottomrule
\end{tabular}
\end{minipage}
\end{table}

Across all scenarios, the Monte~Carlo estimates match the exact values within sampling uncertainty,
supporting the finite-length formulas for both coded and derivation-constrained decoding
(Theorems~\ref{thm:exact-Pc} and~\ref{thm:unc-error-complete}).

\subsection{Summary of numerical findings}
\label{subsec:exp-summary}

Figures~\ref{fig:penalty}--\ref{fig:dispersion_phase_depth} and Table~\ref{tab:mc} jointly corroborate the paper's
key predictions:
(i)~the derivation penalty converges to \(1/\varepsilon\) and obeys the second-order refinement
(Theorems~\ref{thm:refined-penalty} and~\ref{thm:multi-penalty-mq});
(ii)~below the coded threshold, the success probability decays at KL-divergence rate with the
Bahadur--Rao prefactor (Theorems~\ref{thm:sc-exponent} and~\ref{thm:bahadur-rao});
(iii)~coded caching exhibits nonzero dispersion and a sharpening phase transition, whereas the
derivation-constrained scheme has effectively zero dispersion (Theorem~\ref{thm:dispersion-dichotomy});
and (iv)~the architecture-dependent depth--resilience gap grows as \(\varepsilon\to 0\)
(Theorem~\ref{thm:depth-resilience-duality}).

\section{Discussion}
\label{sec:discussion}

\paragraph{The derivation penalty as a fundamental limit.}
The central message of this paper is that requiring a decoder to
produce a \emph{logical derivation}---rather than merely output
the correct answer---inflates the minimum reliable cache by a
factor of $1/\varepsilon$ relative to coded caching
(Theorem~\ref{thm:source-channel-separation}).
This penalty is \emph{universal}: it depends only on the channel
erasure rate~$\varepsilon$ and holds regardless of the number of
queries, their overlap structure, or the target
reliability~$\delta$
(Theorem~\ref{thm:multi-penalty-mq}).
Its proof-theoretic origin is the structural caching rigidity
of Theorem~\ref{thm:ancestral-relevance}: under faithful
derivation with unique traces, only cache facts lying within
the derivation DAG of the target query contribute to resilience,
preventing the cross-coordinate error correction that MDS codes
exploit.

\paragraph{Relation to coded caching and channel coding.}
The coded-caching literature initiated by Maddah-Ali and
Niesen~\cite{maddah2014fundamental} studies cache-aided
delivery over broadcast channels, where coded multicast
messages yield multiplicative gains over uncoded placement.
Our setting is structurally different: the ``receiver'' is a
proof engine that may only combine cached logical facts with
surviving premises via valid inference steps, not an arbitrary
algebraic decoder.
The MDS codes used in our achievability results
(Theorem~\ref{thm:coded-achievability}) are the same as those
in network coded caching, but the converse mechanism differs
entirely---it rests on the DAG constraint rather than on
cut-set bounds or network topology.
On the channel-coding side, the strong converse exponent
$D(\rho\|\varepsilon)$ (Theorem~\ref{thm:sc-exponent}) and
the reliability function (Corollary~\ref{cor:reliability})
coincide with the classical $m$-ary BEC exponents, since MDS
codes achieve the BEC sphere-packing bound at all
rates~\cite[Theorem~5.8.3]{gallager1968information}.
The cache dispersion
$V_{\mathrm{cache}}=\varepsilon(1{-}\varepsilon)(\log m)^2$
likewise matches the channel dispersion of the $m$-ary BEC
with uniform input~\cite{polyanskiy2010channel}, confirming
that the coded caching problem inherits the channel's
second-order behavior exactly.

\paragraph{The dispersion dichotomy and its operational
  consequences.}
The zero effective dispersion of the derivation-constrained
scheme (Theorem~\ref{thm:dispersion-dichotomy}) is arguably
the most striking structural finding.
It implies that finite-length penalties relative to
$\sqrt\kappa$ vanish: the second-order term is
$-N^*\log m=O(\log m)$, independent of~$\kappa$.
All known capacity-achieving coded schemes for memoryless
channels exhibit positive dispersion~\cite{polyanskiy2010channel},
so the vanishing dispersion is a distinctive signature of the
derivation constraint.
The operational origin is the distinction between a
\emph{sum} statistic (coded scheme: total erasure count
$E=\sum E_j$, governed by the CLT) and a \emph{conjunction}
(derivation-constrained: all $N$ unprotected facts must
survive, governed by a product probability with no
CLT-type concentration).
This dichotomy persists in the multi-query setting
(Theorem~\ref{thm:multi-dispersion-mq}), where the joint
coded dispersion scales as
$n_{\mathrm{eff}}\cdot V_{\mathrm{cache}}$ while the
derivation-constrained dispersion remains identically zero.

\paragraph{Architecture dependence.}
The depth-space analysis of Section~\ref{sec:arch-depth}
demonstrates that the architecture-dependent capacity mapping
$d\mapsto\kappa_{\mathcal A}(d)$ from CT2
(Theorem~\ref{thm:exponential-capacity}) propagates to every
downstream quantity: critical depth, transition width, error
exponent per depth unit, and maximum resilient depth.
The merge architecture's exponential growth
$\kappa_{\mathrm{merge}}(d)=k\cdot 2^{d-1}$ yields
exponentially sharper phase transitions
(Theorem~\ref{thm:arch-phase-ds}) and exponentially smaller
maximum resilient depth
(Theorem~\ref{thm:depth-resilience-duality}) compared to the
chain, establishing a precise information-theoretic cost
of parallelism in derivation.
More broadly, this demonstrates that the \emph{branching factor}
of the derivation program---not merely its depth---is the
controlling parameter for noise resilience.

\paragraph{Limitations.}
Several modeling assumptions constrain the scope of the present
results.
First, the i.i.d.\ erasure model is central to the channel
interpretation: both the image-size bound
(Lemma~\ref{lem:image-size}) and the Bahadur--Rao asymptotics
(Theorem~\ref{thm:bahadur-rao}) exploit independence.
Under correlated erasure---e.g., when related facts share a
physical storage subsystem---the effective channel acquires
memory, and the binomial erasure count is replaced by a more
complex statistic.
Second, the distinct-coordinate regime
$\kappa(q,B)=a(q)$ underpins the clean
$\mathrm{BEC}_m(\varepsilon)^{\otimes\kappa}$ factorization.
Lemma~\ref{lem:distinct-coordinates} ensures this holds for a
$1{-}o(1)$ fraction when $a=o(\sqrt m)$, but for very deep
queries the arity may exceed this threshold, requiring a
modified channel model.
Third, the cost model
(Definition~\ref{def:cost-model}) treats derivation steps as
unit-cost, ignoring potential variation in rule-application
complexity or cache-lookup overhead; a finer-grained model
could alter the critical frequency of CT3.

\paragraph{Open problems.}
We highlight five directions.

\emph{(i)~Adaptive caching.}
The present framework fixes the cache $S$ before observing the
erasure pattern.
An adaptive protocol that iteratively updates the cache after
partial observation of~$\tilde B$ could potentially reduce the
derivation penalty.
Whether the $1/\varepsilon$ factor persists under adaptivity
is open; the structural rigidity theorem
(Theorem~\ref{thm:ancestral-relevance}) remains valid per
round, but multi-round interaction may circumvent it.

\emph{(ii)~Correlated and adversarial erasure.}
Extending the theory to Markov or worst-case erasure models
would require channel-coding techniques beyond the i.i.d.\
framework~\cite{gallager1968information}.
The derivation penalty may depend on the correlation structure,
and new converse techniques are needed since the image-size
bound relies on product distributions.

\emph{(iii)~Approximate derivation.}
If the decoder may output a formula $\hat q$ that is
$\epsilon$-close to~$q$ in a suitable metric (e.g., symmetric
difference of models), the derivation penalty may be reduced.
This connects to approximate query answering in database
theory~\cite{abiteboul1995foundations} and lossy
source coding~\cite{cover2006elements}.

\emph{(iv)~General branching architectures.}
Theorem~\ref{thm:depth-resilience-duality}(III) shows that
$b$-ary merges yield
$d_{\max}(b)=1+\log_b(N^*/k)$ with resilient capacity
$N^*\log m$.
Whether architectures exist that simultaneously achieve
sub-logarithmic depth growth and linear resilient depth remains
open; such a result would require departing from the
fixed-branching paradigm.

\emph{(v)~Threshold behavior under enriched proof systems.}
The derivation-constrained scheme exhibits no phase transition
(Remark~\ref{rem:no-transition}).
If the proof engine is augmented with probabilistic reasoning
or sampling-based search, the resulting ``stochastic proof
system'' could conceivably introduce collective statistical
effects and threshold behavior.
Characterizing whether enriched proof systems can close the
$1/\varepsilon$ gap---or whether the penalty is intrinsic to
any deductive mechanism---is a foundational open question.

\section{Conclusion}
\label{sec:conclusion}

This paper has developed an information-theoretic framework for
premise-erasure caching in derivation-based reasoning engines,
determining the fundamental limits of reliable query recovery when
the premise base is subject to i.i.d.\ stochastic loss.

The framework is built on four coding theorems.
CT1 establishes that each derivation step in a faithful Datalog
program carries exactly $\log m$ bits of conditional algorithmic
information.
CT2 reveals an exponential capacity separation between the
tuple-assembly chain ($\kappa=k{+}d{-}1$) and the balanced-merge
architecture ($\kappa=k\cdot 2^{d-1}$) at equal derivation
depth~$d$.
CT3 identifies a critical access frequency
$f_c=\Theta(\rho_s\cdot\log(m{+}d))$ separating the regimes
where caching and on-demand derivation are respectively optimal.
CT4 proves that the minimum derivation-constrained cache under
i.i.d.\ erasure at rate~$\varepsilon$ with resilience
target~$\delta$ is
$\sigma^*_{\mathrm{unc}}=(\kappa{-}N^*)\log m+O(\kappa)$ with
$N^*\approx\delta/\varepsilon$, decomposing the conditional
information of a query into reliable cache plus noisy channel
capacity.

The central result is the \emph{derivation penalty}: the ratio
$\sigma^*_{\mathrm{unc}}/\sigma^*_{\mathrm{code}}$ converges to
$1/\varepsilon$, universally across query counts, overlap
structures, and reliability targets.
This penalty originates from the structural caching rigidity
theorem: under faithful derivation with unique traces, only cache
facts within the derivation DAG of the target query contribute to
resilience, precluding the cross-coordinate error correction that
MDS codes exploit over the $m$-ary erasure channel.

Beyond first-order capacity, the paper provides a complete
error-probability characterization.
Below the coded capacity threshold, the success probability decays
at the KL-divergence rate $D(\rho\|\varepsilon)$ with exact
Bahadur--Rao prefactors; above the threshold, MDS codes achieve
the sphere-packing exponent at all rates.
The two decoder models exhibit a \emph{dispersion dichotomy}: the
coded scheme inherits the BEC channel dispersion
$V_{\mathrm{cache}}=\varepsilon(1{-}\varepsilon)(\log m)^2$,
producing a $\Theta(\sqrt\kappa\,\log m)$ second-order term,
while the derivation-constrained scheme has identically zero
effective dispersion---a structural signature of the derivation
constraint arising from the distinction between a sum statistic
(total erasure count, governed by the CLT) and a conjunction (all
unprotected dependencies must survive, governed by a product
probability).
The complete phase diagram comprises five coded regimes and three
derivation-constrained regimes, the latter exhibiting no phase
transition.
In the multi-query setting, joint coding exploits overlap
deduplication and statistical pooling, yet the $1/\varepsilon$
penalty persists as a universal constant depending only on the
channel parameter.

The architecture-dependent mapping
$d\mapsto\kappa_{\mathcal A}(d)$ from CT2 propagates through the
entire error-probability landscape: the merge architecture's
exponential dependency growth yields exponentially sharper phase
transitions in depth space, an exponentially smaller maximum
resilient depth, and exponentially faster per-depth-unit error
decay compared to the chain, establishing a precise
information-theoretic cost of parallelism in derivation.
All four coding theorems and their downstream
consequences---including the penalty, dispersion, and
phase-diagram results---transfer across synonymous logical
representations (Theorem~\ref{thm:CT-transfer}), ensuring that
these quantities are canonical invariants of the underlying
inference system rather than artifacts of a particular encoding.

The present results rest on several structural
assumptions---i.i.d.\ premise erasure, faithful derivation with
unique traces, and the asymptotic regime of large base size~$m$%
---and the quantitative conclusions may not directly extend to
settings with correlated or adversarial noise, non-faithful or
non-deterministic inference systems, derivation architectures
beyond the two Datalog programs studied here, or smaller-scale
instances where finite-length deviations are non-negligible.
Whether the $1/\varepsilon$ derivation penalty persists under
adaptive caching protocols, approximate derivation, or enriched
proof systems with probabilistic reasoning remains open.
These directions connect the present framework to broader
questions about the fundamental information-theoretic price of
requiring logical structure---rather than mere correctness---in
the output of a noisy computation.

\appendices

\counterwithin{definition}{section}
\counterwithin{axiom}{section}
\counterwithin{assumption}{section}
\counterwithin{theorem}{section}
\counterwithin{lemma}{section}
\counterwithin{proposition}{section}
\counterwithin{corollary}{section}
\counterwithin{remark}{section}
\counterwithin{example}{section}

\section{Logical Substrate, Information Axioms, and Noise Model}
\label{sec:substrate}

Throughout we fix a formal logical system \(\mathcal{L}\) taken to be
\emph{first-order logic with least fixed-point operators} \(\mathrm{FO(LFP)}\)
\cite{immerman1999descriptive}, extended with multiple sorts including at least
\(\mathsf{Obj}\) (entities), \(\mathsf{Time}\) (time points), and \(\mathsf{Carrier}\) (carriers).
We restrict attention to \emph{finite} structures over a discrete, bounded time domain
\(T=\{t_0,\ldots,t_n\}\).
Under this restriction, \(\mathcal{L}\) is consistent, its satisfaction relation
\(\mathcal{M}\models\varphi\) is decidable for every finite structure \(\mathcal{M}\) and formula
\(\varphi\in\mathcal{L}\), and---by the Immerman--Vardi theorem
\cite{immerman1982relational,vardi1982complexity}---\(\mathrm{FO(LFP)}\) captures exactly the class of
polynomial-time computable properties over ordered finite structures.

\begin{definition}[Expressible and effectively representable state sets]
\label{def:expressible-effective}
A state set \(S(X,T)\) is \emph{expressible in \(\mathcal{L}\)} if there exists a finite
\(\mathcal L\)-signature \(\Sigma\) and a finite \(\mathcal L\)-structure whose domain includes (encodings of)
\(X\) and \(T\), such that membership \(s\in S(X,T)\) and all relations used to describe the state dynamics
are definable by formulas of \(\mathcal L\) over \(\Sigma\).
It is \emph{effectively representable} if each \(s\in S(X,T)\) admits a finite encoding from which the
relevant \(\mathcal{L}\)-predicates can be evaluated effectively.
\cite{qiu2025research}.
\end{definition}

\begin{assumption}[Fixed effective proof system and deductive closure]
\label{assump:proof-system}
We fix an effective proof system \(\mathsf{PS}\) for \(\mathcal L\) such that proof checking is decidable.
We write \(\Gamma \vdash_{\mathcal L} \varphi\) to denote derivability of a well-formed formula \(\varphi\)
from a \emph{finite} set of formulas \(\Gamma\) in \(\mathsf{PS}\), and define the deductive closure operator
\(
\Cn(\Gamma)\;:=\;\{\varphi:\ \Gamma\vdash_{\mathcal L}\varphi\}.
\)
\end{assumption}

\begin{axiom}[Binary attribute (two-domain information)]
\label{ax:binary-attribute}
Information is modeled by two state domains \((S_O,S_C)\), where \(S_O\) is the semantic/ontological
state set and \(S_C\) is the carrier/physical state set.
\end{axiom}

\begin{axiom}[Existence duration (domain-wise time and precedence)]
\label{ax:existence-duration}
There exist time domains \((T_O,\prec_O)\) and \((T_C,\prec_C)\) and time-index maps
\(\timeindex_O:S_O\to T_O\) and \(\timeindex_C:S_C\to T_C\).
Moreover, there is a fixed cross-domain precedence relation \(\prec\) such that whenever a carrier
state \(s_c\) is an enabled realization of a semantic state \(s_o\), we have
\[
  \timeindex_O(s_o) \prec \timeindex_C(s_c).
\]
\end{axiom}

\begin{axiom}[State representation (effective encodability)]
\label{ax:state-representation}
There exist finite binary representations \(\enc_O:S_O\to\{0,1\}^*\) and \(\enc_C:S_C\to\{0,1\}^*\) such that:
\begin{enumerate}[label=\textup{(R\arabic*)}]
  \item \emph{Decidable identity:} equality of encoded states is decidable.
  \item \emph{Effective predicate evaluation:} all relations/predicates used to describe membership and
  structure on \(S_O\) and \(S_C\) in the fixed logic \(\mathcal L\) are effectively evaluable from the encodings.
\end{enumerate}
\end{axiom}

\begin{axiom}[Enabling mapping (computable totality and surjective coverage)]
\label{ax:enabling-mapping}
There exists a relation \(R_{\mathcal E}\subseteq S_O\times S_C\) and the induced set-valued map
\(\mathcal E:S_O\Rightarrow S_C\), \(\mathcal E(s_o):=\{s_c\in S_C:(s_o,s_c)\in R_{\mathcal E}\}\), such that:
\begin{enumerate}[label=\textup{(E\arabic*)}]
  \item \emph{Totality on semantic states:} for every \(s_o\in S_O\), \(\mathcal E(s_o)\neq\varnothing\).
  \item \emph{Surjective coverage of carrier reality:} for every \(s_c\in S_C\), there exists \(s_o\in S_O\)
  with \(s_c\in \mathcal E(s_o)\). Equivalently, \(\bigcup_{s_o\in S_O}\mathcal E(s_o)=S_C\).
  \item \emph{Computable enabling selector:} there exists a partial computable function
  \(e:\{0,1\}^*\to\{0,1\}^*\) such that for every \(s_o\in S_O\), \(e(\enc_O(s_o))\) halts and equals
  \(\enc_C(s_c)\) for some \(s_c\in\mathcal E(s_o)\).
  \item \emph{Temporal precedence of realization:} if \(s_c\in \mathcal E(s_o)\), then
  \(\timeindex_O(s_o) \prec \timeindex_C(s_c)\).
\end{enumerate}
\end{axiom}

By Axiom~\ref{ax:enabling-mapping}, there exists a relation
\(R_{\mathcal E}\subseteq S_O\times S_C\).
We write the induced set-valued enabling (realization) map
\begin{equation}\label{eq:enabling-map}
  \mathcal E:S_O\Rightarrow S_C,
  \;
  \mathcal E(s_o)
  :=
  \{\,s_c\in S_C:\ (s_o,s_c)\in R_{\mathcal E}\,\}.
\end{equation}

\begin{definition}[Computable enabling (realization) mechanism]
\label{def:enabling}
A relation \(R_{\mathcal E}\subseteq S_O\times S_C\) (equivalently, the induced \(\mathcal E\) in
\eqref{eq:enabling-map}) is called a \emph{computable enabling mechanism} if it satisfies all clauses
\textup{(E1)}--\textup{(E4)} of Axiom~\ref{ax:enabling-mapping}.
\end{definition}

\begin{definition}[Information instance]
\label{def:info-instance}
An \emph{information instance} is a tuple
\[
  \mathcal I
  \;=\;
  \langle
    S_O, T_O, S_C, T_C,\timeindex_O,\timeindex_C, R_{\mathcal E}
  \rangle,
\]
where \((S_O,S_C)\) are the two state domains (Axiom~\ref{ax:binary-attribute}),
\(\timeindex_O,\timeindex_C\) are the time-index maps (Axiom~\ref{ax:existence-duration}),
and \(R_{\mathcal E}\) is a computable enabling mechanism (Definition~\ref{def:enabling}).
\end{definition}

\begin{definition}[Compositional interpretation between sub-signatures]
\label{def:comp-interp}
Let \(\Sigma_1,\Sigma_2\) be sub-signatures of \(\mathcal{L}\).
A \emph{compositional interpretation} \(\tau:\Sigma_1\hookrightarrow\Sigma_2\) in \(\mathcal{L}\) consists of:
\textup{(i)} for each sort \(s\) of \(\Sigma_1\), a domain formula \(\delta_s(y)\) over \(\Sigma_2\);
\textup{(ii)} for each relation symbol \(R\) of \(\Sigma_1\), a formula \(\tau_R(y_1,\ldots,y_n)\) over \(\Sigma_2\).
The map \(\tau\) extends to all formulas of \(\Sigma_1\) by the standard recursive clauses.
\end{definition}

\begin{definition}[Synonymous state sets]
\label{def:synonymous-states}
State sets \(S_1(X_1,T_1)\) and \(S_2(X_2,T_2)\), each expressible in \(\mathcal L\) as an interpretation of a
WFF family \(\Phi_1\) and \(\Phi_2\) (assumed finite or effectively enumerable) over sub-signatures
\(\Sigma_1\) and \(\Sigma_2\), are \emph{synonymous relative to \(\mathcal{L}\)}, written
\(S_1\equiv_\mathcal{L}S_2\), if there exist:
\begin{enumerate}[label=\textup{(\roman*)}]
  \item compositional interpretations \(\tau_{12}:\Sigma_1\hookrightarrow\Sigma_2\) and
  \(\tau_{21}:\Sigma_2\hookrightarrow\Sigma_1\) in \(\mathcal{L}\), and
  \item a bijection \(\sigma:\Phi_1\to\Phi_2\),
\end{enumerate}
such that for every \(\varphi\in\Phi_1\),
\(
  \vdash_{\mathcal{L}}\;\sigma(\varphi)\leftrightarrow\tau_{12}(\varphi),
\)
and round-trip coherence holds on generators (bi-interpretability).
\end{definition}

\begin{definition}[Ideal information]
\label{def:ideal-info}
An information instance \(\mathcal I\) is \emph{ideal} (with respect to \(\mathcal L\)) if
\[
  S_O \equiv_{\mathcal L} S_C.
\]
\end{definition}

\begin{assumption}[Common semantic universe for set operations]
\label{assump:semantic-universe}
There exists a fixed ambient set \(\mathbb{S}_O\) such that all semantic state sets considered satisfy
\(S_O\subseteq \mathbb{S}_O\).
Moreover, \(\mathbb{S}_O\) is effectively representable in the sense that \(\enc_O\) and \(\timeindex_O\) extend to
\(\mathbb S_O\), and \(\mathbb S_O\) is closed under the syntactic renamings/interpretations used to form synonymous
representatives in Section~\ref{sec:substrate} (in particular, the representatives produced in
Lemma~\ref{lem:sem-trans} can be chosen as subsets of \(\mathbb S_O\)).
\end{assumption}

\begin{definition}[Noisy semantic base (set perturbation)]
\label{def:noisy}
Let \(S_O\subseteq \mathbb S_O\) be the intended semantic state set.
A \emph{noisy semantic base} associated with \(S_O\) is any set of the form
\[
  \tilde S_O := (S_O \setminus S_O^-) \cup S_O^+,
\]
where \(S_O^- \subseteq S_O\) (lost states) and \(S_O^+ \subseteq \mathbb S_O\setminus S_O\) (spurious states).
\end{definition}

\begin{lemma}[Semantic transferability under signature isomorphism]
\label{lem:sem-trans}
Let \(\mathcal{L}=\mathrm{FO(LFP)}\) be the underlying logical system over finite relational structures.
For any state set \(S(X,T)\) expressible in \(\mathcal{L}\) over a sub-signature \(\Sigma=\{R_1,\ldots,R_n\}\),
and any arity-matching \(\Sigma'=\{R'_1,\ldots,R'_n\}\), there exists \(S'(Y,T')\) expressible over \(\Sigma'\)
such that \(S\equiv_{\mathcal{L}}S'\).
\end{lemma}

\begin{proof}
Uniformly rename symbols \(R_j\mapsto R'_j\) and transport the generating formula family; this yields a
bi-interpretability witness by syntactic isomorphism.
\end{proof}

\begin{definition}[Noisy information (semantic description relative to a fixed carrier)]
\label{def:noisy-info}
Fix a carrier domain \(S_C\).
A \emph{noisy information} associated with \(S_C\) is an abstract object
\[
  \tilde{\mathcal I}=\langle \tilde S_O,\, S_C,\, (\tau_{OC},\tau_{CO},\sigma)\rangle,
\]
where \((\tau_{OC},\tau_{CO},\sigma)\) witnesses \(\tilde S_O \equiv_{\mathcal L} S_C\).
\end{definition}

\begin{proposition}[Existence of carrier-fixed noisy information]
\label{prop:noisy-exist}
Let \(\mathcal I\) be an information instance with semantic domain \(S_O\) and carrier domain \(S_C\).
Then there exists a noisy semantic base \(\tilde S_O\) of the form
\(\tilde S_O=(S_O\setminus S_O^-)\cup S_O^+\) such that \(\tilde S_O \equiv_{\mathcal L} S_C\),
and hence a noisy information object \(\tilde{\mathcal I}\) with semantic component \(\tilde S_O\).
\end{proposition}

\begin{proof}
Apply Lemma~\ref{lem:sem-trans} to obtain some semantic representative \(S'_O\equiv_\mathcal L S_C\),
then set \(S_O^-:=S_O \setminus S'_O\) and \(S_O^+:=S'_O\setminus S_O\).
\end{proof}

\section{Strong Converse Exponent: Detailed Calculation}
\label{app:sc-exponent-calc}

The following calculation is used in
Theorem~\ref{thm:sc-exponent}.
We verify that the minimum of
$g_0(u):=u(\varepsilon{-}\gamma)
+\log(1{-}\varepsilon{+}\varepsilon e^{-u})$
over $u\ge 0$ satisfies $-g_0^*=D(\varepsilon'\|\varepsilon)$
with $\varepsilon':=\varepsilon{-}\gamma$.

Setting $\alpha:=\varepsilon'(1{-}\varepsilon)/
[\varepsilon(1{-}\varepsilon')]$ and $u^*=-\log\alpha$:
\begin{align*}
  -g^*
  &=-\varepsilon'\log\alpha
    -\log(1{-}\varepsilon{+}\varepsilon\alpha).
\end{align*}
The argument of the second logarithm simplifies:
\begin{align*}
  1{-}\varepsilon{+}\varepsilon\alpha
  &=(1{-}\varepsilon)
    +\frac{\varepsilon'(1{-}\varepsilon)}
          {1{-}\varepsilon'}
  =\frac{1{-}\varepsilon}{1{-}\varepsilon'}\,.
\end{align*}
Therefore:
\begin{align*}
  -g^*
  &=-\varepsilon'\log
    \frac{\varepsilon'(1{-}\varepsilon)}
         {\varepsilon(1{-}\varepsilon')}
    -\log\frac{1{-}\varepsilon}{1{-}\varepsilon'}\\
  &=-\varepsilon'\log\frac{\varepsilon'}{\varepsilon}
    -\varepsilon'\log\frac{1{-}\varepsilon}{1{-}\varepsilon'}
    +\log\frac{1{-}\varepsilon'}{1{-}\varepsilon}\\
  &=\varepsilon'\log\frac{\varepsilon'}{\varepsilon}
    +(1{-}\varepsilon')\log\frac{1{-}\varepsilon'}
                              {1{-}\varepsilon}
  =D(\varepsilon'\|\varepsilon).
\end{align*}


\section*{Acknowledgment}

During the writing and revision of this paper, I received many insightful comments from Associate Professor Rui Wang of the School of Computer Science at Shanghai Jiao Tong University and also gained much inspiration and assistance from regular academic discussions with doctoral students Yiming Wang, Chun Li, Hu Xu, Siyuan Qiu, Zeyan Li, Jiashuo Zhang, Junxuan He, and Xiao Wang. I hereby express my sincere gratitude to them.

\bibliographystyle{IEEEtran}
\bibliography{ref}

\vfill
\end{document}